\documentclass[a4paper, 11pt]{article}

\usepackage{fullpage}
\usepackage[english]{babel}
\usepackage[latin1]{inputenc}
\usepackage{mathrsfs}
\usepackage{amsfonts}
\usepackage{amssymb}
\usepackage{braket}
\usepackage{amsmath}
\usepackage{amsthm}
\usepackage{physics}
\usepackage{graphicx}
\usepackage[usenames,dvipsnames]{color}
\usepackage[colorlinks=true,citecolor=blue,linkcolor=magenta,bookmarks=false]{hyperref}
\usepackage{lmodern}
\usepackage{microtype}
\usepackage{braket}
\usepackage{tabularx,colortbl}
\usepackage{pifont}
\usepackage{mathtools}
\usepackage{dsfont}
\usepackage{algpascal}
\usepackage{algorithm}
\usepackage{algpseudocode}
\usepackage{chngcntr}
\usepackage{complexity}
\usepackage[affil-it]{authblk}
\usepackage[table, x11names]{xcolor}
\usepackage{booktabs}
\usepackage{multirow}
\usepackage{tabulary}
\usepackage{array,makecell,multirow,ragged2e}
\usepackage[font=small,labelfont=bf]{caption}
\usepackage{csquotes}
\input{Qcircuit}

\usepackage{todonotes}
\usepackage{soul}
\usepackage{hyperref}

\usepackage[section]{placeins}

\renewcommand{\thefootnote}{\fnsymbol{footnote}}

\renewcommand{\arraystretch}{1.5}
\newtheorem{theorem}{Theorem}
\newtheorem{cla}{Claim}

\newtheorem{definition}{Definition}
\newtheorem{lemma}{Lemma}
\newtheorem{proposition}{Proposition}

\newcommand{\comment}[1]{}
\newcommand{\idty}[1]{\mathbb{1}}
\newcommand{\ovsqrt}[1]{\frac{1}{\sqrt{2}}}

\newcommand{\Ord}[1]{\mathcal{O}\left(#1\right)}

\renewcommand{\poly}[1]{\ensuremath\operatorname{poly}\left( #1 \right)}

\usepackage[backend=bibtex, natbib=true, bibstyle=alphabetic,citestyle=alphabetic, doi=true, sorting=nyt, firstinits]{biblatex}
\title{Quantum linear systems algorithms: a primer}

\begin{document}

\author[1]{Danial Dervovic}
\author[1]{Mark Herbster}
\author[1,2]{Peter Mountney}
\author[1]{Simone Severini}
\author[1]{Na\"iri Usher}
\author[1]{Leonard Wossnig}

\affil[1]{Department of Computer Science, University College London, London, UK}

\affil[2]{Siemens Healthineers, Medical Imaging Technologies, Princeton, NJ, USA}

\date{}
\maketitle
\renewcommand{\thefootnote}{\arabic{footnote}}

\begin{abstract}
The Harrow-Hassidim-Lloyd (HHL) quantum algorithm for sampling from the solution of a linear system provides an exponential speed-up over its classical counterpart. The problem of solving a system of linear equations has a wide scope of applications, and thus HHL constitutes an important algorithmic primitive. In these notes, we present the HHL algorithm and its improved versions in detail, including explanations of the constituent subroutines. More specifically, we discuss various quantum subroutines such as quantum phase estimation and amplitude amplification, as well as the important question of loading data into a quantum computer, via quantum RAM. The improvements to the original algorithm exploit variable-time amplitude amplification as well as a method for implementing linear combinations of unitary operations (LCUs) based on a decomposition of the operators using Fourier and Chebyshev series. Finally, we discuss a linear solver based on the quantum singular value estimation (QSVE) subroutine.
\\
\end{abstract}

\hypersetup{linkcolor=blue}

\tableofcontents


\section{Introduction} \label{sec:intro}

\subsection{Motivation} \label{subsec:mot}

Quantum computing was introduced in the $1980$s as a novel paradigm of computation, whereby information is encoded within a quantum system, as opposed to a system governed by the laws of classical physics.
Wider interest in quantum computation has been motivated by Shor's quantum algorithm for integer factorisation~\cite{shor1999polynomial}, which provides an exponential speed-up over the best known classical algorithm for the same task.
If implemented at scale, this would have severe security consequences for the ubiquitous RSA cryptographic protocol.
Since then, a number of quantum algorithms demonstrating advantage over classical methods have been developed for a substantial variety of tasks; for a detailed survey, the reader is directed to ~\cite{cleve1998quantum, montanaro2015quantum}.

Consistent advances on both theoretical and experimental research fronts have meant that the reality of a quantum computer has been edging ever closer.
Quantum systems are extremely sensitive to noise, with a primary challenge being the development of error correction in order to achieve fault tolerance~\cite{Campbell2017}.
Nonetheless, the current advances observed over recent years in the quest to build a scalable universal quantum computer from both academic and industrial groups raise the question of applications of such a device.

Although powerful quantum algorithms have been devised, their application is restricted to a few use cases. Indeed, the design of a quantum algorithm directly relies on exploiting the laws and features of quantum mechanics in order to achieve a speed-up. More precisely, in quantum computing, a quantum state is first prepared, to which quantum operations are applied before final measurements are performed, thus yielding classical data. As discussed in~\cite{aaronson2015read}, this model raises a number of challenges. In particular, careful consideration of how classical data can be input and obtained as output is crucial to maintaining the theoretical advantage afforded by quantum algorithms.

The question of solving a system of linear equations can be found at the heart of many problems with a wide scope of applications. An important result in recent years has been the Quantum Linear System algorithm (QLSA)~\cite{harrow2009quantum}, also called Harrow-Hassidim-Lloyd (HHL) algorithm, which considers the quantum version of this problem. In particular, the HHL algorithm run on the quantum problem (that is, with quantum states as input and output) offers an exponential speed-up over the best known classical algorithm run on the classical problem. In the following, we present a complete review of the HHL algorithm and subsequent improvements for sampling from the solutions to linear systems. We have aimed to be as complete as possible, including relevant background where necessary. We assume knowledge of elementary linear algebra and some experience with analysis of classical algorithms. 

\subsection{Quantum linear systems algorithms}

Solving a linear system is the task of taking a given matrix $A$ and vector $\vb{b}$ and returning a vector $\vb{x}$ satisfying $A\vb{x} = \vb{b}$.
As a preview of what is coming up ahead, Table~\ref{tab:LSAQLSAcompare} shows the runtime of the best classical algorithm for solving linear systems, conjugate gradient (CG), compared with the quantum algorithms we shall introduce throughout these notes.

We note that CG solves a linear system completely, i.e. it returns the solution vector $\vb{x}$.
The quantum algorithms allow one to \emph{sample} from the solution efficiently, providing one has an efficient preparation method for the input state, i.e. a mapping from the vector $\vb{b}$ to a quantum state $\ket*{b}$.
  
\begin{table}[h!]
	\centering
	\begingroup
	\setlength{\tabcolsep}{8pt} 
	\renewcommand{\arraystretch}{1.25} 
	\begin{tabular}{|c|c|c|}
		\hline
		\textbf{Problem} & \textbf{Algorithm} & \textbf{Runtime Complexity} \\
		\hline\rule{0pt}{2ex}
		LSP & CG~\cite{Shewchuck1994} & $\Ord{N s \kappa \log(1/\epsilon)}$ \\
		QLSP & HHL~\cite{harrow2009quantum} & $\Ord{\log(N) s^2 \kappa^2 / \epsilon}$ \\
		QLSP & VTAA-HHL~\cite{ambainis2010variable} & $\Ord{\log(N) s^2 \kappa / \epsilon}$ \\
		QLSP & Childs \emph{et. al.}~\cite{childs2015quantum} & $\Ord{s \kappa \operatorname{polylog}(s \kappa / \epsilon)}$ \\
		QLSP & QLSA~\cite{wossnig2017quantum} & $\Ord{\kappa^2 \operatorname{polylog}(n) \norm{A}_F / \epsilon}$ \\
		\hline
	\end{tabular}
	\\[0.05cm]
	\caption{Runtime comparison between various quantum linear systems algorithms and the best general-purpose classical algorithm, conjugate gradient (CG).
	The parameter $N$ is the dimension of the system, $s$ is the sparsity of the matrix $A$, $\kappa$ is the condition number of $A$, $\epsilon$ is the desired precision and $\norm{\,\cdot\,}_F$ is the Frobenius norm. 
	}
	\label{tab:LSAQLSAcompare}
	\endgroup
\end{table}

We shall formally define the linear systems problem (LSP) and its quantum variant (QLSP) in section~\ref{sec:LinSysdefinitions} and will discuss the differences between the two.
We will discuss efficient state preparation in section~\ref{subsec:qram}.

\subsection{Quantum computing} \label{subsec:qc}
In this section, we introduce \emph{gate-model quantum computation}, the computational model which will be used throughout. For a complete introduction to quantum computing we refer the reader to Nielsen and Chuang~\cite{nielsen2002quantum}.

In classical computing, the input is a classical bit string which, through the application of a circuit, is transformed to an output bit string. This is achieved via the application of a finite number of classical gates picked from a universal gate set such as $\{$NAND$\}$. This framework is known as the classical circuit model of computation. In the case of quantum computing, there exist different yet equivalent frameworks in which the computation can be described: the quantum circuit model~\cite{nielsen2002quantum}, measurement-based quantum computing (MBQC)~\cite{raussendorf2003measurement} and adiabatic quantum computing~\cite{farhi2000quantum}. In the following, we shall concentrate on the quantum circuit model, which is the closest quantum analogue of the classical circuit model as well as the model in which quantum algorithms are generally presented.

In classical computing the fundamental unit of information is the bit, which is either $0$ or $1$, whereas in quantum computing, it is the qubit, $|\psi \rangle= \alpha_0 |0\rangle + \alpha_1 |1\rangle$, such that $\alpha_0, \alpha_1 \in \mathbb{C}$ and $|\alpha_0|^2 + |\alpha_1^2| = 1$. This is represented by a two-dimensional column vector belonging to a complex Hilbert space $\mathcal{H} \cong \mathbb{C}^2$.
We shall denote by $\bra*{\psi}$ the conjugate-transpose of $\ket*{\psi}$.
The states $|0\rangle $ and $|1\rangle $ are basis vectors corresponding to a bit value of `0' and `1' respectively, and we can write  $\ket*{0} \cong \smqty(1 \\ 0)$ and $\ket*{1} \cong \smqty(0 \\ 1)$. Thus, we have that a qubit is a normalised complex superposition over these basis vectors.
Multiple qubits are combined using the tensor product, that is, for two qubits $\ket*{\phi},\ket*{\psi}$, their joint state is given by $\ket*{\phi} \otimes \ket*{\psi}$.
Thus, an $n$-qubit quantum state can be expressed as
\begin{equation}
|\psi \rangle = \sum_{i_1, \ldots, i_n} \alpha_{i_1 \ldots i_n} |i_1 \ldots i_n \rangle,
\end{equation}
where $i_k \in \{0,1\}$, $\sum_{i_1 \ldots i_n} |\alpha_{i_1 \ldots i_n}|^2=1$ and $\ket*{ i_1 \ldots i_n } \equiv \ket*{i_1}\otimes \cdots \otimes \ket*{i_n}$.
Note that $2^n$ complex coefficients are required to describe a quantum state, a number growing exponentially with the system's size.
We call the basis $\left\{ \ket*{ i_1 \ldots i_n }\mid i_k \in \{0,1\} \right\}$ the \emph{computational basis}, as each basis vector is described by a string of $n$ bits.

There are two types of operations we can perform on a quantum state: \emph{unitary operators} and \emph{measurements}.
A unitary operator $U$ has the property that $UU^\dagger=U^\dagger U = I$, i.e. its inverse is given by the hermitian conjugate. Furthermore, this implies that the operator is norm-preserving, that is, unitary operators map quantum states to quantum states.
A unitary operator on $n$ qubits can be expressed as matrix of dimension $2^n \times 2^n$. Moreover, we have that unitary operators are closed under composition.
A measurement is described by a collection of (not necessarily unitary) operators $\{ M_k \}$, where the index $k$ indicates a given measurement outcome.
The operators $M_k$ act on the state's Hilbert space and satisfy the \emph{completeness equation}, $\sum_k M_k^\dagger M_k = I$.
For a quantum state $\ket*{\psi}$, the probability of measuring outcome $m$ is given by
\begin{equation}
	p(m) = \bra*{\psi} M_m^\dagger M_m \ket*{\psi}
\end{equation}
and the resulting quantum state is then
\begin{equation}
	\frac{M_m \ket*{\psi}}{\sqrt{\bra*{\psi} M_m^\dagger M_m \ket*{\psi}}}.
\end{equation}
The completeness equation encodes the fact that measurement probabilities over all outcomes sum to unity.
A \emph{computational basis measurement}, $\{ M_x \}$ for $x \in \{0, 1\}^n$ consists of operators $M_x = \ketbra*{x}$, the projectors onto the computational basis states.

In the circuit model of quantum computation, we are given an input $x\in \{0, 1\}^n$, which is a classical bit string.
The first step is to prepare an $m$-qubit qubit quantum input state $|\psi \rangle $, where $m= \,$poly$(n)$. A unitary operator $U$ is then applied to $\ket*{\psi}$, and finally the output state is (customarily) measured in the computational basis -- without loss of generality. The measurement outcome corresponds to a classical bit string $y \in \{0, 1\}^m$, which is obtained with probability $\abs{\bra*{y} U \ket*{\psi}}^2$ and which we refer to as the output of the computation.

In practice, a quantum computer will be built using a finite set of quantum gates which act on a finite number of qubits. Typically, we consider gates acting on either one or two qubits, leaving the others invariant.
A set of quantum gates is said to be \emph{universal} if any unitary operator can be approximated `well-enough' using only gates from this set.
More precisely, a set of gates $S$ is universal if any unitary operator $U$ can be decomposed into the sequence $U_L U_{L-1}\ldots U_1$, such that $||U-U_L\ldots U_1||_2 \leq \varepsilon$, for any $ \varepsilon > 0$, where the $U_k \in S$. There are many such universal gate sets, such as for instance the Toffoli gate (which acts on three bit/qubits) and the Hadamard gate, or single-qubit rotations with a CNOT. Thus, any arbitrary unitary operator $U$ can be implemented given a universal set of gates.

This thus tells us that any arbitrary unitary operator $U$ can be approximated by the sequence $U_L U_{L-1}\ldots U_1$ to accuracy $\varepsilon$. But, how many gates $L$ are required to achieve a good accuracy? The Solovay-Kitaev theorem (see~\cite[Appendix 3]{nielsen2002quantum}) states that $L=\mathcal{O}\big( \log^2 \frac{1}{\varepsilon} \big)$, and thus exponential accuracy can be achieved using only a polynomial number of gates.

Finally, we discuss an important tool used in quantum computation, the oracle. Here, we are given a boolean function $f: \{0,1\}^n \to \{0,1\}^m$. The function is said to be queried via an \emph{oracle} $\mathcal{O}_f$, if given the input $|x\rangle |q\rangle$ (where $x \in \{0,1\}^n$ and $q \in \{0,1\}^m$), we can prepare the output $|x\rangle |q \oplus f(x)\rangle$, where $\oplus$ denotes addition modulo $2$. That is, the mapping
\begin{equation}
|x\rangle |q\rangle \to |x\rangle |q \oplus f(x)\rangle
\end{equation}
can be implemented by a unitary circuit $U_f$, which takes the form
\begin{equation}
	U_f = \sum_{x\in\{ 0,1 \}^n} \sum_{q\in \{0,1\}^m } \ketbra*{x} \otimes \ketbra*{q \oplus f(x)}{q}.
\end{equation}
The effect of the oracle needs to be determined on all basis states, and the definition will always be given in terms of a state $|q\rangle$.

\subsection{Quantum algorithms and machine learning} \label{subsec:qml}

Quantum algorithms, in some cases, have the capacity to achieve significant speed-ups compared to classical algorithms. Most notably, classically, the prime factorization of an $n$-bit integer using the general number field sieve is $\Omega (N^{1/3}\log^{2/3}N)$, where $N=\mathcal{O} (2^n)$, and thus takes time exponential in the number of bits. In contrast, Shor's factorization algorithm \cite{shor1999polynomial} achieves an astonishing exponential speed-up with a polynomial runtime of $\mathcal{O}((\log n)^3)$. Another impressive result is the quadratic speed-up obtained by Grover's algorithm for unstructured database search~\cite{grover1996fast,grover1997quantum}, which we discuss in detail in section \ref{subsec:aa}. These are just some examples of the many quantum algorithms which have been devised over the past decades \cite{cleve1998quantum,montanaro2015quantum}.

Machine learning  \cite{abu2012learning} has had and continues to have a significant impact for artificial intelligence and more generally for the advancement of technology. Naturally, this raises the question of whether quantum computing could enhance and improve current results and techniques. The HHL algorithm \cite{harrow2009quantum} considers the problem of solving a system of linear equations, which on a classical computer takes time polynomial in the system size $n$. At its heart, the problem reduces to matrix inversion, and the HHL algorithm offers an exponential speed-up for this task, with a certain number of important caveats. This in turn raised the question of whether quantum algorithms could accelerate machine learning tasks, which is referred to as \emph{quantum machine learning} (QML) --  see the following reviews~\cite{aimeur2006machine, ciliberto2017quantum}.

An interesting example which illustrates how quantum computing might help with machine learning tasks is quantum recommendation systems \cite{kerenidis2016quantum}. Here, the \emph{netflix problem} \cite{koren2009matrix, bell2007lessons} is considered, whereby we are given $m$ users and $n$ films, and the goal is to recommend a film to a user which they have not watched, that they would rate highly given their previous rating history.
The users' preferences can be represented by a matrix $P$ of dimension $m \times n$, where the $(i,j)^{\text{th}}$ entry corresponds to the $i^{\text{th}}$ user's rating of the $j^{\text{th}}$ film. Of course, the elements of $P$ are not all known, and the question is to provide good recommendations to the user. Classical algorithms for this problem run in time polynomial with matrix dimension, that is $\mathcal{O}\big(\text{poly}(MN)\big)$. In contrast, there exists a quantum algorithm with runtime complexity scaling as $\mathcal{O}\big(\text{poly}\log (MN)\big)$, thus providing an exponential speed-up \cite{kerenidis2016quantum}.

Another important example is the classical perceptron model, where we are given $N$ labeled data points which are linearly separable and the goal is to find the separating hyperplane. Classically, we have that the number of training updates scales as $\mathcal{O}\big( \frac{1}{\gamma^2}\big)$, where $\gamma$ is the margin, i.e. the shortest distance between a training point and the separating hyperplane. In contrast, the quantum perceptron model \cite{wiebe2016quantum} exploits Grover's search algorithm (see section \ref{subsec:aa}), in order to achieve a quadratic speed-up $\mathcal{O}\big(\frac{1}{\gamma}\big)$.

Another important classical model for supervised learning is support vector machines (SVM), which are used for data classification and regression. For a special case of SVMs, known as least-squares SVMs, the classical runtime is polynomial in the number of training examples $N$ and their dimension $d$, $\mathcal{O}(\log\frac{1}{\varepsilon}\text{poly}(d,N))$, where
 $\varepsilon$ denotes the accuracy. In contrast, quantum SVM \cite{rebentrost2014quantum} offer an exponential speed-up in the dimensions and input number with a runtime of $\mathcal{O}(\frac{1}{\epsilon}\log (dN))$.

Finally, we note that quantum algorithms have also been developed for unsupervised learning, that is, in the case of unlabeled data \cite{aimeur2013quantum, wiebe2014quantum}, which also present a speed-up over their classical counterparts.

All of the algorithms mentioned here use the HHL -- or some other quantum linear systems algorithm -- as a subroutine.

\subsection{Structure} \label{subsec:struc}
We have aimed for a mostly modular presentation, so that a reader  familiar with a particular subroutine of the HHL algorithm, say, phase estimation, can skip the relevant section if they wish.
The structure of the text goes as follows.

First, in section \ref{sec: quantum alg fund comp}, we review some of the key components used in quantum algorithms, namely notation conventions \ref{subsec:convention}, the quantum Fourier transform \ref{subsec:QFT}, Hamiltonian simulation \ref{subsec:ham_sim}, quantum phase estimation \ref{subsec:phase est}, phase kickback \ref{sec:phase_kickback}, amplitude amplification \ref{subsec:aa}, the uncompute trick \ref{subsec:uncompute} and finally quantum RAM \ref{subsec:qram}. Next, in section \ref{sec:HHL}, we present a detailed discussion of the HHL algorithm. We first formally define the problem in \ref{sec:LinSysdefinitions}, before then discussing the algorithm in detail in \ref{subsec:HHL}, along with an error analysis in \ref{subsec:hhl_err_analysis}. In section \ref{subsec:optimality} we consider the computational complexity of the problem, and in \ref{subsec:optimality} its optimality. Then, in \ref{subsec:non-hermitian}, the algorithm is extended to the case of non-Hermitian matrices. 

Next, in section \ref{sec: improvements HHL}, we introduce modern updates to the algorithm, namely: improvements in the dependency on the condition number in \ref{subsec:variable time}; improvements on the precision number in \ref{subsec:exp. improved} and in section \ref{sec:QSVE}, further improvements based on quantum singular value estimation giving an algorithm for dense matrices. More specifically, we discuss Jordan's lemma and its consequences in \ref{subsec:intersection_subspaces}, before reviewing the singular value decomposition in \ref{subsec:svd_and_subspaces} and finally presenting the quantum singular value estimation \ref{subsec:qsve} and its application to linear systems in \ref{subsec:qlsa_from_qsve}.
This section deviates from the pedagogical style of the previous sections, giving an overview of the important ideas behind the results as opposed to all of the gory details. 

\section{Quantum algorithms: fundamental components} \label{sec: quantum alg fund comp}
Here we review some of the fundamental ideas routinely used in the design of quantum algorithms, which will subsequently be applied in the HHL algorithm.

\subsection{Notation and conventions} \label{subsec:convention}
Any integer $k$ between $0$ and $N-1$, where $N=2^n$ may be expressed as an $n$-bit string $k=k_1 \ldots k_n$, i.e. $k=\sum_{l=1}^{n} k_l 2^{n-l}$. Furthermore, given an integer $j$, it is easy to verify that $\frac{j}{2^m}=0.j_{n-m+1}\ldots j_n$.

The Hadamard matrix is defined as $H := \frac{1}{\sqrt{2}}\smqty[1 & 1 \\ 1 & -1]$. Given a vector $\textbf{x}= (x_1, \ldots, x_n)$, the Euclidean norm of $\textbf{x}$ is $\norm{\textbf{x}}_2=\sum \sqrt{|z_1|^2+\ldots + |z_n|^2}$. For an $m \times n$ matrix $A$ with elements $a_{ij}$, the operator norm is given by $\norm{A}_2= \sum_{j=1}^n \sqrt{\sum_{i=1}^m |a_{ij}|}$, i.e., it is the sum of the Euclidean norms of the column vectors of $A$. Note that in both cases this is equivalent to $\ell_2$-norm. Next, we present the QFT.

\subsection{Quantum Fourier transform}
\label{subsec:QFT}
The QFT is at the heart of many quantum algorithms. It is the quantum analogue of the discrete Fourier transform, see \ref{subsubsec:DFT}, and is presented in section \ref{subsubsec:QFT}. In section \ref{subsubsec: QFT implementation}, we will see how the QFT may be implemented efficiently using a quantum computer.

\subsubsection{The discrete Fourier transform} \label{subsubsec:DFT}
The Fourier transform is an important tool in classical physics and computer science, as it allows for a signal to be decomposed into its fundamental components, i.e. frequencies. The Fourier transform tells us what frequencies are present and to what degree.

In the discrete setting, we have that the DFT  is a square invertible matrix $D$ of dimension $N$, where $D_{jk}=\frac{1}{\sqrt{N}} \omega^{jk}$ and $\omega=e^{i2\pi/N}$. It is easy to show that the columns of this matrix are orthogonal and have unit length, and thus the set of column vectors form an orthonormal basis which we refer to as the Fourier basis. If the DFT is applied to a vector using matrix multiplication, then the time complexity scales as $\mathcal{O}(N^2)$. Crucially, by using a divide and conquer approach this can be improved to $\mathcal{O}(N\log N)$, which is referred to as the fast Fourier transform (FFT).

\subsubsection{The quantum Fourier transform} \label{subsubsec:QFT}
In a similar way, the QFT is defined by mapping each computational basis state $|j\rangle$ to a new quantum state $|f_j\rangle=\frac{1}{\sqrt{N}} \sum_{k=0}^{N-1} \exp\Big( i\frac{2\pi j k}{N} \Big)|k\rangle$. The set of orthonormal states $\{|f_j\rangle \}$ form an orthonormal basis set called the Fourier basis. The quantum Fourier transform with respect to an orthonormal basis $\ket*{x} \in \{\ket*{0},\ldots, \ket*{N-1}\}$ is defined as the linear operator with the following action on the basis vectors:
\begin{equation}
\text{QFT}: \ket*{x} \rightarrow \frac{1}{\sqrt{N}} \sum_{k=0}^{N-1} \omega^{x\cdot k} \ket*{k}.
\end{equation}
The inverse Fourier transform is then defined as
\begin{equation}
\text{QFT}^{\dagger}: \ket*{k} \rightarrow \frac{1}{\sqrt{N}} \sum_{x=0}^{N-1} \omega^{-k\cdot x} \ket*{x}.
\end{equation}

But, what does this mean in terms of individual qubits?
First, we represent the integer $k$ in binary notation, $k=\sum_{l=1}^{n} k_l 2^{n-l}$, and thus the Fourier basis states can be expressed as: $|f_j\rangle=\frac{1}{\sqrt{N}}\sum_{k=0}^{N-1} \exp\Big( i2\pi j \sum_{l=1}^{n} k_l 2^{-l} \Big)|k\rangle$. Expanding this expression, we have
\begin{equation}
|f_j\rangle =\frac{1}{\sqrt{N}} \sum_{k_1, \ldots, k_n} \exp\big(i2\pi j k_1 2^{-1}\big) \exp\big( i2\pi j k_2 2^{-2}\big)\ldots \exp\big( i2\pi j k_n 2^{-n}\big)|k_1 \ldots k_n\rangle .
\end{equation}
Expanding the summation gives
\begin{equation}
|f_j\rangle =\frac{1}{\sqrt{N}} \Big(|0\rangle + \exp\big(i2\pi j  2^{-1}  \big)|1\rangle\Big) \Big(|0\rangle+ \exp\big( i2\pi j  2^{-2}\big)|1\rangle\Big) \ldots \Big(|0\rangle+ \exp\big( i2\pi j  2^{-n}\big) |1\rangle\Big).
\end{equation}
Finally, the operation $\frac{j}{2^m}=0.j_{n-m+1}\ldots j_n$  corresponds to the decimal expansion of $j$ up to $m$ bits, and we can thus write
\begin{equation}\label{equ:QFT}
|f_j\rangle =\frac{1}{\sqrt{N}} \Big(|0\rangle+ \exp\big(i2\pi 0.j_n  \big)|1\rangle\Big) \Big(|0\rangle+ \exp\big( i2\pi 0.j_{n-1}j_n \big) |1\rangle\Big) \ldots \Big(|0\rangle+ \exp\big( i2\pi j  0.j_1 \cdots j_n \big)|1\rangle\Big).
\end{equation}

We initially applied the QFT to the computational basis state $|j\rangle =|j_1\ldots j_n \rangle $. From Eq.~\eqref{equ:QFT}, we see that information pertaining to the input $j$ is disseminated throughout the relative phase on each individual qubit. Thus, given the final state $|f_j\rangle$, applying the inverse QFT would yield the input string $j$. Equivalently, one could obtain $j$ by performing a measurement in the Fourier basis.

\subsubsection{Implementation of the QFT} \label{subsubsec: QFT implementation}
The goal is to obtain the quantum state $|f_j\rangle$ after applying a quantum circuit to an all-zero input state. From Eq.~\eqref{equ:QFT} we see that the state $|f_j\rangle$ corresponds to a state where each qubit is initialised in the state $|+\rangle =\frac{|0\rangle +|1\rangle}{\sqrt{2}}$ and subsequently acquires a relative phase of $\exp\big(i2\pi j 2^{-k}\big)$, where $1 \leq k \leq n$, and where we recall that $1 \leq j \leq N-1$.

We now consider the quantum circuit which can implement this state. First, it is easy to see that a state of the form $\frac{1}{\sqrt{2}}\Big(|0\rangle+ \exp\big(i2\pi 0.j_n  \big)|1\rangle\Big) $ corresponds to either the state $|+\rangle $ or $|-\rangle $ depending on the value of $j_n \in \{0,1\}$. This can be expressed as $H|j_n\rangle$, and can thus be obtained by the application of a Hadamard gate $H$ to a qubit in the state $|j_n\rangle$.
Next, the state of the second qubit is given by $\frac{1}{\sqrt{2}}\Big(|0\rangle+ \exp\big( i2\pi 0.j_{n-1}j_n \big) |1\rangle\Big)$, which can be re-expressed as $\frac{1}{\sqrt{2}}\Big(|0\rangle+ (-1)^{j_{n-1}} \exp\big( i2\pi \frac{j_n}{2^2} \big) |1\rangle\Big)$.
Thus, this corresponds to first preparing the state $H|j_{n-1}\rangle$, and then applying a controlled rotation to the qubit, where the control is the $n$th qubit in the state $|j_n\rangle$. Thus, the state on the first two qubits can be obtained by preparing the state $|j_{n-1} j_n \rangle$, applying the Hadamard gate to the $n$th qubit, and then a controlled rotation $R_2$ with qubit $|j_{n-1}\rangle $ as control, where we have that $R_k$ is given by:
 \[
 R_k=
 \left[ {\begin{array}{cc}
 	1 & 0 \\
 	0 & e^{i2\pi \frac{1}{2^k}} \\
 	\end{array} } \right],
 \]
and controlled by qubit $|j_{n-k+1}\rangle$. Finally, SWAP operations are performed  throughout for the qubits to be in the correct order. This approach can be extended to the $n$ qubits, where the number of gates scales as $\mathcal{O}(n^2)$, and thus the QFT can be efficiently implemented in a quantum circuit.

\subsection{Hamiltonian simulation}
\label{subsec:ham_sim}

Most quantum algorithms for machine learning, and in particular the HHL algorithm,
leverage quantum Hamiltonian simulation as a subroutine.
Here, we are given a Hamiltonian operator $\hat{H}$, which is a Hermitian matrix,
and the goal is to determine a quantum circuit which implements the unitary operator $U=e^{-i\hat{H}t}$,
up to given error.
The evolution of a quantum state $\ket*{\Psi}$ under a unitary operator is given, for simplicity, by the time-independent Schr\"{o}dinger equation:
\begin{equation}
i \frac{d}{dt} \ket*{\Psi} = \hat H \ket*{\Psi}, 
\end{equation}
the solution to which can be written as $\ket*{\Psi (t)} = \exp (-i\hat Ht)\ket*{\Psi}$.

Depending on the input state and resources at hand, there exists many different techniques to achieve this~\cite{berry2009black,berry2015hamiltonian,low2017hamiltonian,low2017optimal,low2016hamiltonian}. We give a brief introduction to this large field of still ongoing research, and interested readers can find further details in the lecture notes of Childs~\cite{childs2017lecture}[Chapter V] and the seminal work~\cite{childs2003exponential}.

The challenge is due to the fact that the application of matrix exponentials are computationally expensive. For instance, naive methods require time $\Ord{N^3}$ for a $N \times N$ matrix, which is restrictive even in the case of small size matrices. 
In the quantum case, the dimension of the Hilbert space grows exponentially with the number of qubits, and thus any operator will be of exponential dimension. Applying such expensive matrix exponentials has been studied classically, and in particular the classical simulation of such time-evolutions is known to be hard for generic Hamiltonians $\hat H$. As a consequence, new more efficient methods need to be introduced. In particular, a quantum computer can be used to simulate the Hamiltonian operator, a task known as \emph{Hamiltonian simulation}, which we wish to perform efficiently. More specifically, we can now define an efficient quantum simulation as follows: \\

\begin{definition}\emph{(Hamiltonian Simulation)}
	We say that a Hamiltonian $\hat{H}$ that acts on $n$ qubits can be efficiently simulated if for
	any $t > 0, \epsilon > 0$, there exists a quantum circuit $U_{\hat{H}}$ consisting of
	$\text{poly}(n, t, 1/\epsilon)$ gates such that $\norm{U_{\hat{H}} - e^{-i\hat{H}t}} < \epsilon$.
	Since any quantum computation can be implemented by a sequence of Hamiltonian simulations,
	simulating Hamiltonians in general is $\mathsf{BQP}$-hard, where $\mathsf{BQP}$ refers to the complexity
	class of decision problems efficiently solvable on a universal quantum computer \cite{kitaev2002classical}.
\end{definition}
Note that the dependency on $t$ is important and it can be shown that at least time $\Omega(t)$ is required to
simulate $\hat{H}$ for time $t$, which is stated formally by the  \textrm{no fast-forwarding} theorem~\cite{berry2007efficient}. 
There are, however, no nontrivial lower bounds on the error dependency $\epsilon$.
The hope to simulate an arbitrary Hamiltonian efficiently is diminished, since it \textsf{NP}-hard to find an approximate decomposition into elementary single- and two-qubit gates for a generic unitary and hence also for the evolution operator~\cite{shende2006synthesis,iten2016quantum,knill1995approximation}. Even more so, the optimal circuit sythesis was even shown to be \textsf{QMA}-complete~\cite{janzing2003identity}. However, we can still simulate efficiently certain classes of Hamiltonians, i.e.\ Hamiltonians with a particular structure.
One such example is the case when $\hat{H}$ only acts nontrivially on a constant number of qubits, as any unitary
evolution on a constant number of qubits can be approximated with error at most $\epsilon$ using
$\text{poly}(\log(1/\epsilon)$ one- and two-qubit gates, on the basis of Solovay-Kitaev's theorem.
The origin of the hardness of Hamiltonian simulation stems from the fact that we need to find a
decomposition of the unitary operator in terms of elementary gates, which in turn can be very hard for generic Hamiltonians.
If $\hat{H}$ can be efficiently simulated, then so can $c \hat{H}$ for any $c = \text{poly}(n)$~\cite{childs2017lecture}.
In addition, since any computation is reversible, $e^{-i\hat{H}t}$ is also efficiently simulatable and this must hence also hold for $c<0$.\\
Finally, we note that the definition of efficiently simulatable Hamiltonians further extends to unitary matrices, since every operator $U_{\hat H}$ corresponds to a unitary operator, and furthermore every unitary operator can be written in the form $\exp (i \hat H)$ for a Hermitian matrix $\hat H$. Hence, we can similarly speak of an efficiently simulatable unitary operator, which we will use in the following.

\subsubsection{Trotter-Suzuki methods}
\label{sec:ham_sim_decomp}

For any efficiently simulatable unitary operator $U$, we can always simulate the Hamiltonian $\hat H$ in a transformed basis $U\hat H U^\dag$, since
\begin{equation}
	e^{-iU\hat{H}U^{\dagger}t} = U e^{-i\hat{H}t} U^{\dagger},
\end{equation}
which follows from the fact that if $U=U^{\dagger}$ is unitary, then we have that $(U\hat{H}U^{\dagger})^m = U \hat{H}^m U^{\dagger}$, which can easily be proven by induction.
Another simple but useful trick is given by the fact that, given efficient access to any diagonal element of a Hamiltonian $H_{ii} = \bra*{i} \hat{H} \ket*{i}$, we can simulate the diagonal Hamiltonian using the following sequence of operations. Let $\Rightarrow$ indicate a computational step, such that we can denote in the following a sequence of maps to a state:
\begin{eqnarray}
\label{eq:sim_diag_ham}
	\ket*{i,0} \rightarrow \ket*{i, {H}_{ii}} \\
	\rightarrow e^{-i {H}_{ii}t} \ket*{i, {H}_{ii}} \\
	\rightarrow e^{-i {H}_{ii}t} \ket*{i, 0} = e^{-i\hat{H}t} \ket*{i} \otimes \ket*{0}.
\end{eqnarray}
In words, we first load the entry $H_{ii}$ into the second register, then apply a conditional-phase gate $\exp(-i {H}_{ii}t)$ and then reverse the loading procedure to set the last qubit to zero again.
Since we can apply this to a superposition and using linearity, we can simulate any efficiently diagonalisable Hamiltonian.
More generally, any $k$-local Hamiltonian, i.e.\ a sum of polynomially many terms in the number of qubits that each act on at most $k=\Ord{1}$ qubits,
can be simulated efficiently. Indeed, since each of the terms in the sun acts only on a constant number of qubits, it can be efficiently diagonalised and thus simulated.
In general, for any two Hamiltonian operators $\hat{H}_1$ and $\hat{H}_2$ that can be efficiently simulated,
 the sum of both $\hat{H}_1 + \hat{H}_2$ can also be efficiently simulated, as we will argue below, first for the commuting case and then for the non-commuting case.\\
This is trivial if the two Hamiltonians commute. Indeed, we now omit the coefficient $it$ and consider for simplicity the operator $\exp (\hat{H}_1 + \hat{H}_2)$. By applying a Taylor expansion, followed by the Binomial theorem and the Cauchy product formula (for the product of two infinite series), we have 
\begin{eqnarray}
\exp (\hat H_1 + \hat H_2) &=& \sum_{n=0}^{\infty}\frac{(\hat H_1 + \hat H_2)^n}{n!}\\
&=& \sum_{n=0}^{\infty}\sum_{k=0}^{n}\binom{n}{k}\frac{\hat H_1^k \hat H_2^{n-k}}{n!}\\
&=& \sum_{n=0}^{\infty}\sum_{k=0}^{n}\frac{\hat H_1^k \hat H_2^{n-k}}{k!(n-k)!}\\
&=& \left(\sum_{k=0}^{\infty}\frac{\hat H_1^k}{k!}\right)\cdot\left(\sum_{n=0}^{\infty}\frac{\hat H_2^n}{n!}\right)\\
&=& \exp (\hat H_1) \cdot \exp (\hat H_2).
\end{eqnarray}
Note that this is only possible since for the Cauchy formula we can arrange the two terms accordingly and do not obtain
commutator terms in it. However (recall the famous Baker-Campbell-Hausdorff formula,
see e.g.\ \cite{rossmann2002lie}), this is not so for the general case, i.e.\ when the operators don't commute.
Here we need to use the Lie-Product formula~\cite{rossmann2002lie}:
\begin{equation}\label{eq:lieprodformula}
	e^{-i(\hat{H}_1+\hat{H}_2)t} = \lim_{m \rightarrow \infty} \left( e^{-i\hat{H}_1 t/m}  e^{-i\hat{H}_2 t/m} \right)^m.
\end{equation}
If we want to restrict the simulation to a certain error $\epsilon$, it is sufficient to truncate the above product formula after a certain number of iterations $m$, which we will call the number of steps:
\begin{equation}
\label{error_ham_sim}
	\norm{e^{-i(\hat{H}_1+\hat{H}_2)t} - \left( e^{-i\hat{H}_1 t/m}  e^{-i\hat{H}_2 t/m} \right)^m}_2 \leq \epsilon,
\end{equation}
which, as we will show, can be achieved by taking $m = \Ord{(\max{(||\hat{H}_1||,||\hat{H}_2||)})^2 t^2/\epsilon}$,
where we require that $\max{(||\hat{H}_1||,||\hat{H}_2||)}\sim \Ord{\text{poly}(n)}$ for the evolution to be
efficiently simulable. To see this, observe that, from the Taylor expansion,
\begin{eqnarray}
\left( e^{-i\hat{H}_1 t/m}  e^{-i\hat{H}_2 t/m} \right)^m   & = \left( I - i(\hat{H}_1 + \hat{H}_2) \frac{t}{m} + \mathcal{O}\left(\frac{t^2\max{(\hat{H}_1,\hat{H}_2)}^2}{m^2}\right) \right)^m \\ 
& =  \left( e^{-i(\hat{H}_1+\hat{H}_2)t/m}  + \mathcal{O}\left(\frac{t^2\max{(\hat{H}_1,\hat{H}_2)}^2}{m^2} \right) \right)^m.  
\end{eqnarray}
We need to expand a product of the form $(A+B)^m$, where the operators $A$ and $B$ are non-commuting. Thus, we have that 
\begin{eqnarray}
\label{expansionprod}
(A+B)^m= \underbrace{A^m}_{0^{th}\text{ order in }B} + \underbrace{A^{m-1}B + A^{m-2} B A + A^{m-3}BA^2 + + \ldots + BA^{m-1}}_{\text{m first order in }B}  + \nonumber \\  + \underbrace{A^{m-2}B^2 + A^{m-3}BAB + A^{m-3}B^2 A}_{\text{Second order in }B} + \ldots  + \underbrace{B^m}_{m^{th}\text{ order in }B}, 
\end{eqnarray}
where $A$ and $B$ do not commute in general.
Specifically we have $A =e^{-i(\hat{H}_1+\hat{H}_2)t/m}$ and $B =\Ord{\frac{t^2\max{\left(\hat{H}_1,\hat{H}_2 \right)}^2}{m^2}},$ and the $m$ first order terms in $B$ have the form
\begin{equation}
 e^{-i(\hat{H}_1+\hat{H}_2)t (m-k-1)/m} \Ord{\frac{t^2\max{\left(\hat{H}_1,\hat{H}_2 \right)}^2}{m^2}} e^{-i(\hat{H}_1+\hat{H}_2)t k/m},
\end{equation}
for $k \in [0,m-1]$. Next, we consider terms of order greater than one, where we have powers of $B^l$ for $l>2$. Let us note that for $m>1$, we have $C/m^2 + C/m^3 + C/m^4 + \ldots  \leq \Ord{C/m^2}$ for $m>1$. Thus, these first order and greater terms can be absorbed in the $\tilde{\mathcal{O}}$ notation.  Furthermore, in the following, we do not explicitly write the exponentials of the form $\exp (-i(H_1+H_2)t (m-k-1)/m)$ as these will not play a role for bounding the error in the norm due to their unitarity. So, we continue to bound 
\begin{equation}
\label{eq_error_intermediate}
e^{-i(\hat{H}_1+\hat{H}_2)t} + \tilde{\mathcal{O}}\left(\frac{t^2\max{(\hat{H}_1,\hat{H}_2)}^2}{m} \right). 
\end{equation}

We can now finally consider the error of the simulation scheme, see Eq.~(\ref{error_ham_sim}), which using Eq.~(\ref{eq_error_intermediate}), yields 
\begin{eqnarray}
\norm{\left( e^{-i\hat{H}_1 t/m}  e^{-i\hat{H}_2 t/m} \right)^m  -
e^{-i(\hat{H}_1+\hat{H}_2)t}} = \mathcal{O}\left(\frac{t^2\max{\left(\norm{\hat{H}_1}_2,\norm{\hat{H}_2}_2 \right)}^2}{m} \right). 
\end{eqnarray}
In order to have this error less than $\epsilon$, the number of steps $m$ must be $m = \mathcal{O} \left( \frac{t^2\max{(\hat{H}_1,\hat{H}_2)}^2}{\epsilon} \right)$.

This is a naive and non-optimal scheme. It can be shown that one can use higher-order approximation schemes, such that
$\hat{H}_1+\hat{H}_2$ can be simulated for time $t$ in $t^{1+\delta}$ for any positive  but arbitrarily small $\delta$~\cite{berry2007efficient,berry2009black}.\\
These so-called Trotter-Suzuki schemes can be generalized to an arbitrary sum of Hamiltonians which then leads to an approximation formula given by
\begin{equation}
\label{eq:trotter-suzuki}
	e^{-i(\hat{H}_1+\ldots + \hat{H}_k)t} = \lim_{m \rightarrow \infty} \left( e^{-i\hat{H}_1 t/m} \cdots  e^{-i\hat{H}_k t/m} \right)^m.
\end{equation}
The following definitions are useful:

\begin{definition} \emph{(Sparsity)} \label{def:s-sparse}
An $N \times N$ matrix is said to be \emph{$s$-sparse} if it has at most $s$ entries per row.
\end{definition}

\begin{definition} \emph{(Sparse matrix)}\label{def:sparse}
An $N \times N$ matrix is said to be \emph{sparse} if it has at most $\poly {\log N}$ entries per row.
\end{definition}
Note that the sparsity depends on the basis in which the matrix is given. However, given an arbitrary matrix, we do not a priori know the basis which diagonalises it
(and hence gives us a diagonal matrix with sparsity $1$), and so we need to deal with a potentially dense matrix, i.e.\ a matrix which has $N$ entries per row.
\begin{definition} \emph{(Row computability)}\label{def:effrow}
The entries of a matrix $A$ are efficiently row computable if, given the indices $i,j$, we can obtain the entries $A_{ij}$ efficiently, i.e. in $\Ord{s}$ time, where $s$ is the sparsity as defined above.
\end{definition}

\subsubsection{Graph-colouring method}

Crucially, the simulation techniques described above can allow us to efficiently simulate sparse Hamiltonians.
Indeed, if for any index $i$, we can efficiently determine all of the indices $j$ for
which the term $\bra*{i} \hat H \ket*{j}$ is nonzero, and furthermore efficiently obtain the values of the
corresponding matrix elements, then we can simulate the Hamiltonian $\hat H$ efficiently, as we will describe below.

This method of Hamiltonian simulation is based on ideas from graph theory, and we will now first briefly introduce a couple of key notions relevant in our discussion. For further information, we refer the reader to the existing literature \cite{west2001introduction}. An \emph{undirected graph} $G = (V,E)$ is specified by a set $V$ of $|V|=N$ vertices and a set $E$ of $|E|=M$ edges, i.e., unordered pairs of vertices. When two vertices form an edge they are said to be \emph{connected}. A graph can be represented by its \emph{adjacency matrix} $A$, where $A_{ij} =1$ if the vertices $i$ and $j$ are connected, and $A_{ij}=0$, otherwise. The \emph{degree} of a vertex is given by the number of vertices it is connected to. The \emph{maximum degree} of a graph refers the maximum degree taken over the set of vertices. The problem of edge colouring considers if, given $k$ colours, each edge can be assigned a specific colour with the requirement that no two edges sharing a common vertex should be assigned the same colour. Vizing's theorem tells us that, for a graph with maximum degree $d$, an edge colouring exists with $k$ at most $d+1$. Finally, a \emph{bipartite graph} is a graph, where the set of vertices can be separated into two disjoint subsets $V_1$ and $V_2$ such that $V=V_1 \cup V_2$ and no two vertices belonging to the same subset are connected, i.e., for every $i, j \in V_k$, we have $(i,j) \notin V_k$, for $k=1,2$.   

Previously, we saw that a Hamiltonian operator can be represented by a square matrix. Thus, a graph can be associated with any Hamiltonian by considering the adjacency matrix with a $1$ at every non-zero entry of the Hamiltonian, and a $0$ elsewhere, in the spirit of combinatorial matrix theory. For a matrix of dimension $N$, this will thus correspond to a graph with $N$ vertices. Previously, we saw that sparse Hamiltonians have at most polylog$(N)$ entries per row, and thus $\Ord{n}$, for $n=\log N$ entries in total. This will translate into a graph having a number of edges $|E|=\Ord{n}$.

Childs~\cite{childs2003exponential} proposed an efficient implementation for the simulation of sparse Hamiltonians by using the Trotterization scheme presented above (c.f. section~\ref{sec:ham_sim_decomp}) and a local colouring algorithm of the graph associated with the $s$-sparse Hamiltonian. The core idea is to colour the edges of the Hamiltonian $\hat H$. Then, the Hamiltonians corresponding to each subgraphs defined by a specific colour can be simulated, and finally the original Hamiltonian recovered via the Trotter-Suzuki method~\cite{childs2003exponential}.\\

More precisely, the first step is to find an edge-colouring of the graph associated with the Hamiltonian $\hat H$. This will be achieved using $k$ colours, which in the case of a sparse Hamiltonian, will be at most polylog$|n|$. Next, the graph can be decomposed by considering the subgraphs corresponding to a single colour. We thus obtain a decomposition of the original Hamiltonian in a sum of sparse Hamiltonians, containing at most polylog$|n|$ terms. It is easy to convince oneself that each of these terms consists of a direct sum of two-dimensional blocks. Indeed, each adjacency matrix corresponding to a subgraph will be symmetric, with at most one entry per row, meaning that the evolution on any one of these subgraphs takes place
in isolated two-dimensional subspaces. Thus, each Hamiltonian term can be easily diagonalised and simulated using the diagonal Hamiltonian simulation procedure as given in Eq.~(\ref{eq:sim_diag_ham}).\\

A crucial step in this procedure is the classical algorithm for determining the edge colouring efficiently. Vizing's theorem guarantees the existence of an edge colouring using $d+1$ colours. But, the question remains as to how this can be efficiently achieved. Indeed, even though we are given the adjacency matrix representation of the entire graph, we will now restrict  ourselves to accessing only local information i.e. each vertex has only access to information regarding it nearest-neighbours. Finding an optimal colouring is an \textsf{NP}-complete problem. However, there are polynomial time algorithms that construct optimal colourings of bipartite graphs, and colourings of non-bipartite simple graphs that use at most $d+1$ colours. It is important to note that the general problem of finding an optimal edge colouring is \textsf{NP}-hard and the fastest known algorithms for it take exponential time.\\

We thus now present a local edge-colouring scheme achieving a $d^2$-colouring (where we recall that $d$ is the maximum degree of the graph, i.e.\ the sparsity) for the case of a bipartite graph. This, using a reduction~\cite{childs2017lecture}, is sufficient for the simulation of an arbitrary 
Hamiltonian. Crucially, this scheme is efficient if the graph is sparse, i.e. if $d=$polylog$(n)$. We note that better schemes exist~\cite{berry2007efficient,berry2009black,berry2015hamiltonian} and can allow for polynomial improvements of the simulation scheme in comparison to the one given here.

\begin{lemma}\emph{(Efficient bipartite graph colouring~\cite{linial1987distributive,linial1992locality})}
\label{lemma:colouring}
Suppose we are given an undirected, bipartite graph $G$ with $n$ vertices and maximum degree $d$ (i.e.\ each vertex is
connected to a maximum of $d$ other vertices - the so called neighbours - which is similar to sparsity $s$),
and that we can efficiently compute the neighbours of any given vertex. Then there is an efficiently computable
edge colouring of $G$ with at most $d^2$ colours.
\end{lemma}

\begin{proof}
The vertices of $G$ are ordered and numbered from $1$ through $n$. For any vertex $a$, let $index(a, b)$ denote the index of vertex $b$ in the list of neighbours of $a$, ordered in increasing number.
For example, let $a$ have the neighbours $c,d$ with the list of neighbours neighbours $(a) := \{c,d\}$ of $a$. Then, we have that index$(a,c)=1$, and index$(a,d)=2$.
Then define the colour of the edge $(ab)$, where $a$ is from the left part of the bipartition and $b$ is from the right for all $a$ and $b$ which have an edge.
The colouring of this edge $(ab)$ is then assigned to be the ordered pair colour$(ab):= ($index$(a, b),$ index$(b, a))$.
Recall that an edge colouring assigns a colour to each edge so that no two adjacent edges share the same colour.
These assigned colours in form of the index-pairs give a valid colouring since if $(a, b)$ and
$(a, d)$ have the same colour, then index$(a, b) = \,$index$(a, d)$, so $b = d$ and similarly, if $(a, b)$
and $(c, b)$ have the same colour, then index$(b, a) =$ index$(b, c)$, so $a = c$.
\end{proof}
Using this lemma we can then perform Hamiltonian simulation in the following manner.
First we ensure that the associated graph is bipartite by simulating the evolution according to the Hamiltonian $\hat B =\sigma_x \otimes \hat H$, a block-anti-diagonal matrix 
\begin{equation}
    \hat B = \left( \begin{array}{c c}
    0 & \hat H\\
    \hat H^{\dagger} & 0
    \end{array} \right).
\end{equation}
The graph associated to this will be bipartite, with the same sparsity as $\hat H$ \cite{childs2017lecture}. Observe that simulating this reduces to simulating $\hat H$ since
\begin{equation}
	e^{-i(\sigma_x \otimes \hat H)t} \ket*{+} \ket*{\psi} = \ket*{+} \otimes e^{-i \hat H t} \ket*{\psi}.
\end{equation}
Without loss of generality let us assume now that $\hat H$ has only off diagonal entries. Indeed, any Hamiltonian $\hat H$ can be decomposed as the sum of diagonal and off-diagonal terms $\hat H_{\mathrm{diagonal}} + \hat H_{\mathrm{off-diagonal}}$, which can then be simulated the sum using the rule given in Eq.~\eqref{eq:lieprodformula}.
We can then, for a specific vertex $x$ and a colour $c$, compute the evolution by applying the following three steps:
\begin{enumerate}
	\item First we compute the complete list of neighbours of $x$ (i.e. the neighbour list and the indices)
		and each of the colours associated to the edges connecting $x$ with its neighbours, using the above local algorithm for graph colouring from Lemma~\ref{lemma:colouring}. 
	\item Let $v_c(x)$ denote the vertex adjacent to $x$ via an edge with colour $c$. We then, for a given $x$, compute $v_c(x)$ and retrieve the Hamiltonian matrix entry $H_{x,v_c(x)}$. We can then implement the following quantum state $\ket*{x, v_c(x), H_{x,v_c(x)}}$ i.e.\ three qubit registers in which we load the elements $x$ in the first one, $v_c(x)$ in the second and then load the matrix element into the last one. More specifically, we here prepare the quantum state $\ket*{x, v_c(x)} \otimes \ket*{a}$ and then using the state preparation oracle (e.g.\ qRAM, see section \ref{subsec:qram}), we obtain $\ket*{x, v_c(x)} \otimes \ket*{a \oplus H_{x,v_c(x)}}$ which
		can be done efficiently, i.e.\ in time $\Ord{\log(n,m)}$, which is the time required to access the data.
	\item We then simulate the ($\hat H$-independent, i.e.\ only depending on the local entry of $\hat H$ but
		not the general matrix) Hamiltonians $\hat h$, i.e.\ the at most polylog$|n|$ Hamiltonians we obtain from the graph colouring. Note that this is a sparse Hamiltonian which acts only on the $d$ neighbours 
         as described in the colouring step. The simulation is efficient, since $\hat h$ can be diagonalised in constant time,
		as it consists of a direct sum of two-dimensional blocks. 
Next, we apply a scheme, described below, whereby each complex matrix entry is decomposed into a real part $x$ and imaginary part $y$ and simulated separately. We can then simulate the diagonalised Hamiltonians such that we implement the following mapping
	\begin{equation}
		\ket*{a,b,h_{a}} \rightarrow h_a \ket*{a,b,h_{a}} \rightarrow ,
	\end{equation}
where $h_a$ is a diagonal element of the diagonalised Hamiltonian $\hat h$. This can also be done in superposition.
\end{enumerate}

The Hamiltonian to be simulated has complex entries, and can thus be decomposed in real and imaginary parts. Let $v_c(a)$ denote the vertex connected to $a$ via an edge of colour $c$. The original Hamiltonian had complex entries, and we can express the entry associated to vertex $v_c(a)$ as a sum of a real part $x_c(a)$ and  imaginary part $y_c(a)$. If we assume that these can be loaded independently i.e.\ $H_{a,v_c(a)}= x_c(a) + i\ y_c(a)$ can be loaded separately, then we can introduce the oracles $V_c,W_c$ which allow for the following mappings to be implemented:
\begin{eqnarray}
V_c \ket*{a,b,z} := \ket*{a, b\oplus v_c(a), z \oplus x_c(a)} \\ 
W_c \ket*{a,b,z} := \ket*{a, b\oplus v_c(a), z \oplus y_c(a)}, 
\end{eqnarray}
and similarly the inverse operations, for which it holds that $V_c^{\dagger} = V_c$, $W_c^{\dagger}=W_c$, since we have bitwise adding modulo $2$ in the loading procedure.

In order to simulate the complex entries we need to implement a procedure which allows us to apply both parts individually and still end up in the same basis-element such that these sum up to the actual complex entry, i.e.\ that we can apply $x_c(a)$ and $i\ y_c(a)$ separately to the same basis state. We use multiple steps to do so.\\
Given the above oracles, we can similarly simulate the following Hermitian operations (note that this is not a unitary operation):
\begin{eqnarray}
S \ket*{a,b,x}:= x \ket*{b,a,x}, \\
T \ket*{a,b,y}:= iy \ket*{b,a,-y},
\end{eqnarray}
where we apply the Hermitian operators to multiply with $x$ (and $iy$) and the swap operation to the registers, which can be implemented efficiently since the swap can be done efficiently. The operator $T$ is described in detail in~\cite{childs2003exponential}
We then can implement the operator \begin{equation}
\tilde{H} = \sum_c [V_c^{\dagger}S V_c + W_c^{\dagger} T W_c],
\end{equation}
where the sum is about all colours $c$. This acts then on $\ket*{a,0,0}$ as $H$, since
\begin{eqnarray}
\sum_c [V_c^{\dagger}S V_c + W_c^{\dagger} T W_c] \ket*{a,0,0} 
= \sum_c [V_c^{\dagger}S \ket*{a,v_c(a),x_c(a)}  + W_c^{\dagger} T \ket*{a,v_c(a),y_c(a)}] \\
= \sum_c [x_c(a) V_c^{\dagger} \ket*{a,v_c(a),x_c(a)}  + i\ y_c(a) W_c^{\dagger} \ket*{v_c(a),a,-y_c(a)}]  \\
= \sum_c [x_c(a) + i y_c(a)] \ket*{a,0,0}
\end{eqnarray}
which can be confirmed using the fact that $v_c(v_c(a)) =a$ and since we have modulo $2$ addition, i.e. $a\oplus a=0$.

\subsection{Erroneous Hamiltonian simulation}
One might wonder what would happen with errors in the Hamiltonian simulator.
For example, imagine that we simulate the target Hamiltonian with simulator such as a quantum computer, and this simulator introduces some random error terms. This will be an issue for as long as we
do not have fully error corrected quantum computers, or if we use methods like quantum density matrix exponentiation which can have
errors in the preparation of the state we want to exponentiate (see~\cite{lloyd2013quantum,kimmel2017hamiltonian}). For example, this is relevant for a method called sample-based Hamiltonian simulation~\cite{lloyd2013quantum,kimmel2017hamiltonian} where we perform the quantum simulation of a density matrix $\rho$, i.e.\ trace-$1$ Hermitian matrix, which can have some errors.\\
For a more in depth introduction and analysis see for example~\citep{cubitt2017universal}. Errors in the Hamiltonian simulator have been investigated
in depth and here, we only want to give the reader some tools to grasp how one could approach such a problem in this
small section. For more elaborate work on this we refer the reader to~\cite{cubitt2017universal}.\\
In the following we will use so called matrix Chernoff-type bounds. Let us recall some results from statistics:
\begin{theorem}[Bernstein~\cite{bernstein1927theory}]
	Let $\{X_i\}_{i=1}^n$ be random variables and let $X = \sum_i X_i$, such that $\mathbb E(X)$ is the expectation value and $\mathbb E(X^2) = \sigma^2$, $|X| \leq M$ is bounded and $X_i$ are independent, identically distributed copies. Then, for all $t > 0$, 
 $$ \mathbb P[|X - \mathbb E (X)| \geq t] \leq \exp{\left( \frac{-t^2}{2n \sigma^2 + \frac{4}{3} t M} \right)}.$$
\end{theorem}
This fundamental theorem in statistics is making use of the independence of the sampling process in order to obtain a concentration of the result in high probability. In order to provide bounds for Hamiltonian simulation, we will need to use matrix versions of these Chernoff-style results. We will thereby make certain assumptions about the matrix, such as for instance that it is bounded in norm, and that the matrix variance statistic - a quantity that is a generalization of the variance - has a certain value. We now state first the result and then prove it.

\begin{lemma}\emph{(Faulty Hamiltonian simulator)}
Hamiltonian simulation of a $N\times N$ Hamiltonian operator $\hat H = \sum_i \hat H_i$ with a faulty simulator
	that induces random error terms $\{\hat H_i^{err}\}$ (random matrices) with expectation value
$\mathbb E [\hat H^{err}_i] = 0 $ with bounded norm $\norm{\hat H^{err}_i}_2 \leq R$
for all $i$ and bounded matrix variance statistic $v(\hat H^{err} := \sum_i \hat H^{err}_i) = \max{ \left( \norm{\mathbb E (\hat H^{err} \hat H^{err \dagger})}_2,\norm{\mathbb E (\hat H^{err \dagger} \hat H^{err})}_2 \right) }$
in each term of the simulation can be simulated with an error less than $\mathcal{O}((1+t)\cdot \epsilon)$ using $m=\Ord{\max_i \norm{\hat H_i}_2^2 t^2/ \epsilon}$ steps with probability of at least
\begin{equation}
1 - 2 N e^{-\epsilon^2/(2 [ v(\hat H^{err}) + R \epsilon/3])}.
\end{equation}
\end{lemma}
\begin{proof}
To prove this we will need a theorem that was developed independently in the
two papers~\cite{tropp2012user,oliveira2009concentration}, which is a matrix extension of Bernstein's inequality.
Recall that the standard Bernstein inequality is a concentration bound for random variables with bounded norm, i.e.
it tells us that a sum of random variables  remains, with high probability, close to the expectation value.
This can be extended to matrices which are drawn form a certain distribution and have a given upper bound to the norm. We call a random matrix independent and centered if each entry is independently drawn from the other entries and the expectation of the matrix is the zero matrix. 
\begin{theorem}[Matrix Bernstein]\label{thm:matrixbernstein}
Let $X_1, . . . ,X_n$ be independent, centered random matrices with dimensionality $N \times N$, and assume that each one is uniformly bounded
\begin{equation}
\mathbb{E}(X_i) = 0, \; \text{and} \; \norm{X_i} \leq L \; \forall k=1,\ldots,n.
\end{equation}
We define the sum of these random matrices $X = \sum_{i=1}^n X_i$, with matrix variance statistic of the sum being
\begin{equation}
v(X) =: \max{ \left( \norm{\mathbb E (X X^{\dagger})},\norm{\mathbb E (X^{\dagger} X)} \right) },
\end{equation}
where $ \mathbb E (X X^{\dagger}) = \sum_{i=1}^n \mathbb E (X_i X_i^{\dagger})$.
Then
\begin{equation}
\mathbb{P} \left\lbrace \norm{X} \geq t \right\rbrace \leq 2N \cdot \exp( \frac{-t^2/2}{v(X) +L t/3}) \; \forall t \geq 0.
\end{equation}
Furthermore,
\begin{equation}
\mathbb{E}\left( \norm{X}\right) \leq \sqrt{2 v(X) \log{(2N)}} + \frac{1}{3} L  \log{(2N)}
\end{equation}
\end{theorem}
Let us then recall the error in the Hamiltonian simulation scheme from above.
\begin{eqnarray}
\norm{e^{-i(H_1+\ldots + H_k)t} - \left( e^{-iH_1 t/m} \cdots  e^{-iH_k t/m} \right)^m} = \\
\norm{e^{-i(H_1+\ldots + H_k)t \frac{m-1}{m}}H^{err} t + \mathcal{O}(\max_i{\norm{H_i}, \norm{H^{err}_i}}^2 t^2/m)} \leq \\
\norm{H^{err}} t + \mathcal{O}(\max_i{\norm{H_i}, \norm{H^{err}_i}}^2 t^2/m)
\end{eqnarray}
We then use the assumption that the matrix variance statistic is bounded and that all the $H^{err}_i$ are bounded in norm by $R$ in order
to be able to apply the above theorem.
Using theorem~\ref{thm:matrixbernstein} in order to probabilistically bound the first term, and
observing that for $m=\Ord{\max_i \norm{H_i}^2 t^2/ \epsilon}$ we can bound the second term by $\epsilon$,
assuming that $\max_i {H^{err}_i} < \min_i{H_i}$ and using theorem~\ref{thm:matrixbernstein} we achieve the proposed Lemma.
\end{proof}
For a more in depth analysis of Hamiltonian simulation and errors we refer the reader to~\cite{cubitt2017universal}.

\subsubsection{Modern methods for Hamiltonian simulation}

Modern approaches like fractional-query model~\cite{cleve2009efficient} are more complicated and we will not describe these here. However, these allow for tighter bounds and faster Hamiltonian simulation, with improved dependency on all parameters.\\
A conceptually different approach to Hamiltonian simulation based on Szegedy's quantum walk \cite{szegedy2004quantum} uses the notion of a discrete-time quantum walk that is closely related to any given time-independent Hamiltonian and applies phase estimation in order to simulate the evolution. This approach has the best known performance as a function of the sparsity $s$ and evolution time $t$ but has a worse $\epsilon$-dependency, i.e.\ in the error.
This improved method for Hamiltonian simulation scales as $\Ord{s \norm{H_{max}} t /\sqrt{\epsilon}}$ for a fixed Hamiltonian.\\
Other recent results based on different methods approach optimality, i.e. linear dependency in the parameters. The fact that such exponentials of the Hamiltonian can be easily performed on a quantum computer is essential to the HHL algorithm because it allows us to perform eigenvalue estimation, which we will discuss below.

\subsection{Quantum phase estimation} \label{subsec:phase est}

The goal of quantum phase estimation \cite{kitaev1995quantum} is to obtain a good approximation of an operator's eigenvalue given the associated eigenstate. Here, we consider a unitary operator $U$ acting on an $m$-qubit state, with a set of given eigenvectors $|\psi_i\rangle$ and associated unknown eigenvalues $\lambda_i$. For simplicity, we will consider a particular given eigenvector $|\psi\rangle$ with associated unknown eigenvalue $\lambda$.
This eigenvalue is a complex number and we can thus write $\lambda=e^{i2\pi \phi }$, where $0 \leq \phi \leq 1$ is referred to as the phase. Thus, we wish to determine a good $n$-bit approximation $\tilde{\phi}= 0.\tilde{\phi}_1 \ldots \tilde{\phi}_n$ of $\phi$, which will thus allow for a good $n$-bit approximation $\tilde{\lambda}$ of $\lambda$. This will be achieved by requiring $n$ ancillary qubits.  The intuition is to encode this approximation within relative phases of the qubits, as we previously saw with the QFT, see section \ref{subsec:QFT}.

In first instance, we note that $U^j |\psi \rangle= \lambda^j |\psi \rangle=e^{i2\pi \phi j}|\psi \rangle$.
To start with, we are given $n$ qubits prepared in the $|0\rangle^{\otimes n}$ state and an $m$-qubit quantum state intialised in the $|\psi\rangle $ state.
Next, $n$ Hadamard gates are applied to the $n$ qubits in the first register, which results in the state $\frac{1}{\sqrt{N}}\sum_x |x\rangle |\psi \rangle$, where $N=2^n$.

Then, in order to obtain the $n$-bit approximation of the eigenvalue, we will apply $n$ unitary operators to the state $\psi \rangle$. More specifically, the controlled-$U^{2^k}$ operator is applied, where $k$ is the control qubit, with $k=0$ for the first up to $k=n-1$ for the last qubit, as illustrated in Figure \ref{fig:qpe}. The sequence of controlled-$U^{2^k}$ operations used in phase estimation can be implemented efficiently using the technique of modular exponentiation, which is discussed in depth in~\cite{nielsen2002quantum}[Box 5.2, Ch. 5.3.1].
This results in the state:
\begin{equation}
\frac{1}{2^{n/2}} \sum_{y=0}^{2^n-1} \exp \big(i2\pi \phi y \big) |y\rangle |\psi\rangle.
\end{equation}
In the following, we wish to consider $\phi$ to $n$ bits of accuracy, where the $n$-bit approximation of $\phi$ is given by $\sum_s \phi_s 2^{-s}$. In all generality, the phase can be expressed as:
\begin{equation}\label{eq:phaseestphi}
\phi= \left( \frac{j}{2^n}+ \delta \right)
\end{equation}
where we have that $j=j_{n-1}\ldots j_0$ or alternatively $\frac{j}{2^n}=0.j_1\ldots j_n$. Substituting into~\ref{eq:phaseestphi} and dividing numerator and denominator by $2^n$ we have
\begin{equation}
\frac{1}{2^{n/2}} \sum_{y=0}^{2^n-1} \exp \big(i2\pi \sum_s \phi_s 2^{n-s} y /2^n \big) |y\rangle |\psi\rangle
\end{equation}
which it is easy to recognise is simply the QFT applied to the state $|\phi_1 \cdots \phi_n\rangle$.
Thus, performing a measurement in the Fourier basis gives the bit string $\phi_1 \cdots \phi_n$.
In the case where $\delta=0$, the measurement will yield the original state $|\phi_1 \cdots \phi_n\rangle$. If $\delta \neq 0$, then the measurement will have to be repeated, in which case an upper bound on the necessary number of repetitions can be obtained.

\begin{figure}[h!]
	\begin{center}
		$\Qcircuit @C=1em @R=1.0em @!R {
			\lstick{\ket*{0}} & \gate{H} & \ctrl{4} & \qw & \qw & \qw & \qw & \qw & \qw & \qw \\
			\lstick{\ket*{0}} & \qw & \qw & \gate{H} & \ctrl{3} & \qw & \qw & \qw & \qw & \qw \\
			\lstick{\ket*{0}} & \qw & \qw & \qw & \qw &\gate{H} & \ctrl{2} & \qw & \qw & \qw \\
			\lstick{\ket*{0}} & \qw & \qw & \qw & \qw & \qw & \qw &\gate{H} & \ctrl{1} & \qw  \\
			\lstick{\ket*{\psi }} & \qw & \gate{U^1} & \qw & \gate{U^2} & \qw & \gate{U^4} & \qw & \gate{U^8} & \qw \\
		}$
	\end{center}

	\caption{Quantum phase estimation circuit for the subroutine to prepare the correct state for $t=4$ ancillary qubits.}
	\label{fig:qpe}
\end{figure}
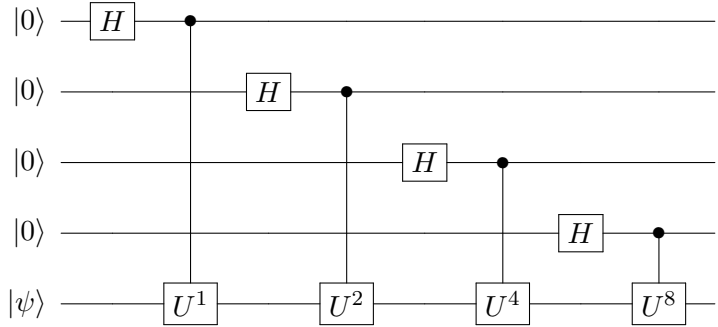

Thus, quantum phase estimation relies on preparing $n$ ancillary qubits in equal superposition, applying $n$ controlled unitary operators to the given eigenstate, and finally performing a measurement in the Fourier basis, as illustrated in Figure \ref{fig:qpe_two}.
\begin{figure*}[h!]
	\begin{center}
		$\Qcircuit @C=1em @R=1.0em @!R {
			\lstick{\ket*{0}^{\otimes n}} & {/} \qw & \ctrl{1} & \qw & \gate{QFT^{\dagger}} & \meter & \cw \\
			\lstick{\ket*{\psi}} & {/} \qw & \gate{U^j} & \qw & \qw & \qw & \lstick{\ket*{\psi}} \\
		}$
	\end{center}
	\caption{Quantum phase estimation circuit.}
	\label{fig:qpe_two}
\end{figure*}
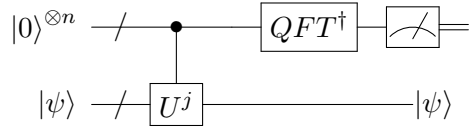

We can summarize the phase estimation procedure in the following Theorem:
\begin{theorem}[Phase estimation \cite{kitaev1995quantum}]
  \label{pest}
  Let unitary $U \ket*{v_j} = \exp(i \theta_j) \ket*{v_j}$ with $\theta_j \in [ - \pi ,\pi ]$ for $j \in [n]$. There is a quantum algorithm that transforms $\sum_{j \in [n]} \alpha_j \ket*{v_j} \mapsto \sum_{j \in [n]} \alpha_j \ket*{v_j} \ket*{\tilde{\theta}_j}$ such that $|\tilde{\theta}_{j} - \theta_{j} |\leq \epsilon$ for all $j\in [n]$ with probability $1-1/\poly{n}$ in time $\Ord{T_U \log{(n)} / \epsilon}$, where $T_U$ is the time to implement $U$.
\end{theorem}

Crucially, we do not need to be given access to the actual eigenvector $|v_j\rangle$, as this mapping can be applied to a superposition of the eigenvectors. Indeed, any quantum state $\ket*{u}$ can be decomposed in an arbitrary orthonormal basis, such as for instance in the operator eigenbasis $\ket*{v_j}$:
\begin{equation}
\ket*{u} = \sum_j \bra*{u}\ket*{v_j} \ket*{v_j} = \sum_j \alpha_j \ket*{v_j}, 
\end{equation}

where $\alpha_j = \bra*{u}\ket*{v_j}$. Hence the quantum phase estimation procedure can be applied to an arbitrary state $\ket*{u}$, which, as it is just a matter of representation, be directly applied to the operator eigenbasis, apply this mapping in the specific (eigen-)basis without knowing the actual basis:
\begin{equation}
\ket*{0}\ket*{u} =\ket*{0} \sum_j  \bra*{u}\ket*{v_j} \ket*{u_j} = \sum_j \alpha_j \ket*{\tilde{\theta_{j}}} \ket*{v_j}.
\end{equation}

Next, we consider the process of phase kickback, where we query an oracle and encode some information pertaining to it as a relative phase.
\subsection{Phase kickback}\label{sec:phase_kickback}

In the following, we consider a Boolean function $f: \{0,1\}^n \to \{0,1\}^m$. The function is queried via an oracle $\mathcal{O}_f$, and so
\begin{equation}
|x\rangle |q\rangle \to |x\rangle |q \oplus f(x)\rangle,
\end{equation}
where $|x\rangle$ is an input state and where $|q\rangle $ is an ancillary register. This operation can be implemented by a unitary circuit $U_f$.

If the ancillary qubit is in the state $|-\rangle=\frac{|0\rangle - |1\rangle}{\sqrt{2}}$, then by applying the oracle we obtain the state:
\begin{equation} \label{equ:ffbar}
|x\rangle \frac{|f(x)\rangle - |\overline{f(x)}\rangle}{\sqrt{2}},
\end{equation}
where $\overline{\cdot}$ denotes the negated bits.
By considering the cases where the output $f(x)$ is either a $1$ or a $0$, it can easily be shown that this is equivalent to the state:
\begin{equation}\label{equ:kick}
(-1)^{f(x)} |x\rangle|-\rangle.
\end{equation}
Thus, inputs which evaluate to $1$ acquire a relative phase. This process is referred to as phase kickback. Finally, we note that we have assumed that the function can be classically computed in polynomial time, i.e. this is not an expensive step and no complexity is hidden in this call.

\subsection{Amplitude amplification}\label{subsec:aa}

Amplitude amplification~\cite{brassard2002quantum} is an extension of Grover's search algorithm~\cite{Grover2002}. Here, we shall present Grover's search algorithm, then show how this leads to amplitude amplification.

We are given a set containing $N$ elements $\{1, \ldots, N\}$ and the goal is to find a particular element of the set which is marked. This may be modeled by the function $f:\{1, \ldots, N\} \to \{0,1\}$ such that there exists uniquely one item $a \in \{1, \ldots, N\}$ satisfying $f(a)=1$. Otherwise we have that $f(x)=0$ for all $x \in \{1, \ldots, N\} \setminus \{ a \}$. Furthermore, we can evaluate $f(x)$ using an oracle that can be queried in superposition with unit computational cost, see section~\ref{sec:phase_kickback}.
Grover's algorithm is for our purposes more appropriately seen as an algorithm for finding the (unique) root of the Boolean function $\neg f(x)$, i.e. replace our set $\{1, \ldots, N\}$ with the bitstring representations of $\{0,\ldots, N-1\}$. Then  $a$ is the bitstring such that $\neg f(a) = 0$, or the root of $\neg f(x)$.
We shall show later how Grover's algorithm generalises to finding all roots of a Boolean function $f(x)$.

Classically, finding $a$ requires $\mathcal{O}(N)$ oracle queries, where in the worst case we evaluate $f$ on all elements in the set.
In contrast, Grover's algorithm achieves a quadratic speed-up requiring $\mathcal{O}(\sqrt{N})$ oracle queries. Furthermore, it has been shown that there is a lower bound $\Omega(\sqrt{N})$ on the number of queries, that is, this is the optimal scaling~\cite{Dohotaru2001}.

Grover's algorithm uses the phase kickback technique previously discussed in section~\ref{sec:phase_kickback}, which associates a relative phase with the marked item. In order to do this, $n$ qubits are prepared in the $|0\rangle$ state, and by the application of the operator $H^{\otimes n}$, the uniform superposition state $|\psi_0\rangle =\frac{1}{\sqrt{N}} \sum_x |x\rangle $ is obtained.
An ancillary qubit is prepared in the $|-\rangle $ state. By applying the phase kickback protocol, from Eq.~\eqref{equ:ffbar}  and \eqref{equ:kick}, we obtain $\frac{1}{\sqrt{N}}\sum_x (-1)^{f(x)} |x\rangle |-\rangle $.
We now discard the second register.
The basis state corresponding to the marked item has now acquired a relative phase, and we are interested in the remaining state $|\psi \rangle = \frac{1}{\sqrt{N}}\sum_x (-1)^{f(x)}|x\rangle$.
We denote by $Z_f$ the unitary operator that queries the oracle and applies the phase $(-1)^{f(x)}$ to each computational basis state $\ket*{x}$.

Ultimately, a measurement in the computational basis is to be performed. Ideally, this would yield the computational basis state corresponding to the marked item with high probability. How can this be achieved?
The idea is to apply a unitary operator to the input which will dampen the coefficients associated with unmarked items, and strengthen the coefficients corresponding to the marked item. For this to be practically achievable, we need to check that this operator has an efficient implementation.

The Grover operator, to be defined shortly, does precisely this in an efficient manner. The state $|\psi \rangle $ can be expressed as a linear combination of $2^n$ basis terms $|x\rangle $, each with a corresponding complex coefficient $\alpha_x=(-1)^{f(x)}$. One of these corresponds to the marked item $a$. The idea is to strengthen the coefficient $\alpha_a$ whilst simultaneously weakening coefficients $\alpha_k$, where $1\leq k \leq N$ and $k\neq a$.

The unitary operator which achieves this is called the \emph{Grover operator} and is given by $G:= D Z_f$, where $D=2|\psi_0\rangle \langle \psi_0| - \mathbb{I}$ is called the diffusion operator. It is not a priori obvious what this operator does or why one would wish to apply it. There are two arguments which can shed insight into this, one algebraic and one geometric, which we next introduce.

\emph{Algebraic argument.} First, we consider an arbitrary state $|\psi \rangle = \sum_x \alpha_x |x \rangle$. It is easy to see that applying the diffusion operator $D$ results in the state $\sum_x (2 \langle \alpha \rangle -  \alpha_x) |x\rangle $, where $\langle \alpha \rangle =\sum_k \frac{\alpha_k}{N}$ is defined as the mean  value of the coefficients.
Thus the new coefficient corresponding to $\ket*{x}$ is given by $\langle \alpha \rangle  + (\langle \alpha \rangle -\alpha_x)$.
Before the first application of $D$, the mean is given by $\langle \alpha \rangle = \frac{(N-2)/N}{\sqrt{N}}$. Here, we have that the positive coefficients will be dampened and approach to zero, whereas the negative coefficients will be magnified and become positive. This is inversion about the mean. Next, the item is marked with a negative sign by $Z_f$ in order for the inversion about the mean to be applied in the next step. As this process repeats, the unmarked items' coefficients will tend to zero, whereas the marked coefficient goes towards one \cite{Whaley09}.

\emph{Geometric argument.} Strengthening means bringing the initial state closer to the state $|a\rangle $, whilst preserving the norm. Visually, this can be seen as performing a rotation of angle $2\theta$ towards $|a\rangle$ in the plane defined by $\operatorname{span}\{\ket*{\psi_0},\ket*{a}\}$, where $\theta = \arcsin(\bra*{\phi_0}\ket*{a})$.
Let $|a^{\perp}\rangle $ be the orthogonal complement of the marked item state in this subspace. Then $\theta$ is the angle between the equal superposition state $\ket*{\psi_0}$ and the orthogonal complement $\ket*{a^\perp}$.

A rotation by an angle $2\theta$ can be implemented via two successive reflections: one through $|a^\perp \rangle $ and one through $|\psi_0\rangle$. A reflection about an axis mean that any component of a particular vector orthogonal to the axis acquire a negative phase, and the component along the axis remains invariant. Let $|r \rangle$ be the axis of reflection, and let $|\psi \rangle $ be the state to reflect. After reflection, we should have $|\psi'\rangle = \alpha_r |r\rangle -\sum_{i \neq r}\alpha_i |i\rangle $.
This can be expressed as $-|\psi\rangle +2\alpha_r |r\rangle $, which corresponds to application of the operator $\mathbb{I} - 2|r\rangle \langle r|$, up to a global phase. Thus, this gives us the reflection operation. In the case when we reflect about the axis $|\psi_0\rangle $ we recover the diffusion operator.
In the case where we reflect about $\ket*{a^\perp}$ we recover $Z_f$.
Thus, in order to implement the rotation by an angle $2\theta$, we first apply a reflection about $|a^\perp \rangle$, followed by a reflection about $|\psi_0 \rangle$.
We see an illustration of this in Figure~\ref{fig:grover}.
\begin{figure}
	\centering
	\includegraphics[width=0.45\textwidth]{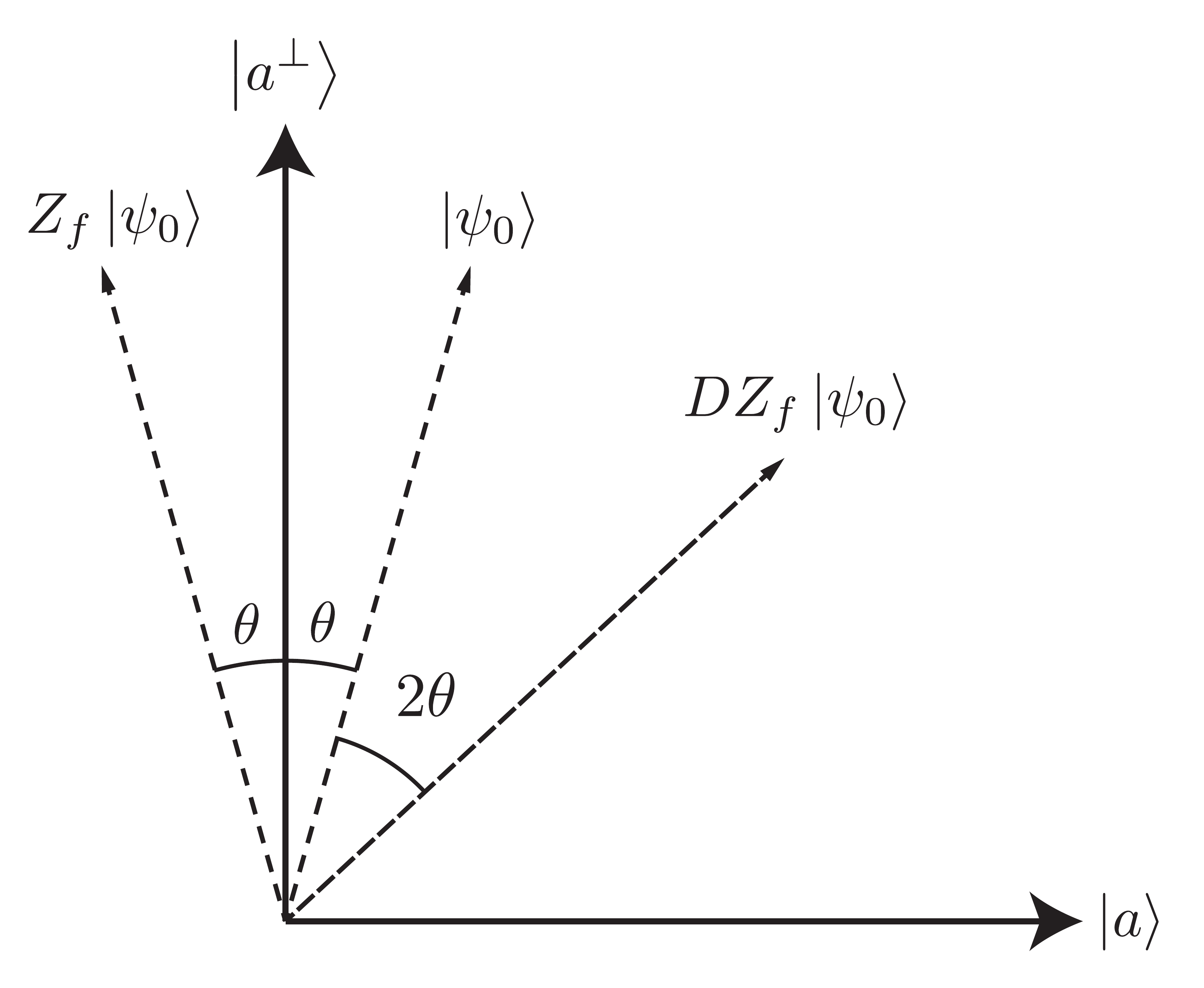}
	\caption{An illustration of the action of one iteration of the Grover operator $G = D Z_f$ on the initial state $\ket*{\psi_0}$.}
	\label{fig:grover}
\end{figure}

We wish to have a number of rotations $k^*$ such that a measurement in the computational basis yields $\ket*{a}$ with high probability. This probability, after $k$ rotations, is given by $\abs{\bra*{a} G^k \ket*{\psi_0} }^2 = \sin^2 ((2k+1)\theta)$.
To have this close to one, we need $(2k^*+1)\theta \approx \pi/2$, which we have when $k^* = \frac{\pi/2}{2\theta}$.
We have that $\langle \psi_0 |a\rangle = \sin \theta$ but also $\langle \psi_0 |a\rangle =\frac{1}{\sqrt{N}}$. Using the small angle approximation, we have that $\sin \theta \approx \theta $ and so $\theta \approx \frac{1}{\sqrt{N}}$.
So, $k^*$ and thus Grover's algorithm scales as $\mathcal{O}(\sqrt{N})$.

This idea can be generalized to the case where we there are now $M$ marked items out of a total of $N$ items, where $1 \leq M \leq N$. First, we construct the states $|\alpha \rangle = \frac{1}{\sqrt{M}} \sum_{x_m} |x_m\rangle $ and the state over unmarked items $|\beta \rangle =\frac{1}{\sqrt{N-M}} |\sum_{x_u} |x_u\rangle $, where $\{x_m\}$ are the marked elements and $\{x_u\}$ are the unmarked elements.
Once again, the goal will be to obtain a marked state with high probability when a computational basis measurement is performed. That is, the coefficients corresponding to the marked item will be boosted whereas those corresponding to unmarked items  dampened.
Here, we can think of this as splitting the total space $\mathcal{H}$ into two subspaces, the \emph{good} subspace and the \emph{bad} subspace. The good subspace corresponds to the marked basis states, i.e. those containing relevant information, which we would wish to obtain upon performing a measurement, and thus boost their corresponding amplitudes. We have that $\mathcal{H}=\mathcal{H}_g \oplus \mathcal{H}_b$, and thus we can write any arbitrary quantum state as
\begin{equation}
|\psi\rangle = \alpha_g |\phi_g\rangle + \alpha_b |\phi_b\rangle,
\end{equation}
where $|\phi_g \rangle$ and $|\phi_b\rangle$ are two orthogonal states, and where we have that $|\alpha_g|^2 + |\alpha_b|^2=1$. In the context of Grover search, we have $|\phi_g\rangle = \frac{1}{\sqrt{M}}\sum_{x_m} |x_m\rangle$, and $|\phi_b\rangle = \frac{1}{\sqrt{N-M}}\sum_{x_u}|x_u\rangle $, and thus we can write the complete state as
\begin{equation}
|\psi \rangle = \sqrt{\frac{M}{N}} |\phi_g\rangle  + \sqrt{\frac{N-M}{N}}|\phi_b\rangle,
\end{equation}
i.e. $\alpha_g= \sqrt{\frac{M}{N}}$ and $\alpha_b=\sqrt{\frac{N-M}{N}}$. The goal is to now, via the application of a unitary operator, amplify the coefficient $\alpha_g$ whilst weakening $\alpha_b$.  Geometrically, this means that the quantum state will be rotated towards the good subspace.

Amplitude amplification is the process of applying Grover's algorithm to tasks where we have an oracle for a function $f$ and need to sample $x$ from the `good' subset of strings, i.e. $G = \{x \mid f(x) = 1\}$, more efficiently than with a classical algorithm.
By querying the oracle $\Ord{\sqrt{\frac{N}{M}}}$ times the amplitude of the `good' subset of strings is then close to unity and the `bad' amplitude is close to zero.
Stated more generally, this gives us
\begin{lemma}\label{lemma:amplitude_amplification}
\emph{(Amplitude Amplification)} Suppose we have an algorithm $\mathcal{A}$ that succeeds with probability $\epsilon$.
Using amplitude amplification, we can take $\Ord{1/\sqrt{\epsilon}}$ repetitions of $\mathcal{A}$ to yield an algorithm $\mathcal{A}'$ with success probability arbitrarily close to one.
\end{lemma}
Note that $\mathcal{A}$ can be a classical or quantum algorithm and that we obtain the lemma by taking the `good' subspace as successful bitstring outputs of $\mathcal{A}$ and the `bad' subspace as the unsuccessful outputs.
One can do this by simply appending a bit at the end of the evaluation of the result of $\mathcal{A}$ with `1' for success and `0' for failure.



\subsection{The uncompute trick} \label{subsec:uncompute}

The uncompute trick is a commonly used technique in quantum algorithms for carrying out a computation, then retrieving the initial state. For this subsection, we take much of the presentation from the discussion at~\cite{uncomputeStackExchange}.
Previously, in section \ref{sec:phase_kickback}, we saw that given a boolean function $f:\{0,1\}^n \to \{0,1\}^m$, there exists an oracle $\mathcal{O}_f$ acting as $\ket*{x}\ket*{0} \mapsto \ket*{x}\ket*{f(x)}$, which allows for the mapping $|x\rangle \to |f(x)\rangle $ to be enacted. Yet, we did not discuss how this unitary operator could be efficiently implemented. In particular, from the no-deleting theorem~\cite{KumarPati2000} there is no single-qubit unitary operator that sets an arbitrary qubit state to $\ket*{0}$.
Indeed, for most algorithms we have the map $\ket*{x}\ket*{0}\ket*{0} \mapsto \ket*{x}\ket*{g(x)}\ket*{f(x)}$, where $\ket*{g(x)}$ is a \emph{garbage state} in a \emph{working register} we have used along the way.
Here, we wish to return the garbage state to $\ket*{0}$ in the working register, so that it doesn't disrupt future computations after we discard it. 
More precisely, we assume in general for any quantum algorithm that the working register is initialised in state $\ket*{0}$.
If a previous computation has left the working register state as $\ket*{g(x)}$ the next algorithm to use the working register will output incorrect results.
Even worse, in the case that the mapping $\ket*{x}\ket*{0}\ket*{0} \mapsto \ket*{x}\ket*{g(x)}\ket*{f(x)}$ isn't perfect, the working register and output register are \emph{entangled}. That is, we have some state $\ket*{x}\ket*{0}\ket*{0} \mapsto \sum_{y \in \{0,1\}^m}\alpha_y\ket*{x} \ket*{g(y)}\ket*{f(y)}$ and operations on the garbage register will affect the output register.
Moreover, we wish to keep the input vector $\ket*{x}$. Let us see how this is achieved.

Suppose we have a unitary operation $U_f$ that takes an input $\ket*{x}\ket*{0}\ket*{0}$ and produces as output the state $\ket*{\phi} = \sum_y \alpha_y \ket*{x} \ket*{y} \ket*{f_y}$, where $\ket*{x}$ is in the input register, $\ket*{y}$ is in the working space and $\ket*{f_y}$ is in the output register, and $f_y\in\{ 0,1 \}^m$, the output of the computation. Ideally, we would have that $f_y = f(x)$, and in this case the output is given by the state $\sum_y \alpha_y \ket*{x} \ket*{y} \ket*{f(x)}$, i.e. $\ket*{g(x)} = \sum_y \alpha_y \ket*{y}$.
If an additional computational register is added, and apply the CNOT gate controlled on the third register, giving $\sum_y \alpha_y \ket*{x} \ket*{y} \ket*{f(x)} \ket*{f(x)}$.
Then, if the inverse operator $U_f^{-1}$ is applied to the first three registers, then the state evolves to $\ket*{x}\ket*{0}\ket*{0}\ket*{f(x)}$. Finally, applying a SWAP operator\footnote{The SWAP operator, as the name suggests, swaps the state between two registers on the same number of qubits, i.e., $\operatorname{SWAP}: \ket*{x}\ket*{y} \mapsto \ket*{y}\ket*{x}$.} on the last two registers produces the state $\ket*{x}\ket*{0}\ket*{f(x)}\ket*{0}$.
We can now safely discard the working register as well as the  ancillary register, leaving us $\ket*{x}\ket*{f(x)}$ as desired.

This action of applying $U_f$, appending an additional register, applying the CNOT, then applying $U_f^{-1}$ followed by a SWAP is known as the \emph{uncompute trick}.
It allows one to record the final state of a quantum computation and then reuse the working register for some other task.

By assuming $f_y = f(x)$ in the discussion above, we assumed that the action of $U_f$ was error-free, which is a highly unrealistic scenario. We now assume an error probability of $\epsilon$, that is,
\begin{equation}
	\sum_{y \mid f_y = f(x)} \abs{\alpha_y}^2 = 1- \epsilon.
\end{equation}
The output state after the uncompute operation, that is -- applying the oracle, appending the ancillary register and applying the CNOT -- is given by $\ket*{\phi'}=\sum_y \alpha_y \ket*{x} \ket*{y} \ket*{f_y} \ket*{f_y}$, which is different to the ideal final state $\ket*{\phi'} = \sum_y  \alpha_y \ket*{x} \ket*{y} \ket*{f(x)} \ket*{f(x)}$. The inner product gives
\begin{equation}
	\begin{aligned}
		\bra*{\phi'}\ket*{\phi} &= \sum_{y'}\sum_y  \alpha^*_{y'}\alpha_y \braket*{x} \bra*{y'}\ket*{y} \bra*{f_{y'}}\ket*{f(x)} \bra*{f_{y'}}\ket*{f(x)} \\
		&=
		\sum_y \abs{\alpha_y}^2 \delta^{f_y}_{f(x)} =
		\sum_{y \mid f_y = f(x)} \abs{\alpha_y}^2 =
		1- \epsilon.
	\end{aligned}
\end{equation}
Finally, the inverse operator $U_f^{-1}$ is applied, which preserves the inner product, and thus we have that $\left\langle U_f^{-1}\ket*{\phi'} ,  U_f^{-1}\ket*{\phi} \right\rangle = 1-\epsilon$. Thus, we have that if the unitary $U_f$ acts to within error $\epsilon$, the mapping $\ket*{x}\ket*{0} \mapsto \ket*{x}\ket*{f(x)}$ can be enacted to within error $\epsilon$.

\subsection{Quantum RAM} \label{subsec:qram}

\begin{figure}
	\centering
	\includegraphics[width=\textwidth]{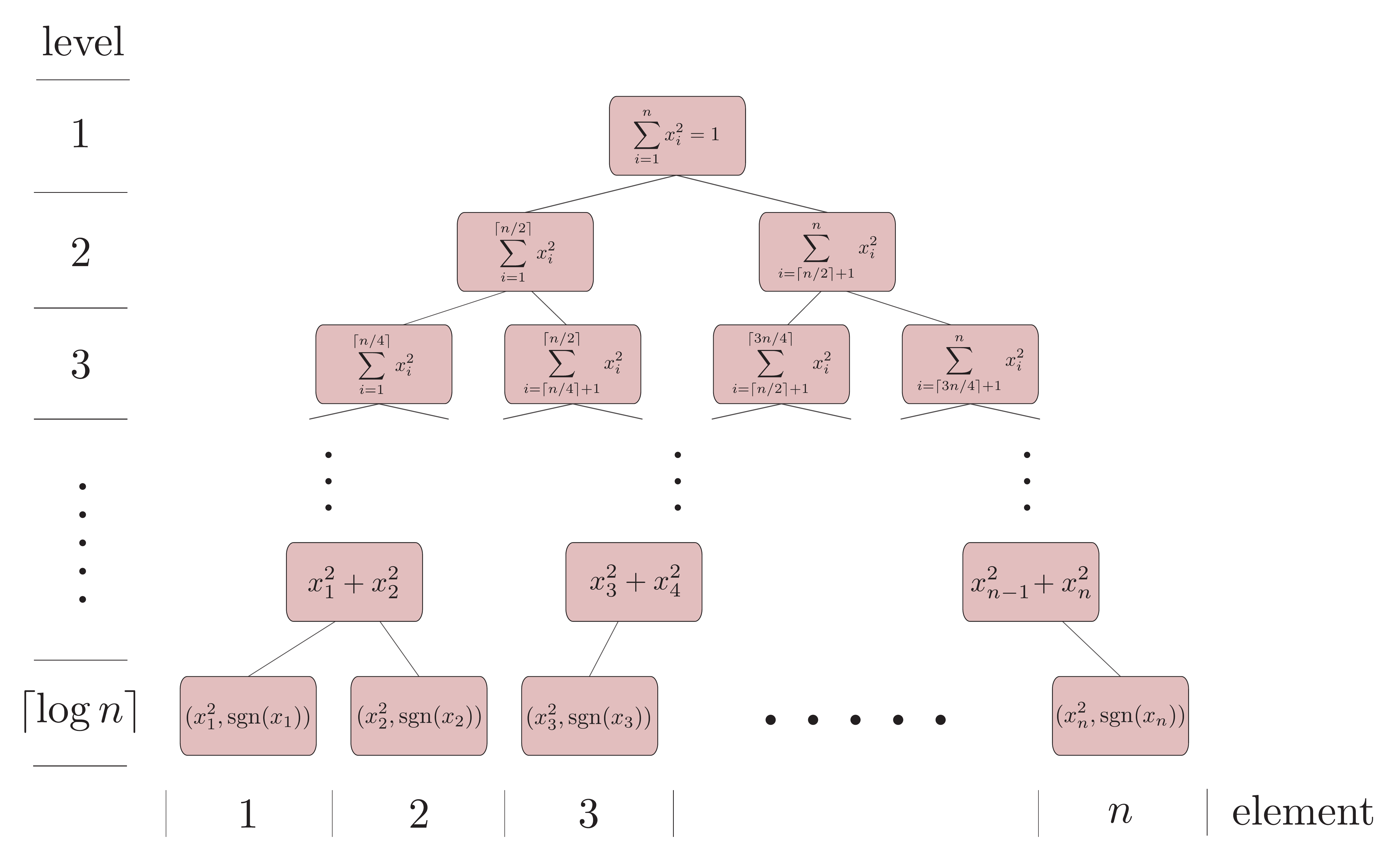}
	\caption{An illustration of a classical data structure $B_{\vb{x}}$ that, when equipped with quantum access, constitutes a qRAM, storing a vector $\vb{x}\in \mathbb{R}^n,\ \norm{\vb{x}}_2 = 1$.
	The vector $\vb{x}$ is stored in the binary tree shown.
	To each element of $\vb{x}$, $x_i$, there is a leaf of the tree.
	Each leaf contains the squared amplitude of the element $x_i$ and its sign $\operatorname{sgn}(x_i)$.
	Every other node contains the sum of its child nodes (ignoring the $\operatorname{sgn}(x_i)$ terms for the $(\lceil \log n \rceil - 1)^{\text{th}}$ level).
	To load the vector, we move through the tree from the root node, appending relevant qubits to the computational register where necessary and rotating conditioned on the values stored in the corresponding nodes.
	This procedure is detailed in Algorithm~\ref{alg:qRAM}.}
	\label{fig:qRAM}
\end{figure}

In the HHL algorithm, a real-valued vector has to be manipulated by the quantum computer. That is, it must be loaded into a computational register, where the elements of the vector will be encoded in the amplitudes of a quantum state. As quantum states are normalised, these amplitudes will be the elements of the vector scaled by the norm of the vector.

More precisely, let $\vb{x} \in \mathbb{R}^N$ be the input vector and let $\ket*{x}:=  \norm{\vb{x}}_2^{-1} \textstyle{\sum_{i=1}^N } x_i \ket*{i} $ be a quantum state on $\lceil \log N \rceil$ qubits.
We would like to have an operation $\mathcal{R}$, that takes some finite precision representation of $\vb{x}$,which we call $\tilde{\vb{x}}$, and outputs the state $\ket*{x}$, along with $\norm{\vb{x}}_2$ i.e.
\begin{equation}
	\mathcal{R}: \tilde{\vb{x}} \mapsto \ket*{\widetilde{\norm{\vb{x}}_2}} \ket*{x},
\end{equation}
where $\ket*{\widetilde{\norm{\vb{x}}_2}}$ is a computational basis state encoding the value of $\norm{\vb{x}}_2$ to some finite precision (denoted by the tilde on top).

The operation $\mathcal{R}$ is known as a \emph{quantum RAM} (qRAM) and is non-trivial to implement. For instance, a linear system can be classically solved in polynomial time. Thus, a necessary condition to obtain a quantum exponential speed-up is for $\mathcal{R}$ to run in time at most polylogarithmic in $N$. This statement holds if we use a qRAM for any task that is polynomial-time computable classically.

Here, we present a memory structure used to implement it, taken from the PhD thesis of Prakash~\cite{PrakashThesis} along with the state preparation procedure of Grover and Rudolph~\cite{Grover2002}. 
We shall say that this is a \emph{classical data structure with quantum access}, in that, it stores classical information, but can be accessed in quantum superposition.
This can be implemented with an ordinary quantum circuit. 

First, we consider the following lemma:
\begin{lemma}\emph{(Controlled rotation)}\label{lem:cosrotation}
	Let $\theta \in \mathbb{R}$ and let $\tilde{\theta}$ be its $d$-bit finite precision representation.
	Then there is a unitary, $U_{\theta}$, that acts as
	\begin{equation}\label{eq:controlledrotation}
		U_{\theta} : \ket*{\tilde{\theta}}\ket*{0}  \mapsto \ket*{\tilde{\theta}} \qty( \cos \tilde{\theta} \ket*{0} + \sin \tilde{\theta} \ket*{1} ).
	\end{equation}
\end{lemma}
\begin{proof}
	Let
	\begin{equation}
		U_{\theta} = \sum_{\tilde{\theta} \in \{0,1\}^{ d } } \ketbra*{\tilde{\theta}} \otimes \exp(- i \tilde{\theta} \sigma_y ),
	\end{equation}
	where $\sigma_y = \smqty(0 & -i \\ i & 0)$ is the Pauli $Y$ matrix.
	The matrix $\exp(- i \theta \sigma_y )$ is given by $\exp(- i \theta \sigma_y ) = \smqty(\cos\theta & - \sin\theta \\ \sin\theta & \cos\theta)$. Finally, applying $U_\theta$ to $\ket*{\tilde{\theta}}\ket*{0} $ yields the result in Eq.~\eqref{eq:controlledrotation}.
\end{proof}
We should note that the unitary $U_\theta$ can be implemented in $\Ord{d}$ gates, where $d$ is the number of bits representing $\tilde{\theta}$, using one rotation controlled on each qubit of the representation, with the angle of rotation for successive bits cut in half.

The vector $\vb{x}$ is stored in a binary tree $B_{\vb{x}}$, as illustrated in Figure~\ref{fig:qRAM}. $\tilde{\norm{\vb{x}}_2}$ can be easily loaded into a computational register using CNOT gates and we thus henceforth assume that $\norm{\vb{x}}_2 = 1$ for simplicity. In addition, we assume a finite precision representation throughout.
Each vertex holds the sum of its two children, apart from the leaves, where the $i^{\text{th}}$ leaf holds the squared amplitude of the $i^{\text{th}}$ element of $\vb{x_i}$ along with its sign, i.e. $(x^2_i, \operatorname{sgn}(x_i))$.
Moreover, observe that each of these values must be precomputed then stored. The structure $B_{\vb{x}}$ has $\Ord{N}$ nodes.
We don't take the cost of this preparation into account when computing the time required to prepare the vector $\ket*{\widetilde{\norm{\vb{x}}_2}} \ket*{x}$ as this can be done once, then many copies of the state $\ket*{x}$ can be created.

The procedure for preparing $\ket*{x}$ from $\vb{x}$ is shown in Algorithm~\ref{alg:qRAM}.
An intuitive way to think about it is the following:
the aim is to associate an amplitude $x_i$ to each basis state $\ket*{i}$, for a given $\ket*{x}$. This is achieved  by depth-first-traversal of a binary tree, whose leaves correspond to the basis vectors $\ket*{i}$.
We add qubits to the working register, then rotate them so as to assign the appropriate `amplitude mass' to each subtree of a given node $u$.
We use the term amplitude mass in the same sense as probability mass.
In this way, once we reach the leaves the amplitudes are as desired.
For the node $u$, the precomputed values at each child of $u$ partition this amplitude mass allocated to $u$ and are used as a control.
This whole procedure is encompassed by the $\texttt{processNode}$ subroutine.
The sign of each element of $\vb{x}$ is then easily handled by the $\texttt{processSign}$ subroutine.

\newcommand{\aldent}{\hspace{2em}}
\begin{algorithm}[!htbp]
	\caption{Load vector from Quantum RAM}
	\label{alg:qRAM}
	\begin{enumerate}
		\item Input: vector $\vb{x}\in \mathbb{R}^n$, $\norm{\vb{x}}_2 = 1$, loaded into classical binary tree $B_{\vb{x}}$ (see Figure~\ref{fig:qRAM}).

		Output: state vector $\ket*{x}$.
		\item Initialise $\lceil \log n \rceil$ qubits to $\ket*{0}^{\otimes \lceil \log n \rceil}$, label them $q_1, \ldots, q_{\lceil \log n \rceil}$.
		\item Let $v$ be the root node of $B_{\vb{x}}$.
		Then execute $\texttt{processNode}(v)$.
		\vspace{0.2cm}
		\item $\texttt{processNode}(\texttt{vertex}\ u)\, \{$
		\item[] \aldent Let $u_l$ and $u_r$ be the left and right child nodes of $u$ respectively and let $k$ be the level
		\item[] \aldent of the tree $u$ is in.
		\item[] \aldent $\theta \leftarrow \arccos(\sqrt{\texttt{value}(u_l) / \texttt{value}(u)}) = \arcsin(\sqrt{\texttt{value}(u_r) / \texttt{value}(u)})$.
		\item[] \aldent Perform the controlled rotation $q_k \leftarrow \cos\theta \ket*{0} + \sin\theta \ket*{1}$ (see Lemma~\ref{lem:cosrotation}) conditioned
		\item[] \aldent on the qubits $q_1 \ldots q_{k-1}$ being equal to the binary representation of the vertex $u$.
		\item[] \aldent $\texttt{if}(u_l,u_r \text{ are leaves} )\{$
		\item[] \aldent\aldent $q_k \leftarrow \texttt{processSign}(q_k,\ u_l, \ u_r )$.
		\item[] \aldent\aldent \texttt{return}
		\item[] \aldent $\}$
		\item[] \aldent $\texttt{else}\{$
		\item[] \aldent\aldent $\texttt{processNode}(u_l)$
		\item[] \aldent\aldent $\texttt{processNode}(u_r)$
		\item[] \aldent $\}$
		\item[] $\}$
		\vspace{0.2cm}
		\item $\texttt{processSign}(\texttt{qubit}\ q,\ \texttt{leaf}\  u_l, \ \texttt{leaf}\  u_r )\{$
		\item[] \aldent $\texttt{switch} \{$
		\item[] \aldent\aldent $\texttt{case}(\operatorname{sgn}(u_l) = +,\  \operatorname{sgn}(u_r) = +)\{\ \texttt{return}\ q \ \}$
		\item[] \aldent\aldent $\texttt{case}(\operatorname{sgn}(u_l) = +,\  \operatorname{sgn}(u_r) = -)\{\ \texttt{return}\ \sigma_z q \ \}$
		\item[] \aldent\aldent $\texttt{case}(\operatorname{sgn}(u_l) = -,\  \operatorname{sgn}(u_r) = +)\{\ \texttt{return}\ -\sigma_z q \ \}$
		\item[] \aldent\aldent $\texttt{case}(\operatorname{sgn}(u_l) = -,\  \operatorname{sgn}(u_r) = -)\{\ \texttt{return}\ - q \ \}$
		\item[] \aldent $\}$
		\item[] $\}$
	\end{enumerate}
\end{algorithm}

All we need now is to be sure that Algorithm~\ref{alg:qRAM} does what it is supposed to do and runs in polylogarithmic time.
\begin{theorem}\emph{(qRAM loading)} \label{thm:data}
	Suppose $\vb{x} \in \mathbb{R}^N$ , $\norm{\vb{x}}_2 = 1$, and we have stored $\vb{x}$ in the data structure $B_{\vb{x}}$, as previously described, see Figure~\ref{fig:qRAM}.
	Furthermore, let $\ket*{x}= \textstyle{\sum_{i=1}^N } x_i \ket*{i}$.
	Then, Algorithm~\ref{alg:qRAM} prepares the state $\ket*{x}$ in time $\Ord{\log N}$.
\end{theorem}
\begin{proof}
	First we verify that the state $\ket*{x}$ is prepared by Algorithm~\ref{alg:qRAM}. Let the state output by the algorithm be $\ket*{\tilde{x}}$ and its $i^{\text{th}}$ amplitude $\tilde{x}_i$.
	We can think of the value of $\tilde{x}_i$ as being computed by walking from the root node of $B_{\vb{x}}$ to the leaf $i$ along a path $P_i$, multiplying by a relevant factor at every intermediate node followed by the sign of $x_i$ right at the end.
	Now $P_i=(u_1, u_2, \ldots, u_{\lceil \log n \rceil} )$.
	The factor we multiply by at each intermediate node $u_k$ is $\sqrt{\texttt{value}(u_{k}) / \texttt{value}(u_{k-1})}$.
	We thus have
	\begin{equation}
			\tilde{x}_i =  \prod_{k=2}^{\lceil \log n \rceil} \sqrt{\frac{\texttt{value}(u_{k})}{\texttt{value}(u_{k-1})}} \operatorname{sgn}(x_i) =
			\sqrt{\frac{\texttt{value}(u_{\lceil \log n \rceil})}{\texttt{value}(u_{1})}} \operatorname{sgn}(x_i) = \sqrt{\frac{x_i^2}{1}} \operatorname{sgn}(x) = x_i.
	\end{equation}
	Since this argument works for any $i \in [n]$, we have that $\ket*{\tilde{x}} = \ket*{x}$, as desired.
	For the runtime, there are $2^{k}$ rotations executed at the $k^{\text{th}}$ level of the tree, apart from the last level where there are none.
	For a given level these rotations can be executed in parallel as they are all controlled operations on the same qubit, conditioned on different bit-string values of a shared register.
	To see this, let $U_x$ be a single qubit rotation conditioned on a bitstring $x\in \{0,1\}^k$.
	Then the unitary $\bigoplus_{x\in \{0,1\}^k}U_x$ applied to $\ket*{y} \otimes \ket*{q}$ achieves the desired parallel operation on the single qubit $\ket*{q}$, where $\ket*{y} = \sum_{x\in \{0,1\}^k} \alpha_x \ket*{x}$ is some superposition over bitstrings.
	
	Since there are $\lceil \log n \rceil$ levels, we have the runtime $\Ord{\log n}$, assuming a constant cost for the rotations.
\end{proof}

We now consider the case where we can't execute the rotations $\tilde{R}$ perfectly. Presume there is some constant error, $\epsilon$, on each rotation $\tilde{R}$ such that $\norm*{\tilde{R} \ket*{\psi} - \exp(- i \theta \sigma_y) \ket*{\psi} } \leq \epsilon $ for any qubit state $\ket*{\psi} \in \mathbb{C}^2$.
There are $\sum_{k=1}^{\lceil \log n \rceil - 1} 2^k = \Ord{n}$ single-qubit rotations that are used in total.
We claim that the errors are additive, so we have that $\norm{\ket*{\tilde{x}} - \ket*{x}}_2 = \Ord{ n \epsilon }$.
Thus the error on the rotations needs to scale as $\epsilon = \Ord{1/n}$ for the state $\ket*{\tilde{x}}$ to be prepared to constant precision.
It remains to justify the claim.
We have
\begin{cla}\emph{(Additive unitary error)} \label{cla:close}
	Let $U_1, \ldots, U_k$ be $d\times d$ unitary matrices such that $\norm{U_i \ket*{\psi} - \ket*{\psi} }_2 \leq \epsilon$ for all $i \in \{1,\ldots,k\}$, $\ket*{\psi} \in \mathbb{C}^d$ and $\epsilon > 0$.
	Then $\norm{U_k U_{k-1} \cdots U_1 \ket*{\psi} - \ket*{\psi}}_2 \leq k \epsilon$.
\end{cla}
\begin{proof}
	We proceed by induction.
	From the assumptions, $\norm{U_1 \ket*{\psi} - \ket*{\psi} }_2 \leq \epsilon$.
	Now, we assume the hypothesis holds for $k-1$ and consider
	\begin{equation}
		\begin{aligned}
			\norm{U_k U_{k-1} \cdots U_1 \ket*{\psi} - \ket*{\psi}}_2  &= \norm{U_k U_{k-1} \cdots U_1 \ket*{\psi} - U_{k-1} U_{k-2} \cdots U_1\ket*{\psi} + U_{k-1} U_{k-2} \cdots U_1 \ket*{\psi} - \ket*{\psi}}_2 \\
			&\leq
			\norm{U_k U_{k-1} \cdots U_1 \ket*{\psi} -  U_{k-1} U_{k-2} \cdots U_1\ket*{\psi}}_2 + \norm{U_{k-1} U_{k-2} \cdots U_1 \ket*{\psi} - \ket*{\psi}}_2 \\
			&=
			\norm{ U_k \ket*{\psi'} - \ket*{\psi'} }_2 + \norm{U_{k-1} U_{k-2} \cdots U_1 \ket*{\psi} - \ket*{\psi}}_2 \\
            &\leq \epsilon + (k-1)\epsilon = k\epsilon,
		\end{aligned}
	\end{equation}
	where the first inequality is the triangle inequality, for the second equality we define $\ket*{\psi'} := U_{k-1} U_{k-2} \cdots U_1\ket*{\psi}$ and for the last inequality we use the assumption for the claim, along with the inductive hypothesis.
\end{proof}
Here, we have considered only one type of error in the qRAM, but there are other potential sources of error, such as for instance bit-flip errors on the control qubits.
The errors involved in qRAM are discussed more thoroughly and for a slightly different architecture in~\cite{Arunachalam2015}.

From the discussion in this section we see that preparing a vector as a quantum state is a non-trivial task.
Indeed, it is still unclear whether states can be prepared to sufficient precision in polylogarithmic time at scales desirable for applications.
These considerations notwithstanding, for the remainder of these notes, we presume that the unitary operation $\mathcal{R}$ exists and can be carried out in time polylogarithmic in the size of the vector of interest. This concludes our introduction to crucial ideas in quantum algorithms, and in the next section we introduce the HHL algorithm in detail.

\section{HHL} \label{sec:HHL}
We now consider in section \ref{sec:LinSysdefinitions} the problem of solving a system of linear equations, a well-known problem which is  at the heart of many questions in mathematics, physics and computer science. In section \ref{subsec:HHL}, we present the HHL algorithm, with first a brief summary \ref{subsubsec:summary} followed by a more detailed discussion in \ref{subsubsec:details}. Then, in section \ref{subsec:hhl_err_analysis} we analyse how errors occurring both on the input as well as during the computation affect the algorithm's performance. In section \ref{subsec:red} we briefly cover how this problem is $\mathsf{BQP}$-complete as well as its optimality in \ref{subsec:optimality}. Finally, the non-hermitian case is considered in \ref{subsec:non-hermitian}.

\subsection{Problem definition}\label{sec:LinSysdefinitions}

We are given a system of $N$ linear equations with $N$ unknowns which can be expressed as $A \textbf{x}=\textbf{b}$, where $\textbf{x}$ is a vector of unknowns, $A$ is the matrix of coefficients and $\textbf{b}$ is the vector of solutions. If $A$ is an invertible matrix, then we can write that the solution is given by $\textbf{x}=A^{-1}\textbf{b}$. This is known as the Linear Systems Problem (LSP), and can be expressed more formally as given in Definition \ref{def:LSP}.

\begin{definition} \label{def:LSP}
	\emph{(LSP)} Given a matrix $A \in \mathbb{C}^{N \times N}$ and a vector $\vb{b}\in \mathbb{C}^N$, output a vector $\vb{x} \in \mathbb{C}^N$ such that $A\vb{x}=\vb{b}$, or a flag indicating the system has no solution.
\end{definition}
The `flag' can be an ancilla bit with the value of `1' if there is a solution and `0' otherwise.
\par
The quantum version of this problem is called the QLSP \cite{childs2015quantum}, as given in Definition \ref{def:QLSP}, where the matrix is now required to be Hermitian with unit determinant.

\begin{definition}\label{def:QLSP}
	\emph{(QLSP)} Let $A$ be an $N\times N$ Hermitian matrix with unit determinant~\footnote{These restrictions can be slightly relaxed by noting that, even for the non-hermitian matrix $A$, the matrix $\smqty[0 & A^\dagger \\ A & 0]$ is Hermitian (and therefore also square), and any matrix with non-zero determinant can be scaled appropriately.
	This issue is discussed further in section~\ref{subsec:non-hermitian}.}.
	Also, let $\vb{b}$ and $\vb{x}$ be $N$-dimensional vectors such that $\vb{x}:=A^{-1}\vb{b}$. Let the quantum state on $\lceil \log N \rceil$ qubits $\ket*{b}$ be given by
    \begin{equation}
	\ket*{b}:=\frac{\textstyle{\sum_i}b_i\ket*{i}}{ \norm{\textstyle{\sum_i}b_i\ket*{i}}_2}
    \end{equation}
and $\ket*{x}$ by
    \begin{equation}
  \ket*{x}:=\frac{\textstyle{\sum_i}x_i\ket*{i}}{\norm{\textstyle{\sum_i}x_i\ket*{i}}_2},
	\end{equation}
	where $b_i$, $x_i$ are respectively the $i^{\text{th}}$ component of vectors $\vb{b}$ and $\vb{x}$.
	Given the matrix $A$ (whose elements are accessed by an oracle) and the state $\ket*{b}$, output a state $\ket*{\widetilde{x}}$ such that $\norm{\ket*{\widetilde{x}}-\ket*{x}}_2\le\epsilon$ with some probability larger than $\frac{1}{2}$. Note that in practice, we will introduce a `flag' qubit which will determine whether or not this process has been successful.
\end{definition}

Note that the case of systems with no solution, i.e. when $A$ isn't invertible, has not been considered. Indeed, the unit determinant condition precludes this. Furthermore, for all the algorithms we will discuss, a solution close to $\ket*{(A')^{-1} \vb{b}}$ is returned, where $A'$ is the \emph{well-conditioned} component of $A$, that is, the projection of $A$ onto the subspace associated to eigenvalues that are sufficiently large by some criterion, which we discuss later on (section~\ref{subsubsec:details}).

Although the QLSP and LSP problems are similar, these are nonetheless two distinct problems. In particular, the HHL algorithm considers the question of solving QLSP, which has been proven to be a useful subroutine in other quantum algorithms, see e.g.~\cite{rebentrost2014quantum,ashley_finite_element}.

Importantly, solving QLSP has a number of caveats as compared with solving LSP. The main difference is the requirement that both the input and output are given as quantum states. 
This means that any efficient algorithm for QLSP (for whichever definition of `efficient' one is concerned with) requires \emph{i)} an `efficient' preparation of $\ket*{b}$ and \emph{ii)} `efficient' readout of $\ket*{x}$, both of which are non-trivial tasks. 

In the quantum linear systems algorithm literature, `efficient' is taken to be `polylogarithmic in the system size $N$'. We can immediately see how this is problematic if we wish to read out the elements of $\ket*{x}$, since we require time $\Ord{N}$ for this.
Thus, a solution to QLSP must be used as a \emph{subroutine} in an application where \emph{samples} from the vector $\vb{x}$ are useful. 
More extensive discussion can be found in~\cite{aaronson2015read}.

\subsection{The HHL algorithm}\label{subsec:HHL}
In the following, we first present a summary of the HHL algorithm in section \ref{subsubsec:summary}. Then, in section \ref{subsubsec:details}, we delve into the details of the algorithm.

\subsubsection{Algorithm summary} \label{subsubsec:summary}

The HHL algorithm proceeds in the following three steps: first with phase estimation, followed by a controlled rotation and finally uncomputation.

Let $A=\sum_j \lambda_j |u_j \rangle \langle u_j|$, and let us first consider the case when the input state is one of the eigenvectors of $A$, $\ket*{b} = \ket*{u_j}$, .
As seen in section \ref{subsec:QFT}, given a unitary operator $U$ with eigenstates $\ket*{u_j}$ and corresponding complex eigenvalues $ e^{ i \varphi_j}$, the technique of quantum phase estimation allows for the following mapping to be implemented:
\begin{equation}
\ket*{0}\ket*{u_j} \mapsto \ket*{\tilde{\varphi}} \ket*{u_j},
\end{equation}
where $\tilde{\varphi}$ is the binary representation of $\varphi$ to a certain precision. In the case of a Hermitian matrix $A$, with eigenstates $\ket*{u_j}$ and corresponding eigenvalues $\lambda_j$, we have that the matrix $\exp( i A t)$ is unitary, with eigenvalues $ \exp (i \lambda_j t)$ and eigenstates $\ket*{u_j}$.
Thus, the technique of phase estimation can be applied to the matrix $\exp( i A t)$ in order to implement the mapping:
\begin{equation}
\ket*{0}\ket*{u_j} \mapsto \ket*{\tilde{\lambda}_j}\ket*{u_j},
\end{equation}
where $\tilde{\lambda}_j$ is the binary representation of $\lambda_j$ to a tolerated precision. 

The second step of the algorithm implements a controlled rotation conditioned on $\ket*{\tilde{\lambda}_j}$. In order to do this, a third ancilla register is added to the system in state $\ket*{0}$, and performing the controlled $\sigma_y$-rotation produces a normalised state of the form
\begin{equation}
\sqrt{1-\frac{C^2}{\tilde{\lambda^2_j}}}\ket*{\tilde{\lambda}_j}\ket*{u_j}\ket*{0}+ 	\frac{C}{\tilde{\lambda_j}}\ket*{\tilde{\lambda}_j}\ket*{u_j}\ket*{1},
\end{equation}
\noindent where $C$ is a constant of normalisation.
As seen in Lemma~\ref{lem:cosrotation}, this can be achieved through the application of the operator
	\begin{equation}
    \exp (- i \theta \sigma_y)=
	\left(
	\begin{array}{cc}
	\cos (\theta) & - \sin(\theta) \\
	\sin(\theta) & \cos(\theta)
	\end{array}
	\right),
	\end{equation}
where we have that $\theta = \arccos (C/ \tilde{\lambda})$.


By definition, we have that $A=\sum_j \lambda_j |u_j \rangle \langle u_j|$, and so its inverse is  given by $A^{-1}=\sum_j \frac{1}{\lambda_j} |u_j \rangle \langle u_j|$.
Next, from the definition of the QLSP,  we assume we are given the quantum state $|b\rangle = \sum_i b_i |i\rangle$. This state can be expressed in the eigenbasis $\{|u_j\rangle\}$ of operator $A$, i.e.\ $|b\rangle =\sum_j \beta_j |u_j\rangle $. 
So, enacting the procedure described above on the superposition $|b\rangle =\sum_j \beta_j |u_j\rangle $ we get the state

\begin{equation}
\sum^N_{j=1}\beta_j \ket*{\tilde{\lambda}_j}\ket*{u_j}\left( \sqrt{1-\frac{C^2}{\tilde{\lambda}^2_j}} \ket*{0}+\frac{C}{\tilde{\lambda}_j}\ket*{1}\right).
\end{equation}
We uncompute the first register, giving us

\begin{equation}
\ket*{0} \otimes \sum^N_{j=1}\beta_j \ket*{u_j}\left( \sqrt{1-\frac{C^2}{\tilde{\lambda}^2_j}} \ket*{0}+\frac{C}{\tilde{\lambda}_j}\ket*{1}\right).
\end{equation}

Now notice that $A^{-1}\ket*{b} = \sum^N_{j=1}\frac{\beta_j}{\tilde{\lambda}_j}\ket*{u_j}$.
Thus, the quantum state $|x\rangle=A^{-1} |b\rangle$ (or more precisely, a state close to $\ket*{x}$) can be constructed in the second register by measuring the third register and postselecting on the outcome `1', modulo the constant factor of normalisation $C$.
Later, we will use amplitude amplification at this step to boost the success probability instead of simply measuring and postselecting.

This entire process made of three consecutive steps---phase estimation, controlled rotation and uncomputation---is illustrated in Figure~\ref{fig:hhlalgocircuit}.

\begin{figure}[h!]
	\centering
	\includegraphics[width=0.85\textwidth]{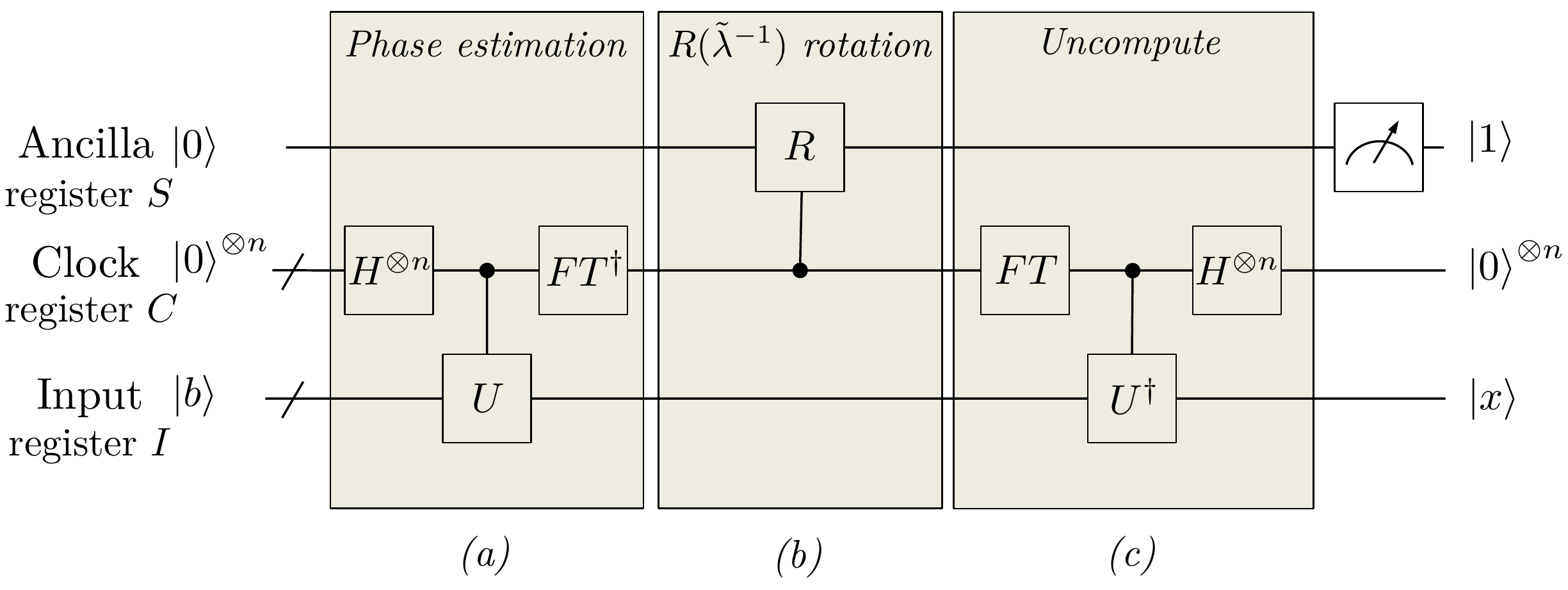}
	\caption{\emph{HHL Algorithm Schematic:} Broadly, the algorithm proceeds in three main steps.
		\emph{(a)} Phase estimation is used to estimate the eigenvalues of $A$, using $U = \sum^{T-1}_{k=0}\ketbra*{k}^C\otimes e^{ i A k t_0/T}$, where $T=2^t$, with $t$ the number of qubits in the clock register $C$ and $t_0=2\pi$.
		After this stage the state of the computation is $ (\sum^N_{j=1}\beta_j\ket*{u_j}^R\ket*{\tilde{\lambda}_j}^I)\otimes\ket*{0}^S$, where the $\ket*{u_j}$ are the eigenvectors of $A$, the $\beta_j$ are $\ket*{b}$'s representation in this basis and the $\ket*{\tilde{\lambda}_j}$ are the binary representations of the eigenvalues of $A$.
		\emph{(b)} The controlled $R(\tilde{\lambda}^{-1})$ rotation extracts the eigenvalues of $A^{-1}$ and executes a $\sigma_y$ rotation, conditioned on $\tilde{\lambda}_j$, leaving the state as $\sum^N_{j=1}\beta_j\ket*{u_j}_I\ket*{\tilde{\lambda}_j}^C( (1-C^2/\tilde{\lambda}_j^2)^{1/2}\ket*{0}+C/\tilde{\lambda}_j\ket*{1})^S$, where $C$ is a normalising constant.
		\emph{(c)} The inverse phase estimation subroutine sets the register to $(\ket*{0}^{\otimes n})^C$ and leaves the remaining state as $\sum^N_{j=1}\beta_j\ket*{u_j}^I( (1-C^2/\tilde{\lambda}_j^2)^{1/2}\ket*{0}+C/\tilde{\lambda}_j\ket*{1})^S$, so postselecting on $\ket*{1}^S$ gives the state $C\sum^N_{j=1}(\beta_j/\tilde{\lambda}_j)\ket*{u_j}^I$, which is proportional to $\ket*{x}^I$.}
	\label{fig:hhlalgocircuit}
\end{figure}

We now consider the run time of the HHL algorithm, and compare it with its classical counterpart.
\begin{definition}\label{def:cn}\emph{(Condition number)}
The condition number of $A$ is the ratio of its largest to smallest eigenvalue, i.e. $\kappa=\lambda_{\text{max}}/ \lambda_{\text{min}}$.
\end{definition}
The best general purpose classical matrix-inversion algorithm, the conjugate-gradient method \cite{Shewchuck1994}, runs with $O\qty(N s \kappa \log(1/\epsilon))$, where $s$ is the matrix sparsity, $\kappa$ the condition number and $\varepsilon$ the precision. In contrast, the HHL algorithm scales as $O\qty(\log(N) s^2 \kappa^2 / \epsilon)$, and is thus exponentially faster in $N$, but linearly slower in sparsity $s$ and condition number $\kappa$. 
In particular, we have that the HHL scales exponentially worse in the precision $\epsilon$, a slowdown which was subsequently eliminated by
Childs \emph{et al.} $\epsilon$~\cite{childs2015quantum} and which is discussed further in section~\ref{subsec:exp. improved}.

One can ask if there might exist an even more efficient classical algorithm. In~\cite{harrow2009quantum}, it is established that this is highly unlikely, as matrix inversion can be shown to be $\mathsf{BQP}$-complete, see section~\ref{subsec:optimality}. More precisely, they show that a classical poly$(\log N, \kappa, 1/\epsilon)$ algorithm would be able to simulate a poly$(N)$-gate quantum circuit in poly$(N)$ time, a scenario which is generally understood to be implausible.
 \par
Finally, we note that it is assumed that the state $\ket*{b}$ can be efficiently constructed, i.e., in polylogarithmic time. Efficient state preparation was previously discussed in section \ref{subsec:qram}, and is an important step in the computational process, with the potential to dramatically slow-down an algorithm.

\subsubsection{Algorithm details}\label{subsubsec:details}
Let us first detail the HHL algorithm rigorously, and then proceed with the analysis. The pseudo-code for the HHL algorithm is given in Algorithm~\ref{alg:HHL}.
\begin{algorithm}[!htbp]
	\caption{HHL algorithm for QLSP.}
	\label{alg:HHL}
	\begin{itemize}
		\item[] \textbf{Input:} State vector $\ket*{b}$, matrix $A$ with oracle access to its elements.
		Parameters $t_0 = \Ord{\kappa /\epsilon}$, $T = \widetilde{\mathcal{O}}(\log(N) s^2 t_0)$, $\epsilon$ is desired precision.
        \item[] $\mathcal{A}_{\text{HHL}}(\ket*{b}, A, t_0, T, \epsilon)\  \{$
        \begin{enumerate}
          \item Prepare the input state $\ket*{\varPsi_0}^C \otimes \ket*{b}^I$, where $\ket*{\varPsi_0} = \sqrt{\frac{2}{T}}\sum_{\tau=0}^{T-1}\sin \frac{\pi(\tau+\frac{1}{2})}{T}\ket*{\tau}^C$.
          \item Apply the conditional Hamiltonian evolution $\sum_{\tau=0}^{T-1} \ketbra*{\tau}^C \otimes  e^{ i A \tau t_0/T}$ to the input.
          \item Apply the quantum Fourier transform to the register $C$, denoting the new basis states $\ket*{k}$, for $k\in \{0,\ldots T-1\}$. Define $\tilde{\lambda}:= 2\pi k/ t_0$.
          \item Append an ancilla register, $S$, and apply a controlled rotation on $S$ with $C$ as the control, mapping states $\ket*{\tilde{\lambda}} \mapsto \ket*{h(\tilde{\lambda})}$, with $\ket*{h(\tilde{\lambda})}$ as defined in Eq~\ref{eq:filter_funcs}.
          \item Uncompute garbage in the register $C$.
          \item Measure the register $S$.
          \item $\texttt{if}($ result $=$ `well' $)$ $\{$ \texttt{return} register $I$ $\}$
          \item[] $\texttt{else}\{$ \texttt{goto} step 1. $\}$
          \item[$\}$]
        \end{enumerate}
        \item[] Perform $\Ord{\kappa}$ rounds of amplitude amplification on $\mathcal{A}_{\text{HHL}}(\ket*{b}, A, t_0, T, \epsilon)$.
		\item[] \textbf{Output:} State $\ket*{\tilde{x}}$ such that $\norm{\ket*{\tilde{x}} - \ket*{x}}_2 \leq \epsilon$.
	\end{itemize}
\end{algorithm}

We start with an input register $I$ and a clock register $C$.
The first step of the algorithm is to prepare the state $\ket*{b}^I$. To do so, we simply assume that there exist a unitary operator $B$ and an initial state $\ket*{\text{initial}}$ such that $B\ket*{\text{initial}}^I=\ket*{b}^I$ to perfect accuracy, requiring $T_B$ gates to implement, where $B$ is the qRAM oracle from section~\ref{subsec:qram} and $T_B$ is polylogarithmic in the dimension of $\ket*{b}$. Note that, as previously discussed, this is a delicate step whereby the complexity of state preparation could dwarf any speed-up achieved by the algorithm itself.

Then, the clock register is prepared in the state
\begin{equation}
\ket*{\varPsi_0}=\sqrt{\frac{2}{T}}\sum_{\tau=0}^{T-1}\sin \frac{\pi(\tau+\frac{1}{2})}{T}\ket*{\tau}^C,
\end{equation}
which  can be prepared up to error $\epsilon_{\varPsi}$ in time $\poly{\log(T/\epsilon_{\varPsi})}$ \cite{Grover2002}. The time $T = \widetilde{\mathcal{O}}(\log (N) s^2 t_0)$, corresponds to the number of computational steps required to simulate $e^{iAt}$ for some time $0 \leq t\leq t_0$ when $A$ is $s$-sparse (see section~\ref{subsec:ham_sim}) and $t_0 = \Ord{\kappa/\epsilon}$.
The quotient $t_0/T$ is the step size of the simulation.
\par
Next, the conditional Hamiltonian evolution $\sum_{\tau=0}^{T-1} \ketbra*{\tau}^C \otimes  e^{ i A \tau t_0/T}$ is applied to the input state $\ket*{\varPsi_0}^C\otimes \ket*{b}^I$ using the Hamiltonian simulation techniques described in section~\ref{subsec:ham_sim}.
By conditional Hamiltionian simulation we mean that the length of the simulation is conditioned on the value of the clock register $\ket*{t}^C$.
The parameter $t_0$ is chosen to achieve the desired error bound, which we further discuss in section~\ref{subsec:hhl_err_analysis}. It can be easily verified that this results in the state
\begin{equation}\label{eq:hhlphasecalc1}
\sqrt{\frac{2}{T}}\sum^N_{j=1}\beta_j \qty(\sum^{T-1}_{\tau=0}  e^{\frac{ i \lambda_j t_0 \tau}{T}}  \sin \frac{\pi(\tau+\frac{1}{2})}{T} \ket*{\tau}^C)\ket*{u_j}^I,
\end{equation}
where $\lambda_j$, $\ket*{u_j}$ are the $j$\textsuperscript{th} eigenvalue and eigenvector of $A$ respectively.
Now,  the state of the first qubit, i.e. the bracketed part of Eq.~\eqref{eq:hhlphasecalc1}, is expressed in the Fourier basis $\ket*{k}$ by taking the inner product with the state $\frac{1}{\sqrt{T}}\sum^{T-1}_{k=0} e^{-\frac{2\pi i k \tau}{T}}\ket*{k}$, see section \ref{subsubsec:QFT}. This leads to the state:
\begin{equation}
\sum^N_{j=1}\beta_j\sum^{T-1}_{k=0}\qty(\frac{\sqrt{2}}{T}\sum^{T-1}_{\tau=0}  e^{\frac{ i \tau}{T}(\lambda_j t_0 - 2\pi k)}  \sin \frac{\pi(\tau+\frac{1}{2})}{T})\ket*{k}^C \ket*{u_j}^I :=\sum^N_{j=1}\beta_j\sum^{T-1}_{k=0}\alpha_{k|j}\ket*{k}^C \ket*{u_j}^I,
\end{equation}
where we have defined the coefficient $\alpha_{k|j}=\qty(\frac{\sqrt{2}}{T}\sum^{T-1}_{\tau=0}  e^{\frac{ i \tau}{T}(\lambda_j t_0 - 2\pi k)}  \sin \frac{\pi(\tau+\frac{1}{2})}{T})$. 

Now, let $\delta:=\lambda_j t_0 - 2\pi k$. The goal is now to derive the following upper bound for the coefficients: $\abs{\alpha_{k|j}}^2 \leq 64\pi^2/\delta^2$ when $\abs{k-\lambda_j t_0 /2\pi}\geq 1$, the full calculation of which can be found in~\cite[Appendix A]{harrow2009quantum}. Here, we discuss the key steps from the proof. First, the identity $2i\sin x=e^{ix}-e^{-ix}$ is applied, giving
\begin{equation}
\alpha_{k|j}=\frac{1}{i\sqrt{2}T} \sum_{\tau=0}^{T-1} \left( e^{\frac{i\pi}{2T}}e^{i\tau \frac{\delta + \pi}{T}}-e^{-\frac{i\pi}{2T}}e^{i\tau \frac{\delta - \pi}{T}} \right) .
\end{equation}
This can then be identified as a geometric sequence, and we can thus apply the well-known expression for the sum of the first $T$ terms, $S_T=\frac{a(1-r^T)}{1-r}$, where $a$ is the first term and $r$ the common ratio. By rearranging, and using the identity $2\cos x=e^{ix}+ e^{-ix}$, we finally have
\begin{equation}\label{eq:hl_pe_calc}
	\alpha_{k|j} = e^{i \frac{\delta}{2} (1- \frac{1}{T}) } \frac{\sqrt{2} \cos(\frac{\delta}{2}) }{ T } \cdot
	\frac{ 2 \cos(\frac{ \delta }{ 2T }) \sin(\frac{ \pi }{ 2T }) }{ \sin(\frac{ \delta + \pi }{ 2T }) \sin(\frac{ \delta - \pi  }{ 2T })  }.
\end{equation}

We can take $\delta \geq 2\pi$ since we want a result for the case when $\abs{k-\lambda_j t_0 /2\pi}\geq 1$.
Also, $T$ is sufficiently large so that $\delta \leq T/ 10$.
Taking note that $\alpha - \alpha^3/6 \leq \sin \alpha$ for small $\alpha$ in the denominator of~\ref{eq:hl_pe_calc}, $\sin \alpha \leq \alpha $ in the numerator and $\cos \alpha \leq 1$ we get that
\begin{equation}
	\abs{\alpha_{k|j}} \leq \frac{ 4\pi \sqrt{2} }{ (\delta^2 - \pi^2) ( 1 - \frac{\delta^2 + \pi^2}{3T^2} ) } \leq \frac{ 4\pi \sqrt{2} }{ (\delta^2 - (\delta / 2)^2 )  ( 1 - \frac{\delta^2 + (\delta / 2)^2}{3(\delta/10)^2} ) } \leq \frac{4\pi \sqrt{2} }{  \frac{3}{4}\delta^2 (1 - \frac{5}{1200}) } \leq \frac{8\pi}{\delta^2}.
\end{equation}

\noindent Thus, $\abs{\alpha_{k|j}}^2 \leq 64\pi^2/\delta^2$ whenever $\abs{k-\lambda_j t_0 /2\pi}\geq 1$.

We now have that $\abs{\alpha_{k|j}}$ is large if and only if $\lambda_j\approx \frac{2\pi k}{t_0}$.
We can relabel the basis states $\ket*{k}$  by defining $\widetilde{\lambda}_k:=2\pi k/t_0$, which gives 
\begin{equation}\label{eq:Ptilde}
\sum^N_{j=1}\beta_j\sum^{T-1}_{k=0}\alpha_{k|j}\ket*{\widetilde{\lambda}_k}^C \ket*{u_j}^I.
\end{equation}
\par
Next, an additional ancillary register $S$ is adjoined to the state which is used to perform a controlled inversion of the eigenvalues. To do so, the first register storing the eigenvalues will be used.

Previously, in \ref{subsubsec:summary}, we saw how the controlled inversion on the eigenvalues of $A$ was used to apply $A^{-1}$ to the input state.
Here, it is important to consider the numerical stability of the algorithm. For instance, suppose we have a quantity $\mu \in \mathbb{R}$ that is close to zero and we wish to compute $1/\mu$.
Any small change in $\mu$ results in a large change in $1/\mu$ and so we can only reliably calculate $1/\mu$ for sufficiently large $\mu$. Thus, in the context of the HHL algorithm, we would wish to only invert the well conditioned part of the matrix, i.e.\ the eigenvalues that lie in a certain range of values that is large with respect to $1/ \kappa$.
Why do we need this range of values to be large with respect to $\kappa$?
Suppose we have an eigenvalue $\lambda := \epsilon_\kappa / \kappa$ for some $0 < \epsilon_\kappa \ll 1 $ and we invert it, i.e., we have $1/\lambda = \kappa / \epsilon_\kappa$.
A small relative error in $1/\lambda$ will give a result deviating from the true value by many times $\kappa$, the `characteristic' scale of the matrix at hand, $A$. This error would dominate all other terms in the sum $\sum_j (\lambda_j)^{-1} \ketbra*{u_j}$ and so the returned value of $A^{-1}$ would deviate from its true value to an unacceptable degree.

To achieve this, we introduce the \emph{filter functions} $f(\lambda)$, $g(\lambda)$, that act to invert $A$ only on its well-conditioned subspace, that is, the subspace spanned by eigenvectors corresponding to $\lambda\geq 1/\kappa$.
This is so that a small error in $\lambda$ doesn't introduce a large error in $A^{-1}$, as discussed in the previous paragraph.
We require (for the error analysis) that the map implementing the filter functions is Lipschitz continuous, that is, it has bounded derivative.
The controlled rotations in the algorithm (step 4 in $\mathcal{A}_{\text{HHL}}$) are controlled by an angle of the form $\theta = \arccos x$, where $x$ is the output of a filter function on $\lambda$. Any argument to $\arccos(\,\cdot\,)$ needs to lie in the interval $[-1, 1]$, so the image of the filter functions needs to be $[-1,1]$.
The domain of the filter functions is $[\lambda_{min}, \lambda_{max} ]$.
We demand that the filter functions are proportional to $1/\lambda$ in the well-conditioned subspace, to carry out the eigenvalue inversion. 
We are also concerned with intermediary eigenvalues, characterised by $1/\kappa' \leq\lambda\leq 1/\kappa$ for some $\kappa'$, say $\kappa'=2\kappa$, where we want interpolating behaviour.
This interpolating behaviour leads to better numerical stability~\cite{rank_deficient}.

The (not unique) choice of filter functions satisfying all of these desiderata chosen by the authors of \cite{harrow2009quantum} are: 
\begin{equation}\label{eq:HHLfiltfunctions}
f(\lambda) =
\begin{cases}
\frac{1}{2\kappa\lambda}, & \lambda \geq 1/\kappa;\\
\frac{1}{2}\sin(\frac{\pi}{2}\cdot\frac{\lambda-\frac{1}{\kappa'}}{\frac{1}{\kappa}-\frac{1}{\kappa'}}), & \frac{1}{\kappa}>\lambda>\frac{1}{\kappa'};\\
0, & \frac{1}{\kappa'}>\lambda;
\end{cases}
\qq{and}
g(\lambda) =
\begin{cases}
0, & \lambda \geq 1/\kappa;\\
\frac{1}{2}\cos(\frac{\pi}{2}\cdot\frac{\lambda-\frac{1}{\kappa'}}{\frac{1}{\kappa}-\frac{1}{\kappa'}}), & \frac{1}{\kappa}>\lambda>\frac{1}{\kappa'};\\
\frac{1}{2}, & \frac{1}{\kappa'}>\lambda.
\end{cases}
\end{equation}
\par

\noindent Notice that this step introduces a $\kappa$ dependency into the algorithm.

After the controlled rotation the register $S$ is then in the state:
\begin{equation}\label{eq:filter_funcs}
\ket*{h(\widetilde{\lambda}_k)}^S := \sqrt{1-f(\widetilde{\lambda}_k)^2 - g(\widetilde{\lambda}_k)^2} \ket*{\text{nothing}}^S + f(\widetilde{\lambda}_k)\ket*{\text{well}}^S + g(\widetilde{\lambda}_k)\ket*{\text{ill}}^S.
\end{equation}
for functions $f(\lambda),g(\lambda)$ which are defined above in Eq.~\eqref{eq:HHLfiltfunctions}. The flag `nothing' corresponds to no inversion taking place, `well' means it has, and `ill' indicates the part of $\ket*{b}$ in the ill-conditioned subspace of $A$.
\par
After the filter functions are applied, the phase estimation procedure is reversed, uncomputing garbage qubits in the process, a technique we discussed in section~\ref{subsec:uncompute}. We denote by $U_{\text{invert}}$ the procedure which we have described up until now.
Applying $U_{\text{invert}}$ to $\ket*{b}$, then measuring $S$ with the outcome `well', returns the state $\ket*{\widetilde{x}}$ with success probability $\widetilde{p} = \Ord{1/\kappa^2}$, which we shall derive in section~\ref{subsec:hhl_err_analysis}.
Thus by the amplitude amplification lemma (Lemma~\ref{lemma:amplitude_amplification}), we have arbitrary success probability with a number $\Ord{1/\sqrt{\tilde{p}}} = \Ord{\kappa}$ of repetitions.
\par
Previously, we saw that preparation of the state $|b\rangle$ takes 
$\widetilde{O}(T_B)$ and running the quantum simulation $\widetilde{O}(t_0 s^2 \log N)$. Thus, the total algorithm run time is $\widetilde{O}(\kappa(T_B+t_0 s^2 \log N))$, the factor $\kappa$ is due to amplitude amplification. Since we have that  $t_0=\Ord{\kappa/\epsilon}$, the runtime can be written as $\widetilde{O}(\kappa T_B + \kappa^2 s^2 \log(N)/\epsilon)$. We will now give a more detailed error analysis.

\subsection{Error analysis}\label{subsec:hhl_err_analysis}
Now, we consider the error analysis of the HHL algorithm. First, we present a short intuitive error analysis, before then discussing the main components of the full error analysis presented in the original paper. 
\subsubsection{A short analysis}
Having introduced the filter functions, we can now use them to bound the error of the computation. Recall that these allow us to invert only the
well-conditioned part of the matrix while it flags the ill-conditioned eigenvalues and
interpolates between these two behaviors when $1/\kappa' < |\lambda| < 1/\kappa$, where we had that $\kappa'=2\kappa$. We therefore only invert eigenvalues which
are larger than $1/\kappa'$. 
We will need the following lemma:
\begin{lemma}[Filter functions are $\Ord{\kappa}$-Lipschitz\cite{harrow2009quantum}]
\label{lem:lipschitz}
The map $\lambda \rightarrow \ket*{h(\lambda)}$ is $\Ord{\kappa}$-Lipschitz, i.e.\ for all eigenvalues $\lambda_i \neq \lambda_j$:
\begin{align}
\norm{ \ket*{h(\lambda_i)} - \ket*{h(\lambda_j)} }_2  = \sqrt{2(1- Re \langle h(\lambda_i) |h(\lambda_j)\rangle )}
\leq c \kappa | \lambda_i - \lambda_j|,
\end{align}
for some constant $c = \Ord{1}$.
\end{lemma}
\noindent
The proof relies on using the filter functions defined in the previous section, and bounding the derivatives piecewise for the different regimes.

In an ideal setting, we would be able to implement the Hamiltonian simulation with negligible error, followed by perfect phase estimation procedure and controlled rotation without postselection (i.e.\ the unitary that performs the inversion of the eigenvalues). Let the operator $Q$ correspond to this error-free operation, given by
\begin{align} \label{equ:error free}
\ket*{\psi} := Q \ket*{b}^I\ket*{0}^S = \sum_i \beta_i \ket*{ u_i}^I \ket*{h(\lambda_i)}^S,
\end{align}
where we had $|b\rangle^I = \sum_i \beta_i |u_i\rangle^I$.

In contrast, let $\tilde{Q}$ be the operation describing the same procedure but where the phase estimation step is erroneous, i.e.\ the eigenvalues are estimated to a certain $\delta$ error each, as discussed in~\ref{subsec:phase est}. Thus, the erroneous operator $\tilde{Q}$ is given by

\begin{align}\label{equ:error}
\ket*{\tilde{\psi}} := \tilde{Q} \ket*{b}^I\ket*{0}^S = \sum_i \beta_i \ket*{ u_i}^I \ket*{h(\tilde{\lambda}_i)}^S.
\end{align}
Indeed, quantum phase estimation, which we summarised with Theorem~\ref{pest} in the phase estimation section, can be expressed in the slightly different notation as the following.

\begin{theorem}[Phase estimation \cite{kitaev1995quantum}]
  \label{pest2}
  Let the unitary $U \ket*{ v_j} = \exp(i \theta_j) \ket*{ v_j}$ with $\theta_j \in [ - \pi ,\pi ]$ for $j \in [n]$. There is a quantum algorithm that transforms $\sum_{j \in [n]} \alpha_j \ket*{ v_j} \to \sum_{j \in [n]} \alpha_j \ket*{ v_j} \ket*{\tilde{\theta}_j}$ such that $|\tilde{\theta_{j}} - \theta_{j} |\leq \delta$ for all $j\in [n]$ with probability $1-1/\poly{n}$ in time $\Ord{T_U \log{(n)} / \delta}$, where $T_U$ defines the time to implement $U$.
\end{theorem}
We will use this theorem now to bound the error $\epsilon$ in the final state in terms of the error $\delta$ of the phase estimation procedure.
The original proof is slightly more complicated and takes into account a improved phase estimation scheme, which we will omit here in order to make the steps more understandable.

The goal is to bound the error in $\norm{ \tilde{Q} - Q}_2$ due to the imperfect phase estimation. By choosing an arbitrary but general state $\ket*{ b}$, this is equivalent to bounding the quantity $\norm{ Q \ket*{ b}^I - \tilde{Q} \ket*{b}^I}_2 := \norm{ \ket*{\tilde{\psi}} - \ket*{\psi} }_2$.

The following norm is used as a distance measure between quantum states:
\begin{equation}\label{equ:norm}
\norm{ \ket*{\tilde{\psi}} - \ket*{\psi} }_2 = \sqrt{2 \left(1-Re \braket*{\tilde{\psi}}{\psi} \right)},
\end{equation}
so it suffices
to lower-bound the quantity $Re\braket*{\tilde{\psi}}{\psi} \in [0,1]$. Here, we shall consider the error-free and erroneous states given in Eqs. \eqref{equ:error free} and \eqref{equ:error} respectively.
We have
\begin{align}
Re \braket*{\tilde{\psi}}{\psi} = \sum\limits_{i=1}^N |\beta_i |^2  Re\braket*{h(\tilde{\lambda}_i)}{h(\lambda_i)}^S
\geq \sum\limits_{i=1}^N |\beta_i |^2 \left( 1 - \frac{c^2 \kappa^2 \delta^2}{2} \right),
\end{align}
where we use the $\Ord{\kappa}$-Lipschitz property of Lemma~\ref{lem:lipschitz} and apply the error bound on the eigenvalues from phase estimation, i.e. $| \lambda_i - \tilde{\lambda}_i | \leq \delta $.

Recalling that $\sum\limits_{i=1}^N |\beta_i |^2 = 1$ and $ 0 \leq Re \braket*{\tilde{\psi}}{\psi} \leq 1$, it follows that
\begin{align}
1 - Re \braket*{\tilde{\psi}}{\psi} \leq
1 - \sum\limits_{i=1}^N |\beta_i |^2 \left( 1 - \frac{c^2 \kappa^2 \delta^2}{2} \right) =
\sum\limits_{i=1}^N |\beta_i |^2 \left(\frac{c^2 \kappa^2 \delta^2}{2} \right).
\end{align}
Finally, using again $\sum_i |\beta_i|^2 = 1$, recalling that $c=\Ord{1}$ and observing that the coefficients are independent of $i$, the distance can be bounded as
\begin{align}
\norm{ \ket*{\tilde{\psi}} - \ket*{\psi} }_2 \leq \Ord{\kappa \delta}.
\end{align}
If this error is required to be of $\Ord{\epsilon}$, then we need to take the phase estimation accuracy to be $\delta = \Ord{\frac{\epsilon}{\kappa}}$.
This results in a runtime that scales as $\Ord{\kappa \cdot \text{polylog}(N)/\epsilon}$ assuming the postselection is successful, since $\text{polylog}(N)$ is required for other steps which are added multiplicatively.\\
The original proof contains an improved phase estimation scheme, which we saw above in section \ref{subsubsec:details}, which makes the analysis slightly more complex. In that case, one needs to determine the probability of measuring a correct $2\pi k$
by proving concentration bounds for the amplitudes, and use the bounds to case-wise bound the overlap (Please see Appendix A of~\cite{harrow2009quantum} for details).
For the improved phase estimation scheme we refer the reader to the results of \cite{luis1996optimum} and \cite{buvzek2005optimal}.

The second $\kappa$ dependence in the runtime of the algorithm is introduced through the postselection process in the eigenvalue inversion step. We give now again a simplified analysis for this step, which assumes that all the eigenvalues are in the well-conditioned subspace, that is, the $S$ register is in the state $\ket*{\text{well}}$.

The postselection on the $S$ register is conditioned on being in state $\ket*{\text{well}}$. The probability of success of this step $\tilde{p}$ is at least $1/\kappa^2$, since 
\begin{equation}\label{eq:tilde_p}
	\tilde{p} \geq \sum_{i\,:\, \lambda_i \geq 1/\kappa} |\beta_i|^2 \cdot |1/\kappa \lambda_i|^2 = \Ord{1/\kappa^2}.
\end{equation}
Using amplitude amplification this can be boosted to $1/\kappa$, by increasing the amplitude of the $\ket*{\text{well}}$ subspace. We hence see that a number of $\Ord{\kappa}$ repetitions on average suffices for success.
Combining with the scaling above results in $\Ord{\kappa^2 \cdot \text{polylog}(N)/\epsilon}$ time, since the $\kappa$ dependences combine multiplicatively.

\subsubsection{Detailed analysis}

First, we introduce the operator $\tilde{P}$ which implements the first three steps of the algorithm (i.e. up until line 3 in $\mathcal{A}_{\text{HHL}}$, Algorithm~\ref{alg:HHL}) and is given by
\begin{equation}
\tilde{P}=\sum_j |u_j \rangle \langle u_j|^I \otimes \sum_k \alpha_{k|j}\ketbra*{k}^C |\text{garbage}(j,k)\rangle \langle \text{initial}|^S. 
\end{equation}
For context, $\tilde{P}$ is the operator that maps the input state to Eq.~\eqref{eq:Ptilde} and is analogous to the operator $\tilde{Q}$ in the previous section.
Previously, we saw that given an input state $|b\rangle =\sum_j  \beta_j |u_j\rangle^I$, the output would ideally be given by
\begin{equation}
|\psi \rangle = \sum_j \beta_j |u_j\rangle^I |h(\lambda_j)\rangle^S. 
\end{equation}

\noindent
In practice, due to the error from phase estimation, the output state obtained will instead be of the form 
\begin{equation}
|\tilde{\psi} \rangle = \tilde{P}^\dag \sum \beta_j |u_j \rangle^I \sum_k \alpha_{k|j} |k\rangle^C |h(\tilde{\lambda}_k)\rangle^S, 
\end{equation}
where we defined $\tilde{\lambda}_k=\frac{2\pi k}{t_0}$ and $\tilde{P}^\dagger$ enacts the uncomputation procedure.

\noindent
The goal is now to upper bound the quantity $\norm{\ket*{\tilde{\psi}} - \ket*{\psi} }_2 $ which, from Eq.~\eqref{equ:norm}, reduces to determining a lower bound on the fidelity
\begin{equation}
\langle \tilde{\psi}|\psi \rangle =\sum_j |\beta_j|^2 \sum_k |\alpha_{k|j}|^2 \langle h(\tilde{\lambda}_k) |h(\lambda_j) \rangle^S .
\end{equation}

The quantity $\langle h(\tilde{\lambda}_k) |h(\lambda_j) \rangle^S $ can be interpreted as a random variable distributed according to the probability distribution $\operatorname{Pr}(j,k)=|\beta_j|^2 |\alpha_{k|j}|^2$.
Thus, we can write that $ Re \langle \tilde{\psi}|\psi \rangle=\mathbb{E}_j \mathbb{E}_k [\langle h(\tilde{\lambda}_k) |h(\lambda_j) \rangle ].$
Applying Lemma \ref{lem:lipschitz} and rearranging, we have that
\begin{equation}
 \langle h(\tilde{\lambda}_k) |h(\lambda_j) \rangle^S \geq 1 - \frac{c^2 \kappa^2\delta^2}{2 t_0^2},
\end{equation}
\noindent
where $c\leq \frac{\pi}{2}$ and where we previously defined $\delta= \lambda_j t_0 -2\pi k=t_0(\lambda_j -\tilde{\lambda}_k)$ and thus $\lambda_j - \tilde{\lambda}_k=\frac{\delta}{t_0}$.
If $\delta \leq 2 \pi$, then we have that
$ \langle h(\tilde{\lambda}_k) |h(\lambda_j) \rangle^S \geq 1 - \frac{2 \pi^2 c^2 \kappa^2}{ t_0^2}$, giving an infidelity contribution of $\frac{2 \pi^2 c^2 \kappa^2}{ t_0^2}$. 
If $\delta > 2\pi$, then we previously saw that $|\alpha_{k|j}|^2\leq 64 \frac{\pi^2 }{\delta^4}$.
Summing up to compute the entire infidelity contribution over all $\lambda_j$ gives
\begin{equation}
 2 \sum_{k=\frac{\lambda_j t_0 }{2\pi }+1}^\infty \frac{64 \pi^2}{\delta^4}\frac{c^2 k^2 \delta^2}{2t_0^2} = \frac{64 \pi^2 c^2 \kappa^2}{t_0^2} \sum_{k=1}^\infty \frac{1}{4\pi^2 k^2}= \frac{8\pi^2 c^2}{3}\frac{\kappa^2}{t_0^2}.
\end{equation}
\noindent
Thus, we have that $Re \langle \tilde{\psi }|\psi \rangle \geq 1-\frac{5\pi^2 c^2 \kappa^2}{t_0^2}$, and $\norm{|\tilde{\psi}\rangle - |\psi \rangle }_2 \leq \frac{4\pi c \kappa}{t_0}$.
So, in the absence of postselection the error is bounded by $\mathcal{O}\bigg( \frac{\kappa}{t_0}\bigg)$, that is, $\norm*{P - \tilde{P}}_2 \leq \Ord{ \frac{\kappa}{t_0}}$.

In the introduction to the algorithm (section~\ref{subsubsec:summary}), we saw that in the final step of the algorithm, a computational basis measurement is performed on the state
\begin{equation}
\sum_j \beta_j |u_j\rangle^I \Big( \sqrt{1-\frac{C^2}{\lambda_j^2}} |0\rangle^S + \frac{C}{\lambda_j}|1\rangle^S \Big),
\end{equation}
\noindent
with postselection on measurement outcome `$1$', which thus resulted in the re-normalised state
\begin{equation}
|x\rangle = \frac{1}{\sum_j C^2 |\beta_j|^2 /|\lambda_j|^2} \sum_j \beta_j \frac{C}{\lambda_j}|u_j\rangle^I. 
\end{equation}
 We then introduced filter functions in order to ensure we were exclusively inverting the matrix on its well-conditioned subspace. In the ideal case, given a state $|b\rangle =\sum_j \beta_j |u_j\rangle^I$, the algorithmic procedure would then result in the state 
\begin{equation}
|\phi \rangle = \sum_j \beta_j |u_j\rangle^I |h(\lambda_j)\rangle^S,  
\end{equation}
\noindent
that is, replacing with the expression for $|h(\lambda_j)\rangle^S $ from Eq.~\eqref{eq:filter_funcs}
\begin{equation}
|\phi \rangle = \sum_j \beta_j |u_j\rangle^I \big(\sqrt{1-f(\lambda_j)^2 - g(\lambda_j)^2} |\text{nothing}\rangle^S  + f(\lambda_j) |\text{well}\rangle^S + g(\lambda_j)|\text{ill}\rangle^S \big). 
\end{equation}
\noindent
If we postselect on the inversion occurring, the final state is then given by 
\begin{equation}
|x\rangle = \frac{\sum_j \beta_j |u_j \rangle^I \big( f(\lambda_j )|\text{well} \rangle^S + g(\lambda_j) |\text{ill} \rangle^S  \big)}{\sqrt{p}},
\end{equation}
\noindent
 which occurs with probability $p=\sum_j |\beta_j|^2 (f(\lambda_j)^2+ g(\lambda_j)^2)$. We can choose to interpret this quantity as the random variable $(f(\lambda_j)^2 + g(\lambda_j)^2)$ occurring with probability $\operatorname{Pr}(j)=|\beta_j|^2$, and thus we have $p=\mathbb{E}_j[f(\lambda_j)^2 + g(\lambda_j)^2)]$.

In practice, the phase estimation procedure produces an error, and we have that the state after postselection is instead given by
\begin{equation}
|\tilde{x}\rangle = \frac{\tilde{P}^\dag \sum_j \beta_j |u_j\rangle^I \sum_k \alpha_{k|j} |k\rangle^C (f(\tilde{\lambda}_k) |\text{well}\rangle^S  + g(\tilde{\lambda}_k)|\text{ill}) \rangle^S }{\sqrt{\tilde{p}}},
\end{equation}
\noindent
where $\tilde{p}= \sum_{j,k} |\beta_j|^2 |\alpha_{k|j}|^2 (f(\tilde{\lambda}_k)^2 +g(\tilde{\lambda}_k)^2)$. Once again, we can choose to interpret the quantity $f(\tilde{\lambda}_k)^2 + g(\tilde{\lambda}_k)^2  $ as a random variable occurring with probability Pr$(j,k)=|\beta_j|^2 |\alpha_{j|k}|^2$, and thus write $\tilde{p}=\mathbb{E}_{j,k}[f(\tilde{\lambda})^2+g(\tilde{\lambda})^2]$.

The goal is to now bound the error on the postselected state, that is on $\norm{|x\rangle - |\tilde{x}\rangle}_2$. From Eq.~\eqref{equ:norm}, we can equivalently determine a lower bound on the fidelity $\langle \tilde{x}|x\rangle $. For clarity, the following notation is introduced: for a given value of $j$, let $\lambda := \lambda_j$,  $\tilde{\lambda}:= 2\pi k/t_0$, $f:= f(\lambda)$, $g:= g(\lambda)$, $\tilde{f}:= f(\tilde{\lambda})$ and $\tilde{g}:= g(\tilde{\lambda})$, and recall that $\delta =t_0(\lambda -\tilde{\lambda})$.

Applying these definitions and taking the inner product between the two states, we obtain
\begin{equation}
\langle \tilde{x}|x\rangle =\frac{\sum_{j,k}|\beta_j|^2 |\alpha_{j|k}|^2 ( \tilde{f} + \tilde{g})}{\sqrt{p\tilde{p}}},
\end{equation}
remembering that $\tilde{P}\ket*{x} = \ket*{x}$ from the definitions.
Using the random variable interpretation, we have:
\begin{equation}
\langle \tilde{x}|x\rangle=\frac{\mathbb{E}[\tilde{f}f+\tilde{g}g]}{\sqrt{p\tilde{p}}}.
\end{equation}
\noindent
From the property that the sum of expectations is the expectation of the sum, this can be re-expressed as:
\begin{equation}
\langle \tilde{x}|x\rangle=\frac{1+\frac{\mathbb{E}[(\tilde{f}-f)f + (\tilde{g}-g)g]}{p}}{\sqrt{1+\frac{\tilde{p}-p}{p}}}.
\end{equation}
\noindent
The denominator can be expanded as a Taylor series, yielding
\begin{equation}
\langle \tilde{x}|x\rangle \geq \Bigg( 1 + \frac{\mathbb{E}[(\tilde{f}-f)f + (\tilde{g}-g)g]}{p} \Bigg)\bigg(1 - \frac{1}{2}\frac{\tilde{p}-p}{p}\bigg). 
\end{equation}
\noindent
From the expression of $p$ and $\tilde{p}$, we have that:
\begin{equation}
\tilde{p}-p=\mathbb{E}[(\tilde{f}-f)^2] + \mathbb{E}[(\tilde{g}-g)^2].
\end{equation}
\noindent
Applying some algebra and the fact that the sum of expectations is expectation of sum, we have
\begin{equation}
\tilde{p}-p=2\mathbb{E}[(\tilde{f}-f)f] +2\mathbb{E}[(\tilde{g}-g)g] + \mathbb{E}[(\tilde{f}-f)^2]+ \mathbb{E}[(\tilde{g}-g)^2].
\end{equation}
\noindent
Substituting into the fidelity, we thus have
\begin{equation}
\langle \tilde{x}|x\rangle \geq 1- \frac{\mathbb{E}[(\tilde{f}-f)^2 + (\tilde{g}-g)^2]}{2p} - \frac{\mathbb{E}[(\tilde{f}-f)f + (\tilde{g}-g)g]}{p} \frac{\tilde{p}-p}{2p}.
\end{equation}
Next, we need a lower bound for both the second and third term of this equation. To do so, we introduce the following lemma:
\begin{lemma}[Errors on filter functions]\label{lem:lem6}
Let $f$, $\tilde{f}$, $g$, $\tilde{g}$ be defined as above with $\kappa '=2\kappa$. Then
\begin{equation}
 |f-\tilde{f}|^2 + |g-\tilde{g}|^2 \leq c \frac{\kappa^2}{t_0^2} \delta^2 |f^2 +g^2|, 
\end{equation}
where $c=\frac{\pi^2}{2}$.
\end{lemma}
\noindent The proof of this lemma essentially considers the three intervals on which the filter functions were defined, and can be found in the original paper. 

Now, by applying Lemma~\ref{lem:lem6} to the second term and substituting with the expression for $p$, we have
\begin{equation}\label{equ:square}
\frac{\mathbb{E}[(\tilde{f}-f)^2 + (\tilde{g}-g)^2]}{2p} \leq \mathcal{O}\bigg( \frac{\kappa^2}{t_0^2} \bigg) \frac{\mathbb{E}[(f^2 +g^2)\delta^2]}{\mathbb{E}[(f^2 +g^2)]},
\end{equation}
which, as $\mathbb{E}[\delta^2] \leq \Ord{1}$, gives
\begin{equation}
 \frac{\mathbb{E}[(\tilde{f}-f)^2 + (\tilde{g}-g)^2]}{2p} \leq \mathcal{O}\bigg( \frac{\kappa^2}{t_0^2} \bigg).
\end{equation}

Next, we consider the second term: first we apply the Cauchy-Schwarz inequality, yielding
\begin{equation}
\frac{\mathbb{E}[(\tilde{f}-f)f + (\tilde{g}-g)g]}{p} \leq \frac{\mathbb{E}[\sqrt{\big((\tilde{f}-f)^2 + (\tilde{g}-g)^2}\big) (f^2 + g^2)]}{p}.
\end{equation}
\noindent
Lemma~\ref{lem:lem6} can be applied to the right hand side, and we obtain
\begin{equation}\label{equ:nsquare}
\frac{\mathbb{E}[\sqrt{\frac{\delta^2 \kappa^2}{t_0^2} (f^2 + g^2)}]}{p}\leq \mathcal{O}\bigg(\frac{\kappa}{t_0} \bigg).
\end{equation}
We now consider the second part of the second term, $\frac{\tilde{p}-p}{p}$. Previously, we saw that
\begin{equation}
\tilde{p}-p=2\mathbb{E}[(\tilde{f}-f)f] +2\mathbb{E}[(\tilde{g}-g)g] + \mathbb{E}[(\tilde{f}-f)^2]+ \mathbb{E}[(\tilde{g}-g)^2].
\end{equation}
 Applying Eq.~\eqref{equ:square} and Eq.~\eqref{equ:nsquare}, we have
\begin{equation}
\frac{|\tilde{p}-p|}{p} \leq \mathcal{O}\bigg(\frac{\kappa }{t_0}\bigg). 
\end{equation}

\noindent We finally obtain
\begin{equation}
\langle \tilde{x}|x\rangle \geq 1-\mathcal{O}\bigg(\frac{\kappa^2}{t_0^2}\bigg).
\end{equation}
\noindent
Here we see why the parameter $t_0 = \Ord{\frac{\kappa}{\epsilon}}$: the quantity $\norm{\ket*{x}- \ket*{\tilde{x}}}_2 \leq \Ord{\epsilon}$ in this case.

In order to successfully perform the postselection step, we need to repeat the algorithm on average $\Ord{\kappa^2}$ times. 
This is because the success probability $\tilde{p} = \Ord{\kappa^2}$ (see Eq.~\eqref{eq:tilde_p})
This additional multiplicative
factor of $\Ord{\kappa^2}$ can be reduced to $\Ord{\kappa}$ using amplitude amplification (see section~\ref{subsec:aa}).

Putting everything together, we have an overall runtime that scales as $\Ord{\kappa^2 \cdot T_U \cdot \text{polylog}(N)/\epsilon}$, since we set $\delta = \Ord{\kappa/ \epsilon}$ in the phase estimation.
Taking into account the runtime of the Hamiltonian simulation $T_U = \Ord{s^2 \log N}$ an upper-bound of the runtime of the HHL is given by $\Ord{s^2 \kappa^2 \cdot \text{polylog}(N) /\epsilon}$.

 \subsection{Matrix inversion is $\BQP$-complete} \label{subsec:red}

The optimality of HHL is shown for a specific definition of the matrix inversion problem, which is defined as follows.
\begin{definition}\label{definition:MatInv}\emph{((Quantum) Matrix Inversion)} An algorithm solves matrix inversion if it has:\vspace{-0.3cm}
	\begin{itemize}
		\item\textbf{Input:} An $\Ord{1}$-sparse matrix $A$ specified using an oracle or via a $\poly{\log(N)}$-time algorithm that returns the non-zero elements of a row.
		\item\textbf{Output:} A bit that equals one with probability $\bra*{x} M\ket*{x}\pm\epsilon$, where $M=\ketbra*{0}\otimes I_{N/2}$ corresponds to measuring the first qubit and $\ket*{x}$ is a normalised state proportional to $A^{-1}\ket*{b}$ for $\ket*{b}=\ket*{0}$.
	\end{itemize}
	We also demand that $A$ be Hermitian and that $\kappa^{-1}I \preceq A \preceq I$.
\end{definition}
Despite this very weak definition for matrix inversion, this task is still classically hard.
Indeed, this problem is shown to be $\BQP$-complete, where $\BQP$ is the class of problems decidable in polynomial time on a quantum computer.
How is this shown?
\par
The authors of~\cite{harrow2009quantum} show that a quantum circuit using $n$ qubits and $T$ gates can be simulated by inverting an $\Ord{1}$-sparse matrix $A$ of dimension $N=O\qty(2^n \kappa)$.
Referring back to table~\ref{tab:LSAQLSAcompare}, we see that for $\kappa = \Ord{T^2}$ the conjugate gradient algorithm can perform this task in polynomial time.
This constitutes a classical algorithm simulating an arbitrary quantum computation in polynomial time, which is widely conjectured to be impossible.

The reduction proceeds as follows: let $\mathcal{C}$ be a quantum circuit acting on $n=\log N$ qubits, applying $T$ two-qubit gates $U_T \cdots U_1$.
The initial state is $\ket*{0}^{\otimes n}$, and the output is given by a measurement of the first qubit.
\par
We now adjoin an ancillary register of dimension $3T$, and define the unitary operator $U$ as
\begin{equation}\label{eq:HHLreductionUnitary}
U=\sum^{T}_{t=1}\ketbra*{t+1}{t}\otimes U_t + \ketbra*{t+T+1}{t+T}\otimes I +\ketbra*{t+2T+1\,\text{mod}\, 3T}{t+2T}\otimes U_{3T+1-t}^\dagger .
\end{equation}
This operator has been chosen such that for $T+1\leq t \leq 2T$, applying $U^t$ to the state $\ket*{1}\ket*{\psi}$ yields the output state $\ket*{t+1}\otimes U_T\cdots U_1 \ket*{\psi}$.
We can see this as the first $T+1$ applications of $U$ return $\ket*{T+2}\otimes U_{T}\cdots U_1 \ket*{\psi}$. We see from the second term of Eq.~\eqref{eq:HHLreductionUnitary} that for the next $t'\leq T-1$ applications the action on the $\ket*{\psi}$ register remains unchanged, while the ancillary variable is merely being incremented.
\par
We can thus now define the operator $A=I-U e^{-1/T}$, which gives $\kappa(A)=\Ord{T}$, by the following.
\begin{lemma}[Condition number of simulation matrix]
	Let $U$ be a unitary matrix, $T > 0$ and $A=I-U e^{-1/T}$.
	Then, the condition number of $A$, $\kappa(A) = \Ord{T}$.
\end{lemma}
\begin{proof}
	By the definition of condition number $\kappa(A) =$
	\begin{equation}\label{eq:kappaBQP}
		\displaystyle
		\frac{\lambda_{max}(A)}{\lambda_{min}(A)} \leq \frac{\max_{\norm{\vb{x}}_2 = 1} \vb{x}^\dagger ( I-U e^{-1/T} ) \vb{x} }{\min_{\norm{\vb{y}}_2 = 1} \vb{y}^\dagger ( I-U e^{-1/T} ) \vb{y}}
		=
		\frac{1 - \min_{\norm{\vb{x}}_2 = 1} \vb{x}^\dagger U \vb{x} \cdot e^{-1/T} }{1- \max_{\norm{\vb{y}}_2 = 1} \vb{y}^\dagger U \vb{y} \cdot e^{-1/T}}
		=
		\frac{1 - (-1) \cdot e^{-1/T} }{1- (+1) \cdot e^{-1/T}},
	\end{equation}
	since $U$ is unitary.
	Now let $f(x) = 2x - \frac{1 + e^{-1/x} }{1 - e^{-1/x}}$, where $x>0$.
	The Laurent expansion of $f(x)$ around $x = \infty$ is $\Theta(1/x)$, so we have that $\lim_{x \to \infty} f(x) = 0$.
	This means that the upper bound on $\kappa(A)$ in \ref{eq:kappaBQP} asymptotically tends to $2T$ and the result follows.
\end{proof}

The matrix $A=I-U e^{-1/T}$ can be expressed as
\begin{equation}\label{eq:HHLexpaInverse}
A^{-1} = \sum_{t\geq 0} U^t  e^{-k/T}.
\end{equation}
We can see Eq.~\eqref{eq:HHLexpaInverse} holds by multiplying by $A$, then observing all terms in the series apart from $k=0$ cancel.
We can also interpret this as applying $U^t$ for $t$ a exponentially distributed random variable, as in, applying $A$ has the same effect as drawing $t$ according to the exponential distribution with parameter $1/T$ then applying $U^t$.
As $U^{3T}=I$, we can assume $1\leq t\leq 3T$.
Measuring the first register and obtaining $T+1 \leq t\leq 2T$ occurs with probability $ e^{-2}/(1+ e^{-2}+ e^{-4})\geq 1/10$ (taking the amplitude of the time-step $\ket*{t}$ as $e^{-t/T}$ modulo normalisation).
If we obtain this result, the second register is left in the state $U_T\cdots U_1\ket*{\psi}$.
If we draw $t$ from the appropriate exponential distribution and apply $U^t$ many times to $\ket*{1}\ket*{0}$, followed by postselecting on $T+1 \leq t\leq 2T$, sampling the resulting state is equivalent to sampling $\ket*{x}=A^{-1}\ket*{0}$.

\par
For any quantum circuit $\mathcal{C}$ there exists an associated matrix $A_{\mathcal{C}}$ which, when inverted then sampled from, will in a fixed fraction of instances give the same results as sampling from the output of $\mathcal{C}$, Matrix Inversion (as defined in Definition~\ref{definition:MatInv})  is  $\mathsf{BQP}$-complete.

\subsection{Optimality of HHL} \label{subsec:optimality}

In order to simulate a quantum circuit up to some pre-determined accuracy $\epsilon$, the simulations need to be iterated.
That is, at iteration $i$, one uses the technique above to simulate the circuit at iteration $i-1$.
The authors show that if one can solve matrix inversion in time $\kappa^{1-\delta}(\log(N)/\epsilon)^{c_1}$ for constants $c_1\geq 2,\delta > 0$, then a computation with $T\leq 2^{2n}/18$ gates can be simulated with a polynomial number of qubits and gates.
The TQBF (totally quantified Boolean formula satisfiability) problem requires time $T\leq 2^{2n}/18$ by exhaustive enumeration and is $\mathsf{PSPACE}$-complete.
Thus a quantum algorithm for matrix inversion running in time $\kappa^{1-\delta}\poly{\log N}$ could efficiently solve TQBF, thus implying $\mathsf{PSPACE}= \mathsf{BQP}$.

\subsection{Non-Hermitian Matrices} \label{subsec:non-hermitian}

We now consider the case where $A$ is neither square nor Hermitian.
Let us now suppose that we have $A\in\mathbb{C}^{m\times n}$, $\vb{b}\in\mathbb{C}^m$ and $\vb{x}\in\mathbb{C}^n$ with $m\leq n$.
Generically the linear system $A\vb{x} =\vb{b}$ is under-constrained.
The singular value decomposition of $A$ is given by:
\begin{equation}
A=\sum_{j=1}^m \sigma_j \ketbra*{u_j}{v_j},
\end{equation}
where $\ket*{u_j}\in\mathbb{C}^m$, $\ket*{v_j}\in\mathbb{C}^n$ and $\sigma_1\geq\cdots\geq\sigma_n\geq 0$. Next, we can define the operator $H$ as
\begin{equation}\label{eq:LinSys:nonhermitian}
H:=\sum^{m}_{j=1}\sigma_{j}\qty(\ketbra*{0}{1}\otimes\ketbra*{u_j}{v_j}+\ketbra*{1}{0}\otimes\ketbra*{v_j}{u_j})\equiv \mqty(0 & A \\ A^\dagger & 0),
\end{equation}
where we now have that the matrix $H$ is Hermitian with eigenvalues $\pm \sigma_1,\ldots,\pm\sigma_m$, corresponding to eigenvectors $\ket*{w_j^\pm}:=\frac{1}{\sqrt{2}}(\ket*{0}\ket*{u_j}\pm\ket*{1}\ket*{v_j})$.
This can be verified by direct substitution into Eq.~\eqref{eq:LinSys:nonhermitian}.
The matrix $H$ also has $n-m$ zero eigenvalues, with eigenspace $V^\perp$, where $V:=\operatorname{span}(\ket*{w_1^{\pm}},\ldots,\ket*{w_m^{\pm}})$.
\par
The HHL algorithm can now be applied to  the input state $\ket*{0}\ket*{b}$.
If $\ket*{b}=\sum^m_{j=1}\beta_j\ket*{u_j}$ then we have
\begin{equation}
\sum_{j=1}^m \beta_j\frac{1}{\sqrt{2}}\qty(\ket*{w_j^+}+\ket*{w_j^-}),
\end{equation}
and upon running of the algorithm we get
\begin{equation}
H^{-1}\ket*{0}\ket*{b}=\sum^m_{j=1}\beta_j\sigma_1^{-1}\frac{1}{\sqrt{2}}\qty(\ket*{w_j^+}-\ket*{w_j^-})=\sum^m_{j=1}\beta_j\sigma_j^{-1}\ket*{1}\ket*{v_j}.
\end{equation}
If we discard the $\ket*{1}$, we can recognise this expression as the Moore-Penrose pseudo-inverse of $A$ applied to $\ket*{b}$. Thus, on non-square, non-Hermitian $A$, HHL outputs the state $\ket*{x}:=A^+\ket*{b}$.
\par
Now in the over-constrained case where $m\geq n$, the equation $A\ket*{x}=\ket*{b}$ is only satisfiable if $\ket*{b}\in U$, where $U=\operatorname{span}(\ket*{u_1},\ldots,\ket*{u_n})$.
In this case, applying $H$ to $\ket*{0}\ket*{b}$ will return a valid solution.
Otherwise the state $\ket*{x^*}$ is returned, which minimises the squared error loss, i.e. $\ket*{x^*}:=\arg\min_{\ket*{x}\in\mathbb{R}^n}\norm{A\ket*{x}-\ket*{b}}_2^2$, which is again the Moore-Penrose pseudoinverse. Formally, this is expressed in Definition \ref{definition:NM:pinv}.

\begin{definition}[Moore-Penrose pseudoinverse]\label{definition:NM:pinv}
	Let $A\in\mathbb{C}^{m\times n}$ have singular value decomposition given by $A = U \varSigma V^\dagger$ and be of rank $r$, with $U$ and $V$ unitary matrices.
	The \emph{Moore-Penrose pseudoinverse} of $A$ is given by
	\begin{equation}
	A^+ = V\varSigma^+ U^\dagger,
	\end{equation}
	where $\varSigma^+ = \operatorname{diag}(\frac{1}{\sigma_1},\ldots, \frac{1}{\sigma_r},0,\ldots,0)$ for the singular values $\sigma_i, i\in[r]$.
	\par
	Further to this,
	\begin{equation}
	A^+ =
	\begin{cases}
	\qty(A^\dagger A)^{-1}A^\dagger, & \text{if }m\geq n; \\
	A^\dagger \qty(A A^\dagger)^{-1}, & \text{if } m \leq n.
	\end{cases}
	\end{equation}
\end{definition}

\section{Improvements to the HHL}\label{sec: improvements HHL}

In this section we briefly discuss several improvements of the HHL algorithm. The first two of these improvements are in terms of the runtime complexity in the condition number $\kappa$ and the precision $\epsilon$ respectively.
The third and final improvement removes the sparsity requirement, that is a algorithm for QLSP on \emph{dense} input matrices.

Ambainis~\cite{ambainis2010variable} reduced the condition number dependence from $\kappa^2$ to $\kappa \log^3 \kappa$ which is close to the lower bound of $\kappa$ discussed in section~\ref{subsec:optimality}.
This is a consequence of the reduction we demonstrated in section~\ref{subsec:red}.
We discuss this improvement in section~\ref{subsec:variable time}.
Further work by Childs et al.~\cite{childs2015quantum} reduced the precision number dependency of the algorithm from $\Ord{\text{poly}(1/\epsilon)}$ to $\Ord{\text{poly}\log(1/\epsilon)}$, see section~\ref{subsec:exp. improved}.
Wossnig, Zhao and Prakash give an algorithm~\cite{wossnig2017quantum} for QLSP taking time $\Ord{\kappa^2 \sqrt{n} \operatorname{polylog}{n} / \epsilon}$, a polynomial speedup for dense matrices (where $s = \Ord{n}$).
This algorithm is based on a quantum singular value estimation (QSVE) procedure; we discuss this in section~\ref{sec:QSVE}

For the remainder of the text, we depart somewhat from the pedagogical style developed thus far in the text. Specifically, we shall cover the general ideas behind these improvements and describe the new contributions, but refer the reader to the original papers for detailed proofs.

\subsection{Variable time amplitude amplification for improved condition number dependence}\label{subsec:variable time}

The HHL algorithm has two parts which contribute a linear factor of $\kappa$ to the runtime. The first one is the eigenvalue estimation part, which requires a certain precision and hence a repeated number of applications of the Hamiltonian simulation $\exp (-i\hat{H}t)$. If all of the eigenvalues $\lambda_i$ are small, i.e.\ of magnitude $\Omega (1/\kappa)$ as a lower bound, then $| \lambda_i - \tilde{\lambda}_i | = \Omega (\epsilon/\kappa)$, where $\tilde{\cdot}$ denotes the estimate obtained via phase estimation, since we require $| \lambda_i - \tilde{\lambda}_i | \leq \epsilon \tilde{\lambda}_i$, similar to our previous discussion about the quantum Fourier transform and phase estimation (see sections~\ref{subsec:QFT}~and~\ref{subsec:phase est}).

The second contribution of $\kappa$ in the runtime comes from amplitude amplification (AA) and the postselection step.
If all of the eigenvalues $\lambda_i$ are large, let's say constant, then the coefficients $1/ \kappa \lambda_i$ can be of order $1/\kappa$, which in turn results in a number of $\Ord{\kappa}$ repetitions leveraging AA.
If now all eigenvalues have roughly similar order of magnitude, for example let $\lambda \in [a, 2a]$ for some constant $a$, then eigenvalue estimation requires $\Ord{1/a\epsilon}$ time and AA takes $\Ord{a \kappa}$ time, which leads to a total of $\Ord{\kappa/\epsilon}$ runtime and an induced error of $\epsilon$.\\

The main idea to improve the $\kappa$-dependency is that a quantum algorithm can have variable stopping times. Let $\mathcal{A}$ be an quantum algorithm that can stop at $m$ different times $t_1, \ldots, t_m$, with an indicator qubit that is set to $1$ whenever the branch stops in an additional clock register (initialised to `$0$'). More specifically, if it stops at time $t_j$, then it sets the $j$-th qubit to `$1$'. Any subsequent operation can then not affect the stopped branches anymore. In the HHL algorithm this can be seen as runs of the eigenvalue estimation with increasing precision and increasing number of steps, as one can recall that we apply the unitary repeatedly for different powers in the phase estimation step.

In this way we can formalise the above idea of of a variable stopping time algorithm by requiring that $\mathcal A$ can be written as a product of $m$ algorithms $\mathcal{A} = \mathcal{A}_m \mathcal{A}_{m-1} \cdots \mathcal{A}_1$, where each $\mathcal{A}_j$ is a controlled unitary that acts on the register $\mathcal{H}_{C_j} \otimes \mathcal{H}_A$ and controlled by the first $j-1$ qubits of $\mathcal{H}_C$ being set to $0$.
Hence $\mathcal{A}$ acts on $\mathcal{H}_{C} \otimes \mathcal{H}_A$, where $\mathcal{H}_{C} = \bigotimes_j \mathcal{H}_{C_j}$. We have a decomposition of the space into a control register and a target register.

We can also view $\mathcal{A}$ as an algorithm that produces $\ket*{\Psi_{succ}}$ with probability $\tilde{p} = | \alpha_{1} |^2$, i.e.\ we obtain a state
\begin{equation}
\alpha_{0} \ket*{0} \otimes \ket*{\Psi_{fail}} + \alpha_{1} \ket*{1} \otimes \ket*{\Psi_{succ}},
\end{equation}
which results in the desired output state $\ket*{\Psi_{succ}}$ with probability $\tilde{p}$ when we measure the ancillary register.

Ambainis~\cite{ambainis2010variable} introduced the \textit{variable time amplitude amplification} (VTAA) scheme which allows to amplify the success probability of quantum algorithms in which some branches of the computation stop earlier than other branches. Conventional AA stops only once all computational branches have finished, which takes time $\Ord{T_{max}/\sqrt{\tilde{p}}}$ and which can thus lead to a substantial inefficiency in certain cases. VTAA achieves a improved runtime of
\begin{equation}
\Ord{T_{max} \sqrt{\log T_{max}} + \frac{T_{av}}{\sqrt{\tilde{p}}} \log ^{1.5} T_{max} },
\end{equation}
\noindent
where $T_{max}$ is the maximum runtime of a branch of the algorithm, $T_{av}= \sqrt{\sum_i p_i t_i^2}$ is the average time over all the branches and we defined $p_i$ as the probability of the algorithm stopping at time $t_i$. This can thus be substantially faster if $T_{av} \ll T_{max}$.

The high level idea behind this result is that quantum amplitude amplification is  split into a variable runtime algorithm. For each branch the algorithm tests whether the computation has stopped or not, and then applies a conditional operation only on the branches which are not yet finished. If the branch didn't stop, then the algorithm first checks whether the success probability surpasses a certain threshold using amplitude estimation, and depending on the outcome either: uses amplitude amplification to first boost the success probability and then applies the algorithm again recursively; or just applies the algorithm recursively if the success probability is already high enough. In this way the branches which haven't achieved the required accuracy yet are executed for longer periods of time, while other ones can stop early.

Applying this scheme to the quantum linear system algorithm, Ambainis~\cite{ambainis2010variable} reduces the $\kappa$-dependence in the HHL algorithm runtime from $\Ord{\kappa^2}$ to $\Ord{\kappa \log^3 \kappa}$, neglecting all other parameters which essentially remain the same.

\subsection{Exponentially improved precision number dependence} \label{subsec:exp. improved}

Let $k$ be the number of bits of precision we desire for the output of a particular algorithm. The precision number in this case is then related to the number of bits $k$ by $\epsilon = 1/10^k$ (assuming no other sources of error in the output). 
We would like to have an algorithm that grows linearly with the number of bits, that is, for every extra bit of precision added to the representation of our data we want the runtime to increase additively, not multiplicatively.
For this to be satisfied, the runtime dependence on precision needs to be $\Ord{\text{poly}\log(1/\epsilon)}$.

Indeed, the HHL algorithm does not satisfy this desideratum.
However, an algorithm has been devised by Childs, Kothari and Somma that solves QLSP with polylogarithmic dependence on $\epsilon$
\cite{childs2015quantum}. 
The main contribution is a framework that allows one to circumvent the phase estimation procedure which has an inherent $1/\epsilon$ dependency.
Let us recap the HHL algorithm, using the best known result for black box Hamiltonian simulation~\cite{berry2015hamiltonian}.

\begin{theorem}[HHL algorithm] \label{thm:hhl}
  The quantum linear system problem can be solved by a gate-efficient algorithm (i.e. a algorithm with complexity $Q$, the gate complexity is only by logarithmic factors larger than $Q$) that makes $\Ord{\kappa^2 s \poly{\log (s \kappa/\epsilon) / \epsilon}}$ queries to the oracles of $A$ and $\Ord {\kappa s \poly{ \log (s \kappa/\epsilon)} / \epsilon}$ queries to the oracle to prepare the state corresponding to $\mathbf b$. Using a form of qRAM for the data access a multiplicative factor of $\Ord{\log N}$ is added in the runtime of the algorithm.
\end{theorem}

The main idea of the approach is to use a technique for implementing linear combinations of unitary operations (LCUs) based on a decomposition of the operators using Fourier and Chebyshev series.
We now present the method of LCUs by considering an example similar to the original work~\cite{childs2015quantum}. Consider the operator $V = U_0 + U_1$, where $U_i$ are unitary operators that are easy to implement. We start with the state $\ket*{\Phi}$ and want to apply $V$ to this state.
First we add another ancillary qubit in the $|+\rangle$ state, and then apply a conditional $U_i$, i.e. $\ket*{0}\bra*{0} \otimes U_0 + \ket*{1}\bra*{1} \otimes U_1$ to the input state $|+\rangle |\Phi\rangle$. This leads to the state
\begin{equation}
\frac{1}{\sqrt{2}} (\ket*{0} U_0 \ket*{\Phi} + \ket*{1} U_1 \ket*{\Phi}).
\end{equation}

Measuring the first qubit in the $x$-basis would then result in the application of $V$ if we obtain the measurement outcome $\ket*{+}$.
If we have the ability to prepare $\ket*{\Phi}$ multiple times or reflect about $\ket*{\Phi}$, then we can use amplitude amplification to boost the success probability of this measurement.

For a more general statement we can generalise this to an arbitrary unitary $V$. Let $U:= \sum_i \ketbra*{i} \otimes U_i$ implement the conditioned $U_i$ operation, and let $V = \sum_i \alpha_i U_i$, where $\alpha_i >0$ without loss of generality since we can absorb the sign into the phase of $U_i$. Then we can apply $V$ using the following trick. Let
\begin{equation}
W : \ket*{0^m} \rightarrow \frac{1}{\sqrt{\alpha}} \sum_i \sqrt{\alpha_i} \ket*{i},
\end{equation}
 where $\alpha := \sum_i \alpha_i$.
We can apply the operator $V$ with high probability using the operator $M=W^{\dagger}UW$, since $W$ first creates a superposition of the control register, then applies the conditional unitary operators $U_i$, and finally performs the uncomputation of the superposition in the first register which results in a superposition of a state proportional to the all-zeros and a state that is orthogonal to it:
\begin{equation}
M \ket*{0^m} \ket*{\Phi} = \frac{1}{\alpha}\ket*{0^m} V \ket*{\Phi} + \ket*{\Psi^{\perp}},
\end{equation}
 where $(\ket*{0^m} \bra*{0^m} \otimes I)\ket*{\Psi^{\perp}} =0$, i.e.\ we have collected all orthogonal terms in the latter state, $\ket*{\Psi^{\perp}}$. This can further be generalised to non-unitary operations, but we omit this here for the sake of simplicity.
The probabilistic implementation here is successful if we measure the state $\ket*{0^m}$, i.e. after successful postselection, the successful outcome occurs with probability $(||V \ket*{\Phi}||/\alpha)^2$. Here again we can repeat the process $\Ord{(\alpha/||V \ket*{\Phi}||)^2}$ times or use amplitude amplification to boost the probability quadratically, i.e.\ using only $\Ord{\alpha/||V \ket*{\Phi}||}$.

Let us assume for now that we can find a LCU that approximated the operator we desire. In order to apply this scheme to the quantum linear system problem (QLSP), we need to find a LCU that closely approximates $A^{-1}$. Now, it remains to consider how the unitary operators can be applied to the state. The commonly used method for this is Hamiltonian simulation, which uses unitary operators of the form $\exp(-iAt)$. This means we need to find a LCU that closely approximates $A^{-1}$ based on the decomposition $\sum_j \alpha_j \exp(-iAt_j)$ for some coefficients $\alpha_j$ and evolution times $t_j$. 
Since both sides of the equation are diagonal in the same basis, obtaining this representation is equivalent to obtaining $x^{-1}$ as a linear combination of $\sum_j \alpha_j \exp(-ix t_j)$, where $x \in \mathbb{R}\setminus \{0\}$.
This is because we can diagonalise $A$ and then work with the eigenvalues directly (recall that we showed in the Hamiltonian simulation section~\ref{subsec:ham_sim} that $\exp(-iU\hat{H}U^{\dagger}t) = U \exp(-i\hat{H}t) u^{\dagger}$). 

Similar to the HHL algorithm we restrict the range of eigenvalues to be inverted to lie in the range $D_{\kappa} := [-1, -1/\kappa] \cup [1/\kappa, 1]$, i.e.\ the well-conditioned subspace, so we only require the representation to be correct in this interval (we have also assumed that $\lambda_{\text{max}} = 1$, which can be enforced by a simple re-scaling).
The strategy is to smooth the function $1/x$ around the singular points ($0$) and then perform a Fourier transform over various $\exp(-ixt)$, since this will give us a sum over feasible functions.
The idea is to use an integral representation of the function $1/x$ which is given by \begin{equation}
\frac{1}{x} = \frac{i}{\sqrt{2\pi}} \int_0^{\infty} dy \int_{- \infty}^{+\infty} dz\ z e^{-z^2/2}e^{-ixyz},
\end{equation}
then take a finite-sum approximation to this integral.

The result is an $\epsilon$-close Fourier expansion of the function $1/x$ on the interval $D_{\kappa}$ in which the eigenvalues are well behaved, i.e. do not come too close to $0$. We can summarize this in the following lemma.
\begin{lemma}[\cite{childs2015quantum}, Lemma 11] \label{lem:childs}
Let the function $h(x)$ be defined as
\begin{equation}
h(x) := \frac{i}{\sqrt{2 \pi}} \sum_{j=0}^{J-1} \delta_y \sum_{k=-K}^K \delta_z z_k e^{-z_k^2/2} e^{-ixy_jz_k},
\end{equation}
where $y_j := j \delta_y$, $z_k := k \delta_z$, for some $J=\Omega(\frac{\kappa}{\epsilon} \log(\kappa/\epsilon))$, $K = \Omega(\kappa \log(\kappa/\epsilon))$, $\delta_y = \Omega(\epsilon/ \sqrt{\log (\kappa/\epsilon)})$ and $\delta_z = \Omega((\kappa \sqrt{\log (\kappa/\epsilon)})^{-1})$. Then $h(x)$ is $\epsilon$-close to $1/x$ on the domain $D_{\kappa}$.
\end{lemma}

We still need to clarify the question of whether a operator which is close to $A^{-1}$ yields a state which is also close to our target state. For the case of the QLSP we want that this operator applied to our target, i.e.\ $\tilde{A}^{-1}\ket*{b}/||\tilde{A}^{-1}\ket*{b}||$ is close to $A^{-1}\ket*{b}/||A^{-1}\ket*{b}||$. The following can be  shown using the triangle inequality:\\
For any Hermitian operator $A$, with $||A^{-1}|| \leq 1$ (i.e. the smallest eigenvalue of $A$ in absolute value is at least $1$) and a approximation to it which satisfies $|| A -\tilde{A}|| \leq \epsilon <1/2$, the states \begin{equation}
    \ket*{x}:= A \ket*{b}/||A \ket*{b}||
\end{equation}
 and 
\begin{equation}
    \ket*{\tilde{x}}:= \tilde{A} \ket*{b}/||\tilde{A} \ket*{b}||
\end{equation}
satisfy
\begin{equation}
  || \ket*{x} - \ket*{\tilde{x}}|| < 4 \epsilon.
\end{equation}

Applying this to the scheme described above the authors of~\cite{childs2015quantum} find that any function can be implemented $\epsilon$-close using the LCU approach.
Using this approximation of the function of $x$ for the matrix $A$ discussed above and the results pertaining to the approximation error, we then have a quantum linear system algorithm with the following properties.

\begin{theorem}[\cite{childs2015quantum} Fourier approach] \label{thm:childs}
The QLSP can be solved with $O\qty( \kappa \sqrt{\log(\kappa/\epsilon)})$ uses of a Hamiltonian simulation algorithm that approximates $\exp(-iAt)$ for $t=O \qty(\kappa \log(\kappa/\epsilon))$ with precision $O\qty( \epsilon/(\kappa \sqrt{\log (\kappa / \epsilon)}))$.
\end{theorem}
Using the best known results for the Hamiltonian simulation algorithm (see section~\ref{subsec:ham_sim}), one obtains a algorithm that is polylogarithmic in the precision, $\epsilon$.

A similar approach is possible using a decomposition of the function via Chebyshev polynomials~\cite{childs2015quantum}. The benefit here is that the Chebyshev polynomials can be implemented via quantum walks, which was demonstrated by Childs~\cite{childs2010relationship}. This leads to a slightly more efficient implementation of the operators, but requires explicit access to the entries of $A$.
Considering a careful modification of the techniques, it is further possible to include the improvement by Ambainis and construct an algorithm solving QLSP which requires $O \qty (\kappa \cdot s\cdot  \text{poly} \log (s \kappa/\epsilon))$ calls to the data oracle~\cite{childs2015quantum}.\\

We will now demonstrate a last improvement which is an alternative approach to solving linear systems of equations using the so called Szegedy quantum walk~\cite{szegedy2004quantum}. This will lead to an improvement in the dimensionality of the problem even for dense matrices, as is the almost fastest algorithm in terms of the dimensionality. Kerenidis and Prakash improved this scaling subsequently by a small factor and applied it to a gradient descent algorithm with affine gradients~\cite{kerenidis1704quantum}.

\renewcommand{\R}{\ensuremath\mathbb{R}}

\subsection{QSVE based QLSA} \label{sec:QSVE}
The HHL algorithm has runtime scaling as $\Ord{\log(N) s^2 \kappa^2 / \epsilon}$, and the subsequent algorithms discussed so far scale polynomially in sparsity $s$.
In contrast, we recall that the classical conjugate-gradient algorithm scales as $\Ord{N s \kappa \log(1/\epsilon)}$~\cite{Shewchuck1994}. For dense matrices, where $s \sim N$, we have that the quantum algorithm achieves an only quadratic speed-up in system size $N$ over the classical version.
A slight improvement to this is still possible, even in the case of dense matrices, by using qRAM in combination with quantum walks \cite{wossnig2017quantum}.

This algorithm is built around a quantum singular value estimation (QSVE) procedure applied to $A$, which was introduced in \cite{kerenidis2016quantum}.
The QSVE algorithm requires the ability to efficiently prepare the
quantum states corresponding to the rows and columns of matrix $A$. The matrix entries are stored in a data structure that must enable the quantum algorithm with access to this data structure to perform the operations we define below. Note that this structure can be any form of qRAM (see section~\ref{subsec:qram}).

Let $A \in \R^{m \times n}$ be the matrix with entries $A_{ij}$. The required data structure must be able to perform the following tasks.

\begin{itemize}
	\item A quantum computer with access to the data structure can perform the following mappings in $\Ord{\text{polylog}(mn)}$ time.
	\begin{align}
	U_\mathcal{M}: \ket*{i}\ket*{0}\rightarrow\ket*{i,\mathbf{A_i}} &= \frac{1}{\|\mathbf{A_i}\|}\sum\limits_{j=1}^{n}A_{ij}\ket*{i,j},\notag \\
	U_\mathcal{N}: \ket*{0}\ket*{j}\rightarrow\ket*{\mathbf{A}_F,j} &= \frac{1}{\|A\|_F}\sum\limits_{i=1}^{m}\|\mathbf{A_i}\|\ket*{i,j},
	\end{align}
	where $\mathbf{A_i}\in \mathbb R^{n}$ corresponds to the $i^\text{th}$ row of the matrix $A$ and $\mathbf{A}_F\in \mathbb R^{m}$ is a vector whose entries are the $\ell_{2}$ norms
	of the rows, i.e.\ $(\mathbf{A}_F)_i=\norm{A_i}_2$.
\end{itemize}
This data structure can be realized with the qRAM architecture presented in section \ref{subsec:qram} by storing each row/column in a separate tree.

The QSVE algorithm is a quantum walk based algorithm that makes use of a well known connection between the singular values $\sigma_i$ of the target matrix $A$ and the principal angles $\theta_i$ between certain subspaces associated with $A$. We will review here some of the necessary ingredients.

\subsubsection{Singular values and subspaces}
\label{subsec:svd_and_subspaces}
First, the singular value decomposition is formally introduced.
\begin{theorem}[Singular Value Decomposition (SVD)] \label{thm:svd}
For $A \in \mathbb R ^{m \times n}$ there exists a decomposition of the form
$A = U \Sigma V^{\dagger}$, where $U \in \mathbb{R}^{m \times m}$, $V \in \mathbb{R}^{n \times n}$ are unitary operators, and $\Sigma\in \mathbb{R}^{m \times n}$ is a diagonal matrix
with $r$ positive entries $\sigma_1 ,\sigma_2, \ldots, \sigma_r$, and $r$ is the rank of $A$. Alternatively, we write $A=\sum_{i}^{r}\sigma_i\mathbf{u}_i\mathbf{v}_i^{\dagger}$, where $\{\mathbf{u}_i\}$, $\{\mathbf{v}_i\}$, and $\{\sigma_i\}$ are the sets of left and right mutually orthonormal singular vectors, and singular values of $A$ respectively.
\end{theorem}
The Moore-Penrose pseudo-inverse of a matrix $A$ with singular value decomposition as defined above is then given by $A^+ = V \Sigma ^{+} U^{\dagger} = \sum_{i}^{r}(1/\sigma_i)\mathbf{v}_i \mathbf{u}_i^{\dagger}$. The matrix $AA^+$ is the projection onto the column space $Col(A)$, while $A^+A$ is the projection onto the row space $Row(A)$.

Let the matrix $A \in \R^{m \times n}$ (with Frobenius norm of $1$) now have the factorization and SVD $A= \mathcal{M}^{\dagger} \mathcal{N} = \sum_i \sigma_i \mathbf u_i \mathbf v_i^{\dagger}$. Then, it can be shown that $\cos \theta_i = \sigma_i$, where $\theta_i$ is the principal angle between the subspaces. Note that the columns of $\mathcal{M}$ and $\mathcal{N}$ are orthogonal and we define the projector into the column space $Col(\mathcal{M})$ as
\begin{equation}
\Pi_1 = \mathcal{M} \mathcal{M}^{\dagger} = \sum_{i \in [m]} \ket*{m^i} \bra*{m^i},
\end{equation}
 where $m^i$ denote the columns of $\mathcal M$. Similarly, the projector onto the column space of $\mathcal{N}$ is given by  $\Pi_2 = \mathcal{N} \mathcal{N}^{\dagger}$.

The eigenvectors of the operator $\Pi_1 \Pi_2 \Pi_1$ are then simply given by $\mathcal{M} \mathbf u_i$ with eigenvalues $\sigma_i^2$, and the eigenvectors of $\Pi_2 \Pi_1 \Pi_2$ are given by $\Pi_2 \mathcal{M} \mathbf u_i$ with eigenvalues $\sigma_i^2$, which can be checked using the identity $\mathcal{M}^{\dagger}\mathcal{M} = I$, from the orthonormality of the columns.
The principal vector pairs are given by $(\mathcal{M} \mathbf u_i, \mathcal{N} \mathbf v_i)$, since $\Pi_2 \mathcal{M} \mathbf u_i = \sigma_i \mathcal{N} \mathbf v_i$ and the principal angles are $\cos \theta_i = \langle \mathcal{M} \mathbf u_i \ket*{\mathcal{N} \mathbf v_i} =\sigma_i$. We will make use of this in the following.

\subsubsection{The QSVE algorithm}
\label{subsec:qsve}

The QSVE algorithm makes use of the factorization $\frac{A}{\|A\|_F}= \mathcal{M}^{\dagger}\mathcal{N}$, where $\mathcal{M}\in\mathbb{R}^{mn\times m}$ and $\mathcal{N}\in\mathbb{R}^{mn\times n}$ are unitary operators.
The key idea is that the unitary operator $W$, a combination of reflections, defined by $W=(2\mathcal{M}\mathcal{M}^{\dagger}-I_{mn})(2\mathcal{N}\mathcal{N}^{\dagger}-I_{mn})$, where $I_{mn}$ is the $mn \times mn$ identity matrix.
The operator $W$ can be implemented efficiently using a qRAM (see section \ref{subsec:qram}).
The operator $W$ has two-dimensional eigenspaces spanned by $\{\mathcal{M}\mathbf{u}_i,\mathcal{N}\mathbf{v}_i\}$ on which it acts as a rotation by angle $\theta_i$, such that $\cos\frac{\theta_i}{2}=\frac{\sigma_i}{\|A\|_F}$. Note that the subspace spanned by $\{\mathcal{M}\mathbf{u}_i,\mathcal{N}\mathbf{v}_i\}$ is also spanned by $\{\mathbf{w}_i^+,\mathbf{w}_i^-\}$, the eigenvectors of $W$ with eigenvalues $\exp(i\theta_i)$ and $\exp(-i\theta_i)$ respectively. We get this result by similarly decomposing the space in terms of the eigenvectors of the rotation matrix, since they also span the two-dimensional space.
In particular we may write the following decomposition, $\ket*{\mathcal{N}\mathbf{v}_i}=\omega_i^+\ket*{\mathbf{w}_i^+}+\omega_i^-\ket*{\mathbf{w}_i^-}$, where $|\omega_i^-|^2+|\omega_i^+|^2=1$.
Note we are now using the notation $\ket{\vb{v}}$ to represent the quantum amplitude encoding of a classical vector $\vb{v}$.
To use the language of section~\ref{subsec:qram}, $\ket{\vb{v}} = \mathcal{R}(\vb{v})$.

Algorithm 1 describes the QSVE algorithm, the analysis of which is contained in the following lemma. Note that the idea is similar to the process in amplitude amplification, where the rotation towards the desired state is performed iteratively.

\renewcommand{\vec}[1]{\mathbf{#1}}

\begin{algorithm}[!htbp]
	\caption{Quantum singular value estimation. \cite{kerenidis2016quantum}}
	\label{alg:qsve}
	\begin{enumerate}
		\item Create the arbitrary input state $\ket*{\alpha} = \sum_i \alpha_{\vec v_i} \ket*{\vec v_i}$.
		\item Append a register $\ket*{0^{\lceil \log{m} \rceil}}$ and create the state $\ket*{\mathcal{N} \alpha} = \sum_i \alpha_{\vec v_i} \ket*{\mathcal{N} \vec v_i}=\sum_i \alpha_{\vec v_i}(\omega_i^+\ket*{\mathbf{w}_i^+}+\omega_i^-\ket*{\mathbf{w}_i^-})$.
		\item Perform phase estimation with precision $2 \delta >0$ on input $\ket*{\mathcal{N}\alpha}$ for $W=(2\mathcal{M}\mathcal{M}^{\dagger}-I_{mn})(2\mathcal{N}\mathcal{N}^{\dagger}-I_{mn})$ and obtain $\sum_i \alpha_{\vec v_i}(\omega_i^+\ket*{\mathbf{w}_i^+,\tilde{\theta}_i}+\omega_i^-\ket*{\mathbf{w}_i^-,-\tilde{\theta}_i})$, where $\tilde{\theta_i}$ is
		the estimated phase $\theta_i$ in binary bit-strings.
		\item Compute $\tilde{\sigma}_i = \cos{(\pm\tilde{\theta_i}/2)}||A||_F$.
		\item Uncompute the output of the phase estimation and apply the inverse transformation of step (2) to obtain
		\begin{equation}
		\sum\limits_i \alpha_{\vec v_i} \ket*{\vec v_i} \ket*{\tilde{\sigma_i}}
		\end{equation}
	\end{enumerate}
\end{algorithm}

\begin{lemma}[Preparation of the unitary operators \cite{kerenidis2016quantum}]
	\label{lem:unitaries}
	Let $A\in \mathbb R^{m\times n}$ be a matrix with singular value decomposition $A = \sum_i \sigma_i \vec u_i \vec v_i^{\dagger}$ stored in any data structure that has the abilities described above. Then there exist matrices $\mathcal{M} \in \R^{mn \times m}$, and $\mathcal{N} \in \R^{mn \times n}$, such that
	\begin{enumerate}
		\item $\mathcal{M}, \mathcal{N}$ are unitary operators, that is $\mathcal{M}^{\dagger}\mathcal{M} = I_m$ and $\mathcal{N}^{\dagger}\mathcal{N} = I_n$ such that $A$ can be factorised as $A/\norm{A}_F = \mathcal{M}^{\dagger}\mathcal{N}$.

		Multiplication by $\mathcal{M},\mathcal{N}$, i.e.\ the mappings $\ket*{\alpha} \rightarrow \ket*{\mathcal{M}\alpha}$ and $\ket*{\beta} \rightarrow \ket*{\mathcal{N}\beta}$ can be performed in time $\Ord{\operatorname{polylog}(mn)}$.
		\item The reflections $2\mathcal{M}\mathcal{M}^{\dagger}-I_{mn}$, $2\mathcal{N}\mathcal{N}^{\dagger}-I_{mn}$, and hence the unitary $W=(2\mathcal{M}\mathcal{M}^{\dagger}-I_{mn})(2\mathcal{N}\mathcal{N}^{\dagger}-I_{mn})$ can be implemented in time $\Ord{\operatorname{polylog}(mn)}$.
		\item The unitary $W$ acts as rotation by $\theta_{i}$ on the two dimensional invariant subspace $\{\mathcal{M}\mathbf{u}_i,\mathcal{N}\mathbf{v}_i\}$ plane, such that $\sigma_i = \cos\frac{\theta_i}{2}\|A\|_F$, where $\sigma_{i}$ is the $i$-th singular value for $A$.
	\end{enumerate}
\end{lemma}

We outline the ideas involved in the analysis of the QSVE algorithm and refer to \cite{kerenidis2016quantum} for further details. The map $\mathcal{M}$ appends to an arbitrary input state vector $\ket*{\alpha}$ a register that encodes the row vectors $\vec{A_i}$ of $A$, such that
\begin{align}
\mathcal{M}: \ket*{\alpha}&=\sum\limits_{i=1}^{m}\alpha_i\ket*{i}\rightarrow\sum\limits_{i=1}^{m}\alpha_i\ket*{i,\vec{A_i}}=\ket*{\mathcal{M}\alpha}. 
\end{align}
The map $\mathcal{N}$ similarly appends to an arbitrary input state vector $\ket*{\alpha}$ a register that encodes the vector $\vec{A_F}$ whose entries are the $\ell_{2}$ norms $\|\vec{A_i}\|$ of the rows of $A$,
\begin{align}
\mathcal{N}: \ket*{\alpha}=\sum\limits_{j=1}^{n}\alpha_j\ket*{j}\rightarrow\sum\limits_{j=1}^{n}\alpha_j\ket*{\vec{A_F},j}=\ket*{\mathcal{N}\alpha}. 
\end{align}

The factorisation of $A$ follows from the amplitude encoding of $\vec{A_i}$ and $\vec{A}_F$. We have $\ket*{i,\vec{A_i}}=\frac{1}{\|\vec{A_i}\|}\sum\limits_{j=1}^{n}A_{ij}\ket*{i,j}$ and $\ket*{\vec{A}_F,j}=\frac{1}{\|A\|_F}\sum\limits_{i=1}^{m}\|\vec{A_i}\|\ket*{i,j}$, implying that $(\mathcal{M}^{\dagger}\mathcal{N})_{ij}=\braket*{i,\vec{A_i}}{\vec{A}_F,j}=\frac{A_{ij}}{\|A\|_F}$. Similarly, it follows that $\mathcal{M}$ and $\mathcal{N}$ have orthonormal columns and thus $\mathcal{M}^{\dagger}\mathcal{M}=I_m$ and $\mathcal{N}^{\dagger}\mathcal{N}=I_n$.

To show the relation between the eigenvalues of $W$ and the singular values of $A$, we consider the following:
\begin{align}
W\ket*{\mathcal{N}\vec{v}_i}=&(2\mathcal{M}\mathcal{M}^{\dagger}-I_{mn})(2\mathcal{N}\mathcal{N}^{\dagger}-I_{mn})\ket*{\mathcal{N}\vec{v}_i}\nonumber\\
=&(2\mathcal{M}\mathcal{M}^{\dagger}-I_{mn})\ket*{\mathcal{N}\vec{v}_i}\nonumber\\=&2\mathcal{M}\frac{A}{\|A\|_F}\ket*{\vec{v}_i}-\ket*{\mathcal{N}\vec{v}_i}\nonumber\\
=&\frac{2\sigma_i}{\|A\|_F}\ket*{\mathcal{M}\vec{u}_i}-\ket*{\mathcal{N}\vec{v}_i},
\end{align}
where we used the singular value decomposition $A=\sum_i\sigma_i\ket*{\vec{u}_i}\bra*{\vec{v}_i}$, and the fact that the right singular vectors $\{\vec{v}_i\}$ are mutually orthonormal. Note that $W$ rotates $\ket*{\mathcal{N}\vec{v}_i}$ in the plane of $\{\mathcal{M}\mathbf{u}_i,\mathcal{N}\mathbf{v}_i\}$ by $\theta_i$, such that
\begin{align}
\cos\theta_i&=\bra*{\mathcal{N}\vec{v}_i}W\ket*{\mathcal{N}\vec{v}_i} \notag \\
&=\frac{2\sigma_i}{\|A\|_F^2} \bra*{\vec v_i} A^{\dagger} \ket*{\vec u_i}-1 \notag \\
&=\frac{2\sigma_i^2}{\|A\|_F^2}-1 ,
\end{align} where we have used the fact that $(2\mathcal{M}\mathcal{M}^\dagger-I_{mn})$ represents a reflection in $\ket*{\mathcal{M}\vec{u}_i}$
and that $A^{\dagger} = \mathcal{N}^{\dagger}\mathcal{M}= \sum_i \sigma_i \ket*{\vec v_i} \bra*{\vec u_i}$.
Therefore the angle between $\ket*{\mathcal{N}\vec{v}_i}$ and $\ket*{\mathcal{M}\vec{u}_i}$ is given by $\frac{\theta_i}{2}$, i.e.\ half of the total rotation angle. Comparing the above expression with the half-angle formula for cosine functions, we obtain the relation $\cos\left(\frac{\theta_i}{2}\right)=\frac{\sigma_i}{\|A\|_F}$.

The two dimensional sub-space spanned by $\{\mathcal{M}\mathbf{u}_i,\mathcal{N}\mathbf{v}_i\}$ is therefore invariant under the action of $W$ which acts on it as a rotation by angle $\theta_i$ in the plane.


The runtime of QSVE is dominated by the phase estimation procedure which returns an $\delta$-close estimate of $\theta_i$, s.t.\ $|\tilde{\theta}_i -\theta_i| \leq 2 \delta$, which translates into the estimated singular value via $\tilde{\sigma}_i = \cos{(\tilde{\theta}_i/2)} \norm{A}_F$.
The error in $\sigma_i$ can then be bounded from above by $|\tilde{\sigma}_i - \sigma_i | \leq \delta \norm{A}_F$.
The unitary $W$ can be implemented in time $\Ord{\text{polylog}(mn)}$ by Lemma~\ref{lem:unitaries}, the running time for estimating of the singular values with additive error $\delta \norm{A}_F$
in $\Ord{\text{polylog}(mn)/\delta}$.

We can now extend this results to a quantum linear system algorithm. This work has been presented in \cite{wossnig2017quantum}.

\subsubsection{Extension to QLSA}
\label{subsec:qlsa_from_qsve}

Without loss of generality we can assume that the matrix $A$ is Hermitian. Otherwise we can perform the reduction that was previously discussed in section \ref{subsec:non-hermitian}. The QSVE algorithm immediately yields a linear system solver for positive definite matrices as the estimated singular values and eigenvalues are related via $\tilde{\sigma_i}=|\tilde{\lambda_i}|$. Let us consider this statement in further detail.
\begin{proposition}
	\label{prop:svd}
	If $A=A^{\dagger}$ is a Hermitian matrix, then the eigenvalues are equal to the singular values of $A$ up to a sign ambiguity and the corresponding
	eigenvectors equal the singular vectors of $A$, i.e.\ $\sigma_i = |\lambda_i|$ and $\mathbf{u}=\mathbf{v}=\mathbf{s}$. Therefore the singular value decomposition and the spectral decomposition are related by $A=\sum_{i}\lambda_i\mathbf{s}_i\mathbf{s}_i^{\dagger} = \sum_i \pm \sigma_i\mathbf{u}_i\mathbf{u}_i^{\dagger}$.
\end{proposition}

In order to solve general linear systems we need to recover the sign of each eigenvalue $\tilde{\lambda_i}$.
The solution is a simple algorithm that recovers the signs using the QSVE procedure as a black box incurring only a constant overhead with respect to the QSVE.

\begin{algorithm}[!hbb]
	\caption{Quantum linear system solver.}
	\label{alg:qlss}
	\begin{enumerate}
		\item Create the state $\ket*{\mathbf{b}} = \sum_i \beta_{i}  \ket*{\vec v_i}$ with $\vec v_i$ being the singular vectors of $A$.
		\item Perform two QSVEs as in Algorithm~\ref{alg:qsve} for matrices $A, A+ \mu I$ with $\delta \le 1/2\kappa $
		and $\mu=1/\kappa$ to obtain
		\begin{equation}
\sum_i \beta_{i}  \ket*{\vec v_i}_A \ket*{|\tilde{\lambda}_i|}_B \ket*{|\tilde{\lambda}_i + \mu|}_C.
\end{equation}

		\item Add an auxiliary register and set it to $1$ if the value in register $B$ is greater than that in register $C$ and apply a conditional
		phase gate:
		\begin{equation}
\sum_i (-1)^{f_i} \beta_{i}  \ket*{\vec v_i}_A \ket*{|\tilde{\lambda}_i|}_B \ket*{|\tilde{\lambda}_i + \mu|}_C \ket*{f_i}_D.
\end{equation}

		\item Add an ancilla register and apply a rotation conditioned on register $B$ with $\gamma=\Ord{1/\kappa}$.
		Then uncompute the registers $B,C, D$ to obtain
		\begin{align}
		\sum_i(-1)^{f_i} \beta_{i} \ket*{\vec v_i}\left( \frac{\gamma}{|\tilde{\lambda_i}|}\ket*{0} + \sqrt {1 - \left(\frac{\gamma}{|\tilde{\lambda_i}|}\right)^{2}}  \ket*{1} \right)
		\end{align}
		Postselect on the ancilla register being in state $\ket*{0}$.
	\end{enumerate}
\end{algorithm}

We assume that $A$ has been rescaled so that its well-conditioned eigenvalues lie within the interval $[-1, -1/\kappa] \cup [ 1/\kappa, 1]$, an assumption under which the other QLSP algorithms also operate.
\begin{theorem}
	\label{thm:qlsa}
	Let $A \in \mathbb R^{n \times n}$ be a Hermitian matrix with spectral decomposition
	$A = \sum_i \lambda_i \mathbf{\vec u_i \vec u_i}^{\dagger}$. Further let $\kappa$ be the condition number $A$, and $\norm{A}_F$ the Frobenius norm and $\epsilon > 0$ be a precision parameter. Then Algorithm~\ref{alg:qlss} has runtime $\Ord{\kappa^2 \cdot \text{polylog}(n) \cdot \norm{A}_F/\epsilon}$ that outputs the state
	$\ket*{\widetilde{A^{-1} \mathbf b}}$
	such that $\norm{ \ket*{ \widetilde{ A^{-1} \mathbf b } } - \ket*{ A^{-1} \mathbf b} }_2 \leq \epsilon $.
\end{theorem}

The error dependence on the Frobenius norm implies that we need $\norm{A}_F$ to be bounded by some constant or scale at least not with $\epsilon,\kappa, n$.
For the case that $\norm{A}_F$ is bounded by a constant, say $1$, the algorithm returns the output state with an $\epsilon$-error in polylogarithmic time even if the matrix is non-sparse.
We can relate this to the QLSA by assuming (as in the HHL) that the spectral norm $\norm{A}_*$ is bounded by a constant.
In general, the Frobenius norm scales with the dimensionality of the matrix, in which case $\norm{A}_F=\Ord{\sqrt{n}}$. Plugging everything back into the runtime of the algorithm we see that it runs in $\Ord{\kappa^2\sqrt{n} \cdot\text{polylog}(n)/\epsilon}$ and returns the output with a constant $\epsilon$-error for dense matrices with bounded spectral norm. 
Furthermore, note that this algorithm has a particular advantage if the rank of the matrix is low, as in this case we get then a square-root scaling with the rank instead of the Frobenius norm.

Let us compare this result to the improved HHL. Berry et al.~\cite{berry2009black} showed, that given black-box access to the matrix elements, Hamiltonian simulation with error $\delta_h$ can be performed in time $\Ord{n^{2/3}\cdot\text{polylog}(n)/\delta_h^{1/3}}$ in the case of dense matrices.
This leads to a linear system algorithm based on the HHL which scales as $\Ord{\kappa^2 n^{2/3} \cdot\text{polylog}(n)/\epsilon}$, where we have assumed the dominant error comes from phase estimation, and hence the error introduced by the Hamiltonian simulation is neglected. It was also shown numerically that the method of ~\cite{berry2009black} attains a typical scaling of $\Ord{\sqrt{n}\cdot\text{polylog}(n)/ \delta_h^{1/2}}$ when applied to randomly selected matrices, leading to a $\Ord{\kappa^2\sqrt{n} \cdot\text{polylog}(n)/\epsilon}$ linear system algorithm. The work of ~\cite{berry2009black} is based on the black-box model where one queries the quantum oracle with an index pair $\ket*{i,j}$ to obtain the matrix entry $\ket*{i,j,A_{ij}}$. The QSVE-based linear system solver achieves a $\Ord{\sqrt{n}}$-scaling in this stronger memory model, and it is an interesting open question if one can achieve a similar scaling in the black-box model. There are recent result for Hamiltonian simulation~\cite{low2016hamiltonian,low2017optimal,low2017hamiltonian} which might give better bounds if considered further.

\newpage
\addcontentsline{toc}{section}{References}
\printbibliography

\end{document}